%% file: SoS-NS.tex
\documentclass[11pt,letter,reqno]{article}
\usepackage{xifthen}
\usepackage{fullpage}
\usepackage{url}
\usepackage[vmargin=0.95in,hmargin=.95in]{geometry}
\usepackage{amssymb,latexsym}
\usepackage{graphicx,amsmath,amssymb}
\usepackage{amsthm}
\usepackage{color}
\usepackage{thmtools}
\usepackage{enumitem}
\newcommand{\me}{iddo}

\usepackage{float}
\usepackage{sidecap}
\usepackage{caption}
\usepackage{iddo}
\newcommand{\demph}[1]{\textbf{\textit{#1}}}

\usepackage{mdframed}

\usepackage{setspace}
\usepackage{mdframed}
\ifthenelse{\equal{\me}{iddo}}{
        \usepackage[usenames,pdf]{pstricks}
        \usepackage{epsfig}
}{}

\input{mfmacros}
\PackageWarning{miforbes}{LATER: is deferred}

         \newtheorem*{lemma*}{Lemma}
         \newtheorem{definition}{Definition}
         \newtheorem*{definition*}{Definition}

\input{bibmacros}
\usepackage[normalem]{ulem}

\newlength{\defbaselineskip}
\newcommand{\sos}{{\rm SoS}}
\newcommand{\sosSig}{\ensuremath{\mathsf{SoS}_{{\rm \Sigma}}}}

\renewcommand{\ps}{Positivstellensatz}

\definecolor{darkgreen}{rgb}{0,0.7,0}
\newcommand{\hirsch}[1]{\textcolor{darkgreen}{[Edward: #1]}}

\newcommand{\car}{{\rm CARRY}}
\newcommand{\cari}[1]{\ensuremath{\car_{#1}}}

\newcommand{\add}{{\rm ADD}}
\newcommand{\addi}[1]{\ensuremath{\add_{#1}}}
\newcommand{\addv}{\ensuremath{\overline \add}}
\newcommand{\prd}{{\rm PROD}}

\newcommand{\prdv}{\ensuremath{\overline \prd}}
\newcommand{\prdvp}{\ensuremath{\overline \prd_+}}
\newcommand{\iabs}{{\rm ABS}}
\newcommand{\iabsv}{\ensuremath{\overline \iabs}}
\newcommand{\val}{{\rm VAL}}

\newcommand{\arit}[1]{\ensuremath{{\rm arit\!}\left(#1\right)}}

\renewcommand{\vy}{\ensuremath{\overline y}}
\renewcommand{\vz}{\ensuremath{\overline z}}
\renewcommand{\ve}{\ensuremath{{\mathbf e}}}
\renewcommand{\vm}{\ensuremath{\overline m}}
\renewcommand{\vs}{\ensuremath{\overline s}}
\renewcommand{\vb}{\ensuremath{\overline b}}
\newcommand{\valpha}{{\ensuremath{\overline \alpha}}}
\newcommand{\ipsprf}[1]{\ensuremath{\vdash^{#1}_{\rm IPS}}}

\newcommand{\vone}{\ensuremath{\mathbf 1}}
\newcommand{\IPS}{\rm IPS}
\newcommand{\tconj}{\ensuremath{\tau}-conjecture}
\newcommand{\C}{\ensuremath{\mathbb{C}}}
\newcommand{\zbit}{{\rm BIT}}
\newcommand{\biti}[1]{\ensuremath{\zbit_{#1}}}
\newcommand{\bitv}{\ensuremath{\overline \zbit}}

\newcommand{\assmp}{\ensuremath{\overline{\mathcal F}}}
\newcommand{\ineqassmp}{\ensuremath{\overline{\mathcal H}}}
\newcommand{\cone}{\ensuremath{{\rm cone}}}
\renewcommand{\tr}[1]{\ensuremath{\left\llbracket #1 \right\rrbracket }}
\newcommand{\monomsize}[1]{\ensuremath{\left|#1\right|_{\mathrm {_{\#monomials}} }}}
\newcommand{\LSInfty}{\ensuremath{{\textup{LS}}^\infty_{*,+}}}

\newcommand{\ipsz}{\textup{IPS}$_\Z$}
\newcommand{\ba}{\ensuremath{\overline x^2-\overline x}}

\newcommand{\bvpn}{\ensuremath{\textup{BVP}_n}}
\newcommand{\bvpnm}{\ensuremath{\textup{BVP}_{n,M}}}
\newcommand{\bvpt}{\ensuremath{\textup{BVP}_t}}
\newcommand{\bvptm}{\ensuremath{\textup{BVP}_{t,M}}}
\newcommand{\ipsq}{\ensuremath{\textup{IPS}_\Q}}
\newcommand{\cpsq}{\ensuremath{\textup{CPS}_\Q}}
\newcommand{\cpsz}{\ensuremath{\textup{CPS}_\Z}}
\newcommand{\cpsqs}{\ensuremath{\cpsq^\star}}
\newcommand{\cpszs}{\ensuremath{\cpsz^\star}}
\newcommand{\ipsqs}{\ensuremath{\textup{IPS}^{\star}_\Q}}
\newcommand{\ipszs}{\ensuremath{\textup{IPS}^{\star}_\Z}}
\renewcommand{\iddo}[1]{}
\renewcommand{\hirsch}[1]{}
\renewcommand{\mar}[1]{}

\urldef\hurl\url{http://logic.pdmi.ras.ru/~hirsch}
\date{}
\author{Yaroslav Alekseev\thanks{Steklov Institute of Mathematics at St.~Petersburg, St.~Petersburg, Russia, and Chebyshev Laboratory at St.~Petersburg State University} \and Dima Grigoriev\thanks{CNRS, Mathematiques, Universite de Lille, Villeneuve d'Ascq, 59655, France. \url{http://en.wikipedia.org/wiki/Dima_Grigoriev}}\and Edward A. Hirsch\thanks{Steklov Institute of Mathematics at St.~Petersburg, St.~Petersburg, Russia. \hurl} \and  Iddo Tzameret\thanks{Department of Computer Science, Royal Holloway, University of London.  Iddo.Tzameret@rhul.ac.uk.
\url{http://www.cs.rhul.ac.uk/home/tzameret}}}

\usepackage{boxedminipage}
\usepackage{makeidx}
\makeindex

\begin{document}

\renewcommand{\ref}[1]{\autoref{#1}} 

\title{Semi-Algebraic Proofs, IPS Lower Bounds and the  $\tau$-Conjecture: Can a Natural Number be Negative?\thanks{A preliminary report on parts of this work was delivered at Dagstuhl Proof Complexity Workshop 2018. \url{https://materials.dagstuhl.de/files/18/18051/18051.IddoTzameret.Slides.pptx}} \thanks{The research presented in Section 3 is supported by Russian Science Foundation (project 16-11-10123).}}

\maketitle
\thispagestyle{empty}
\begin{abstract}

We introduce the \emph{binary value principle} which is a  simple subset-sum instance expressing that a natural number written in binary cannot be negative, relating it to central  problems in  proof and algebraic complexity.  %
We prove conditional superpolynomial lower bounds on the Ideal Proof System (IPS) refutation size of this instance, based on a well-known hypothesis by Shub and Smale about the hardness of computing factorials, where IPS is the strong algebraic  proof system introduced by Grochow and Pitassi \cite{GP14}. Conversely, we show that short IPS refutations of this instance bridge the gap between sufficiently  strong algebraic and semi-algebraic proof systems. Our results extend to full-fledged  IPS  the paradigm introduced in Forbes et al.~\cite{FSTW16}, whereby lower bounds against subsystems of IPS were obtained using restricted algebraic circuit lower bounds, and  demonstrate that the binary value principle captures
the advantage of semi-algebraic over algebraic reasoning, for sufficiently strong  systems. %
Specifically, we show the following:

\begin{description}[leftmargin=*]

\item[Conditional IPS lower bounds:]
The Shub-Smale hypothesis \cite{SS95} implies a superpolynomial lower bound on the size of IPS refutations of the binary value principle over the rationals defined as the unsatisfiable linear equation $\sum_{i=1}^{n} 2^{i-1}x_i = -1$, for boolean $x_i$'s. Further, the related \tconj\ \cite{SS95} implies a superpolynomial lower bound on the size of IPS refutations of a variant of the binary value principle over the ring of rational functions. %
No prior conditional lower bounds were known for IPS or  for apparently much weaker propositional proof systems such as Frege.\footnotemark
\mar{check about condit lower bounds)}


\item[Algebraic vs.~semi-algebraic proofs:]
Admitting short refutations of the binary value principle is  necessary for any algebraic proof system to fully simulate any known semi-algebraic proof system, and for strong enough algebraic proof systems it is also \emph{sufficient}. In particular, we introduce a very strong proof system that simulates all known semi-algebraic proof systems (and  most other known concrete propositional proof systems), under the name Cone Proof System (CPS), as a semi-algebraic analogue of the ideal proof system: CPS establishes the unsatisfiability of collections of polynomial equalities and inequalities over the reals, by representing \emph{sum-of-squares} proofs (and extensions) as algebraic circuits. We prove that  IPS is polynomially equivalent to  CPS iff IPS admits polynomial-size refutations of the binary value principle (for the language of systems of equations that have no 0/1-solutions), over both \Z\ and \Q.

\footnotetext{Though  simple, the binary value principle  is not a direct translation of a boolean formula, hence, similar to \cite{FSTW16} and other results on algebraic proofs (e.g., Razborov \cite{Razb98}), IPS lower bounds on this principle do not necessarily entail lower bounds for Frege or its subsystems.}

\end{description}
\end{abstract}


{\small
\hypersetup{linkcolor=black}
\tableofcontents
}
\thispagestyle{empty}
\normalsize

\clearpage

\pagenumbering{arabic}

\section{Introduction}\label{sec:intro}
This work connects three separate objects of study  in computational complexity: algebraic proof systems, semi-algebraic proof systems and algebraic circuit complexity. The connecting point is a  subset-sum instance  expressing that the value of a natural number given in binary is nonnegative.  \mar{\iddo{It is important to note that CR systems such as depth d PC cannot possibly simulate CPS of course or even IPS, because this would lead to PIT in P. So this rule out IMP19 techniques for strong systems.}
}
We will show that this instance  captures the advantage of semi-algebraic reasoning over algebraic reasoning in the regime of sufficiently strong proof systems, and is expected to be hard even for very strong algebraic proof systems. We begin with a general discussion about proof complexity, and then turn to  algebraic and  semi-algebraic proofs, their inter-relations, and the connection between circuit lower bounds and proof-size lower bounds. \iddo{maybe dont need this para on prf compl. Start with algeraic proof systesm}

Narrowly construed, proof complexity can be viewed as a stratification of the \NP\ vs.~\coNP\ question, whereby one aims to understand the complexity of stronger and stronger propositional proof systems as a gradual approach towards separating \NP\ from \coNP\ (and hence, also \P\ from \NP). This mirrors circuit complexity in which  different circuit classes are analyzed in the hope to provide general super-polynomial circuit lower bounds. Broadly understood however, proof complexity serves as a way to study  the computational resources required in different kind of reasoning, different algorithmic techniques and constraint solvers, as well as  providing propositional analogues to weak first-order theories of arithmetic.

Algebraic proof systems have attracted immense amount of work in proof complexity, due to their simple nature, being a way to study the complexity of computer-algebra procedures such as the Gr\"obner basis algorithm, and their connection to different fragments of  propositional proof systems  with counting gates. Beginning with the fairly  weak Nullstellensatz refutation system by Beame et al.~\cite{BeameIKPP96} and culminating in the very strong Ideal Proof System by Grochow and Pitassi \cite{GP14}, many algebraic proof systems and variants have been studied. In such systems one basically  operates with polynomial equations over a field using simple algebraic derivation rules such as additions of equations and multiplication of an equation by a variable, where variables are usually meant to range over \bits\ values.

Impagliazzo, Pudl\'ak and Sgall \cite{IPS99}, following Razborov \cite{Razb98}, showed that the polynomial calculus, which is the standard dynamic algebraic proof system introduced in \cite{CEI96}, requires exponential-size refutations (namely, those using an exponential number of monomials) for the simple symmetric unsatisfiable subset-sum instance $x_1+\dots+x_n=n+1$. Note that  refuting (that is, showing the unsatisfiability of) a linear equation $\sum_i\alpha_i x_i=\beta$ in which the variables $x_i$ are boolean, establishes that there is no subset of the $\alpha_i$ numbers that sums up to  $\beta$, and hence is considered to be a refutation of a subset-sum instance. Forbes et al.~\cite{FSTW16} showed that variants of this symmetric subset-sum instance are hard for different subsystems of the very strong  IPS algebraic proof system, that is, when IPS refutations are written using various restricted algebraic circuit classes. Loosely speaking, IPS is a static Nullstellensatz refutation in which proof-size is measured by algebraic circuit complexity instead of sparsity (that is, monomial size). In other words,  IPS proofs are written as algebraic circuits, and thus can tailor the advantage that algebraic circuits have over sparse polynomials (somewhat reminiscent to the way Extended Frege can tailor the full strength of boolean circuits in comparison to  resolution which operates merely with clauses).

The realm of semi-algebraic proof systems has emerged as an equally fruitful subject as algebraic proofs. Semi-algebraic proofs have been brought to  the attention of complexity theory from optimization \cite{LS91,Lov94}
by the works of Pudl\'ak \cite{Pud99} and Grigoriev and Vorobojov \cite{GV02} (cf.~\cite{GHP02}), and more recently, through their connection to approximation algorithms with the work of Barak et al.~\cite{BarakBHKSZ12} (see for example \cite{OZ13} and the new excellent survey by Fleming et al.~\cite{FKP19}).
%
While algebraic proofs derive polynomials in the ideal of a given initial set of polynomials, semi-algebraic proofs extend it to allow deriving polynomials also in the cone of the initial polynomials (informally a cone is an ``ideal that preserves positive signs'')\iddo{convex cone? What's the right term here?}, hence potentially utilizing a stronger kind of reasoning.\iddo{Is it even true?}\iddo{need here explanations on semi-algebraic proofs/semi-definite algorithm}~In particular \cite{BarakBHKSZ12}  considered the \emph{sum-of-squares} (SoS) refutation system. What makes SoS  important, for example to polynomial optimization, is the fact that the existence of a degree-$d$ SoS certificate can be formulated as the feasibility of a semidefinite program (SDP), and hence can be solved in polynomial time. \mar{I don't understand precisely this sentence: "This is the degree $d$ SoS relaxation first introduced by Shor [Sho87], and expanded upon by later works of Nesterov [Nes00], Grigoriev and Vorobjov [GV02], Lasserre [Las00, Las01] and Parrilo [Par00]. (see, e.g., [Lau09, BS14] for many more details)", From Raghavendra-Weitz. Any suggestion on what shall we write?}%

Berkholz \cite{Ber18} showed interestingly that in the regime of \emph{weak} proof systems, even static semi-algebraic proofs, such as \sos, can simulate dynamic algebraic proof systems such as polynomial calculus. Grigoriev \cite{Gri01-CC} showed that in this weak regime semi-algebraic proofs are in fact strictly stronger (with respect to degrees and size) than algebraic proofs, where the separating instances are simple polynomials (for example,  symmetric subset sum instances). However, it was not known in general (e.g., for strong systems) whether semi-algebraic reasoning is  strictly stronger than algebraic reasoning.



\iddo{More lower bounds for subset sum? non trivial upper bound in algebraic proofs?}

Another established tradition in proof complexity is to seek synergies between proofs and circuit lower bounds. In particular, \emph{proofs-to-circuits} transformations in the form of feasible interpolation, and other close concepts have been pivotal in the search for proof complexity lower bounds, as well as in circuit lower bounds themselves (see  G{\"{o}}{\"{o}}s et al.~\cite{GKRS19} for a recent example). In fact, the conception of IPS itself was motivated by the attempt to show that very strong proof complexity lower bounds would result in algebraic complexity class separations such as $\VP\neq\VNP$ (see \cite{GP14} and the survey \cite{PT16}). Li et al.~\cite{LTW18} as well as Forbes et al.~\cite{FSTW16} went in the other direction and showed that certain restricted algebraic circuit lower bounds imply size lower bounds on subsystems of IPS. In particular, \cite{FSTW16} devised a simple framework by which lower bounds on (subsystems of) IPS refutations are reduced to algebraic circuit lower bounds. \cite{FSTW16} used this framework to establish lower bounds on subsystems of IPS refutations of  variants of  symmetric subset-sum instances when the IPS refutations are written  as read once algebraic branching programs and multilinear formulas. But lower bounds on the size of full IPS refutations were not known.

\subsection{Background}\label{sec:background}

\para{Algebraic Circuits.}
Algebraic circuits over some fixed chosen field or ring $R$ compute formal polynomials in $R[x_1,\ldots,x_n] $ via addition and multiplication gates, starting from the input variables $\vx$ and constants from $R$. The size $|C|$ of an algebraic circuit $C$ is its number of nodes (see \ref{sec:algebraic_circuits} for formal definitions). 

Formally, an \emph{algebraic circuit} $C$ is a finite directed acyclic graph where edges are directed from leaves (in-degree 0 nodes) towards the output (out-degree 0 node). \emph{Input nodes}  are leaves  that are labeled with a variable from $x_1,\dots,x_n$; every other leaf  is labelled with a scalar  in $R$. All the other nodes have in-degree two  and  are labeled with either $+$ or $\times$. A leaf is said to compute the variable or scalar that labels itself. A $+$ (or $\times$) gate is said to compute the addition (product, resp.) of the polynomials computed by its incoming nodes. $C$ computes the polynomial computed by its output node.

%
%
%
%
%

%

The \demph{size} of an algebraic  circuit $C$ is the number of nodes in it denoted $|C|$, and the \emph{depth} of a circuit is the length of the longest directed path in it. 
Note that in the ``unit-cost'' model the size that a  ring constant  from $R$ contributes to the size of the circuit is  1 irrespective of the value of the constant. Sometimes it is important to consider the size of the coefficients and for this purpose we define a \demph{constant-free} algebraic circuit to be an algebraic circuit in which the only constants used
are $0, 1, -1$. Other constants must be built up using algebraic operations, which then count towards the size of the circuit. Constant-free algebraic circuit computes a polynomial over \Z, but when we allow for constant sub-circuits (and \emph{only} for constant sub-circuits) to contain division gates (in \ref{sec:lower-bounds}) we can also compute polynomials over \Q\ with constant-free circuits.

\para{The $\tau$-Conjecture and Shub-Smale Hypothesis.}
Here we explain  several important assumptions and conjectures that are known to lead to strong complexity lower bounds and complexity class separations, all of which will play a role in  our work (cf.~\cite{Sma98}).


\begin{definition*}[\autoref{def:tau-function}, $\tau$-function \cite{SS95}]
 Let $f\in\Z[\vx]$ be a multivariate polynomial over \Z.\mar{CHECK!} Then $\tau(f)$ is the minimal size of a constant-free  algebraic circuit that computes $f$ (that is, a circuit where the only possible constants that may appear on leaves are $1,0,-1$).
\end{definition*}

When we focus on constant polynomials, that is, numbers $n\in\Z$, $\tau(n)$ is the minimal-size circuit that can construct $n$ from $1$ using additions, subtractions
and multiplications (but not divisions; subtraction of a term $A$ can be constructed by $(-1)\cd A$).
We say that a family of (possibly constant) polynomials $(f_n)_{n\in\N}$ is \emph{\textbf{easy}} if $\tau(f_n)=\log ^{O(1)}n$, for every $n>2$, and \emph{\textbf{hard}} otherwise.
\mar{Change back to the notice about defining tau with some more constant, because we still do not have 2 in it}


A simple known upper bound on $\tau$ is this \cite{MS96}: for every integer $m>2$, $\tau(m)\le 2\log m$. 
 This is  shown by considering the binary expansion of $m$. For every integer $m$, the following lower bound is known $\tau(m)\ge \log \log m$ \cite{MS96}.
 It is not hard to show  that $(2^n)_{n\in\N}$ is \emph{easy}. For instance, if $n$ is a power of 2 then  $\tau(2^n)=\log n+3$, where $\log$ denotes the logarithm in the base 2. We start with 3 nodes to build $2=1+1$ and then by $\log n$ repeated squaring we arrive at $((2^2)^2)^2\dots)^2=2^{2^{\log n}}=2^n$. 
 On the other hand, it is known that $(2^{2^n})_{n\in\N}$ is \emph{hard}. \mar{What is the proof for this??}

While $(2^n)_{n\in\N}$ is easy and $(2^{2^n})_{n\in\N}$  is hard, it is not known whether $(n!)_{n\in\N}$ is easy or hard, and as seen below, showing the hardness of $\tau(m_n\cd n!)$, for every sequence $(m_n\cd n!)_{n\in\N}$ with $m_n\in\Z$ any nonzero integers, has very strong consequences.

It is known that for every integer $m>2$, $\log \log m\le \tau(m)\le 2\log m$ \cite{MS96}. It is not hard to show  that $(2^n)_{n\in\N}$ is \emph{easy} while  $(2^{2^n})_{n\in\N}$ is \emph{hard}. On the other hand, it is not known whether $(n!)_{n\in\N}$ is easy or hard, and as seen below, showing the hardness of $\tau(m_n\cd n!)$, for every sequence $(m_n\cd n!)_{n\in\N}$ with $m_n\in\Z$ some nonzero integers, has very strong consequences.



Let $\P_K$ and $\NP_K$ be the deterministic  and nondeterministic  versions of Turing machines in the Blum-Shub-Smale model \cite{BSS89}, respectively (see \ref{sec:tau-conj}). Further, let \VPZ\ be the ``Valiant'' class consisting of multivariate polynomial-families $(f_n(\vx))_{n\in\N}$ of polynomial degrees that can be computed by \emph{constant-free} (and division-free) polynomial-size circuits, and let \VNPZ\ be the constant-free (and division-free) ``Valiant NP class'' (see definitions in \ref{sec:Algebraic-Complexity-Classes}).
%
%
The following is a condition put forth by Shub and Smale \cite{SS95} (cf.~\cite{Sma98}) towards separating $\P_\C$ from $\NP_\C$, for \C\ the complex numbers:

\newtheorem*{Shub-Smale-hypothesis}{Shub-Smale Hypothesis}

\begin{Shub-Smale-hypothesis}[\cite{SS95,Sma98}]
For every nonzero integer sequence $(m_n)_{n\in\N}$, the sequence $(m_n\cd n!)_{n\in\N}$ is hard.
\end{Shub-Smale-hypothesis}

Shub and Smale, and  B\"urgisser, showed the following consequences of the Shub-Smale hypothesis:  
\begin{theorem*}[\autoref{thm:shub-smale}, \cite{SS95,Bur09}]
\begin{enumerate}
\item If the Shub-Smale hypothesis holds then $\P_\C\neq\NP_\C$.

\item If  the Shub-Smale Hypothesis holds then $\VPZ\neq\VNPZ$. In other words, Shub-Smale Hypothesis implies that the permanent  does not have polynomial size constant-free algebraic circuits.
\end{enumerate}
\end{theorem*}

It is open whether the Shub-Smale hypothesis holds. What is known is that if Shub-Smale hypothesis does \emph{not} hold then factoring of integers can be done in (nonuniform) polynomial time  (cf.~Blum
et al.~\cite[p.126]{BCR+98} and \cite{Che04}). \mar{remember to write this interesting corollary in the intro or Corollaries..}%
 Another related important assumption in algebraic complexity is the \emph{$\tau$-conjecture}. Let $f\in\Z[x]$ be a univariate polynomial with integer coefficients, denote by $z(f)$ the number of distinct integer roots of $f$.\mar{check fields here}

\newtheorem*{tauconj}{$\tau$-Conjecture}

\begin{tauconj}[\cite{SS95,Sma98}]
There is a universal constant $c$, such that for every univariate polynomial $f\in\Z[x]$:
$ (1+\tau(f))^c \ge z(f)\,.$
\end{tauconj}


The consequences of the \tconj\ are similar to the Shub-Smale Hypothesis:\vspace{-5pt}
\begin{theorem}[\autoref{thm:tau-conj-consequences}, \cite{SS95,Bur09}]
If the \tconj\ holds then both
$\P_\C\neq\NP_\C$ and $\VPZ\neq\VNPZ$ hold.
\end{theorem}

\para{Algebraic Proof Systems.}
A \emph{propositional proof system} (following \cite{CR79}) is a polynomial-time predicate $V(\pi,x)$
that verifies purported proofs $\pi$  (encoded naturally in binary) for propositional formulas $x$ (also encoded  in binary), such that $\exists\pi\ (V(\pi,x)=\textsf{true})$ iff $x$ is a tautology. 
In the setting of algebraic proof systems, one can use a broader definition of a proof system: instead of $V(\pi,x)$ being a polynomial-time predicate it is a \coRP\ predicate (since polynomial identity testing is in \coRP), and instead of providing proofs for propositional tautologies the system establishes proofs (in fact refutations) for sets of polynomial equations with no \bits\ solutions.

Grochow and Pitassi~\cite{GP14}, following \cite{Pit98}, suggested the following algebraic proof system which is essentially a Nullstellensatz proof system \cite{BeameIKPP96} written as an algebraic circuit (this was showed in \cite{FSTW16}). A proof in the  Ideal Proof System is given as  a \emph{single} polynomial. We provide below the \emph{boolean} version of  IPS (which includes the boolean axioms), namely the version that establishes the unsatisfiability over 0-1 of a set of polynomial equations.  (In what follows we follow the notation in \cite{FSTW16}):



\begin{definition*}[\autoref{def:IPS}, (boolean) Ideal Proof System (IPS),
Grochow-Pitassi~\cite{GP14}] Let $f_1(\vx),\ldots,f_m(\vx),p(\vx)$ be a collection of polynomials in $\F[x_1,\ldots,x_n]$ over the field \F. An \demph{IPS proof of $p(\vx)=0$ from $\{f_j(\vx)=0\}_{j=1}^m$}, showing that $p(\vx)=0$ is semantically  implied from the assumptions $\{f_j(\vx)=0\}_{j=1}^m$ over $0$-$1$ assignments, is an algebraic circuit $C(\vx,\vy,\vz)\in\F[\vx,y_1,\ldots,y_m,z_1,\ldots,z_n]$ such that (the equalities in what follows stand for  formal polynomial identities\footnote{That is, $C(\vx,\vnz,\vnz)$ computes the zero polynomial and $C(\vx,f_1(\vx),\ldots,f_m(\vx),x_1^2-x_1,\ldots,x_n^2-x_n)$ computes the polynomial $p(\vx)$.}):
\vspace{-6pt} 
\begin{enumerate}
   \item $C(\vx,\vnz,\vnz) = 0$; and\vspace{-5pt}
                \item $C(\vx,f_1(\vx),\ldots,f_m(\vx),x_1^2-x_1,\ldots,x_n^2-x_n)=p(\vx)$.
        \end{enumerate}
        
        \vspace{-6pt} 
        The \demph{size of the IPS proof} is the size of the circuit $C$. If $C$ is assumed to be constant-free, we refer to the size of the proof as the \demph{size of the constant-free IPS proof}.
%
The variables $\vy,\vz$ are  called the \emph{placeholder} \emph{variables} since they are used as placeholders for the axioms. An IPS proof  $C(\vx,\vy,\vz)$ of  $1=0$ from $\{f_j(\vx)=0\}_{j\in[m]}$ is called  an \demph{IPS refutation} of $\{f_j(\vx)=0\}_{j\in[m]}$ (note that in this case  it must hold that  $\{f_j(\vx)=0\}_{j=1}^m$ have no common solutions in $\bits^n$).
\end{definition*}

Notice that the definition above adds the equations $\{x_i^2-x_i=0\}_{i=1}^n$, called the set of \demph{boolean axioms} denoted $\vx^2-\vx$, to the system $\{f_j(\vx)=0\}_{j=1}^m$. This allows  to refute over $\bits^n$ unsatisfiable systems of equations.
Also, note that the first equality in the definition of IPS means that the polynomial computed by $C$ is in the ideal generated by $\overline y,\overline z$, which in turn, following the second equality, means that $C$ witnesses the fact that $1$ is in the ideal generated by $f_1(\vx),\ldots,f_m(\vx),x_1^2-x_1,\ldots,x_n^2-x_n$. (the existence of this witness, for unsatisfiable set of polynomials, stems from the Nullstellensatz theorem \cite{BeameIKPP96}).
In order to use IPS as a propositional proof system for refuting unsatisfiable CNF formulas we fix the usual encoding of clauses as algebraic circuits (\autoref{def:algebraic-transl-CNF}).

\mar{\iddo{Since I take out the propositional IPS terminology...needs to check if consistent}}




\para{Semi-Algebraic Proofs.}

The \textit{Positivstellensatz} proof system, as defined by Grigoriev and Vorobojov \cite{GV02}, is a static refutation system for establishing the unsatisfiability over the reals \R\ of a system consisting of both polynomial equations $\assmp=\{f_i(\vx)=0\}_{i\in I}$ and polynomial inequalities $\ineqassmp=\{h_j(\vx)\ge 0\}_{j\in J}$, respectively.
In \ps\ one essentially derives a polynomial in the \emph{cone} of the initial equalities and inequalities, in contrast to algebraic proofs in which one derives polynomial in the ideal of the initial polynomial equations. Loosely speaking, the cone serves as a non-negative closure of a set of polynomials, or in other words as a ``positive ideal" (see \ref{def:cone} and discussion in \ref{sec:semi-algebraic-proofs-prelim}).



We will distinguish between the \emph{real} \ps\ in which variables are meant to range over the reals and \emph{boolean} \ps\ in which variables range over \bits.
\begin{definition*}[\autoref{def:PS}, real Positivstellensatz proof system (real PS) \cite{GV02}]
Let $\assmp:=\{f_i(\vx)=0\}_{i\in I}$ be a set of polynomial equations and let  $\ineqassmp:=\{h_j(\vx)\ge 0\}_{j\in J}$ be a set of polynomial inequalities, where all polynomials are from $\R[x_1,\ldots,x_n]$. Assume that \assmp, \ineqassmp\ have no common real solutions.
A \emph{Positivstellensatz} \emph{refutation} of \assmp, \ineqassmp\ is a collection of polynomials $\{p_i\}_{i\in I}$ and $\{s_{i,\zeta}\}_{i,\zeta}$ (for $i\in\N$, $\zeta\subseteq J$ and $I_\zeta\subseteq\N$) in $\R[x_1,\ldots,x_n]$ such that the following formal polynomial identity holds:
\begin{equation}\label{eq:intro:ps}
\sum_{i\in I} p_i\cd f_i + \sum_{\zeta\subseteq J} \left(\prod_{j\in \zeta} h_j \cd \left(\sum_{i\in I_\zeta} s_{i,\zeta}^2\right)\right)=- 1\,.
\end{equation}%
 The \textbf{monomial size} of a \ps\ refutation is the combined total number of monomials in $\{p_i\}_{i\in I}$ and $\sum_{i\in I_\zeta} s_{i,\zeta}^2$, for all $\zeta\subseteq J$, that is, $\sum_{i\in I}\monomsize {p_i}+\sum_{\zeta\subseteq J}\monomsize {\sum_{i\in I_\zeta} s_{i,\zeta}^2}$.

\mar{\iddo{Actually it's unclear whether size was even considered there as a measure of sos/ps. Also I chose to count only number of monomials in polynomial coefficients.}%
\hirsch{A lower bound on PS size is proved in GHP02 for the knapsack, and in Itsykson-Kojevnikov for Tseitin formulas. The size of coefficients is considered (but does not help for lower bounds).
The notion of size in propositional proof complexity is always the number of bits in the proof (we are talking about NP vs co-NP!).
The question is only how the proof is represented. In particular, unit-cost notion is not
a p.p.s. notion (you can't represent integers to make them O(1) bit-size).}\iddo{Made changes in the same text in the Prelim section.}}
\end{definition*}

In order to use \ps\ as a refutation  system for  collections of equations \assmp\ and inequalities \ineqassmp\ that are
unsatisfiable over 0-1 assignments, we need to include simple  so-called boolean axioms.
This is done in slightly  different ways in different works  (see for example \cite{GHP02,AH19}).
One way to do this, which is the way we follow,
is the following:

\begin{definition*}[\autoref{def:prop-PS}, (boolean) \ps\ proof system (boolean PS)]
A \demph{boolean \ps\ proof} from a set of polynomial equations \assmp, and polynomial inequalities \ineqassmp, is an algebraic  \ps\ proof in which the following  \demph{boolean axioms} are part of the axioms: the polynomial equations $x^2_i-x_i=0$ (for all $i\in[n]$) are included in \assmp, and the polynomial inequalities $x_i\ge 0,~ 1-x_i\ge 0$ (for all $i\in[n]$) are included in \ineqassmp. \end{definition*}
In this way, \assmp, \ineqassmp\ have no common 0-1 solutions iff there exists a boolean \ps\ refutation of \assmp, \ineqassmp.
 %
%
Eventually, to define the boolean \ps\ as a propositional proof system for the unsatisfiable CNF formula we consider CNF formulas to be encoded as polynomial equalities according to  \autoref{def:algebraic-transl-CNF}. This version is sometimes called \demph{propositional Positivstellensatz}.
\emph{As a default when referring to \ps\ we mean the boolean \ps\ version}.

In recent years, starting mainly with the work of Barak, Brandao, Harrow, Kelner, Steurer and Zhou \cite{BarakBHKSZ12}, a special case of the \ps\ proof system has gained much interest due to its application in complexity and algorithms (cf.~\cite{OZ13}). This is the \demph{sum-of-squares} proof system (\textbf{SoS}), which is defined as follows:

\begin{definition*}[\autoref{def:sos-proof-system}, sum-of-squares proof system (\sos)]
A \demph{sum-of-squares proof} (SoS for short) is a \ps\ proof in which in \ref{eq:intro:ps} above  we restrict the index sets $\zeta\subseteq J$ to be \emph{singletons}, namely $|\zeta|=1$, hence, disallowing arbitrary products of inequalities within themselves. The real, boolean and propositional versions of \sos\ are defined similar to \ps.
\end{definition*}

\subsection{Our Results and Techniques}\label{sec:our-results}

 We consider the following subset-sum instance written as an unsatisfiable linear equation with large coefficients, expressing the fact that natural numbers written in binary cannot be negative:


\bigskip

\begin{boxedminipage}[c]{0.95\textwidth}
\begin{definition}[Binary Value Principle \bvpn]\label{def:knapsack}
The \emph{binary value principle} over the variables $x_1,\dots, x_n$, \bvpn\  for short, is the following unsatisfiable (over \bits\ assignments) linear equation:
\[
x_1+2x_2+4x_3+\dots+2^{n-1} x_n = -1\,.
\]
\end{definition}\end{boxedminipage}
\bigskip

At times we use   a more general principle denoted \bvpnm, which we call the \emph{generalized binary value principle}:
$x_1+2x_2+4x_3+\dots+2^{n-1} x_n = -M, \text{~~for a positive integer $M$}$.

\subsubsection{Lower Bounds}

We prove two kinds of conditional super-polynomial  lower bounds against IPS refutations.
The first is over \Q\ and \Z\ and the second is over the field  $\Q(y)$ of rational functions of univariate polynomials in the indeterminate  $y$denoted $\Q[y]$ (see \ref{def:Q[y]}).
We start with the first lower bound.


\begin{theorem*}[\autoref{thm:first-cond-lower-bound-Q}]
Under the Shub and Smale hypothesis, there are no $\poly(n)$-size constant-free (boolean) IPS refutations of the binary value principle \bvpn\ over \Q.
\end{theorem*}

This result can be viewed as pushing forward to full  IPS  the paradigm initiated by Forbes et al.~\cite{FSTW16} wherein proof complexity lower bound questions are reduced to algebraic circuit size lower bound questions: an IPS proof written as a circuit from a class $\mathcal C$ is obtained by showing that there are no small $\mathcal C$-circuits computing  certain polynomials. 

Therefore, we stress that \emph{this approach can only lead to conditional lower bounds for full unrestricted IPS}, as long as we do not have (explicit) super-polynomial lower bounds against  general algebraic circuits, namely as long as we do not prove $\VP\neq\VNP$.\footnote{Though, it should be mentioned that in proof complexity even non-explicit lower bounds are not  known, and will constitute a breakthrough in the field; hence moving from non-explicit (and thus \emph{known}) circuit lower bounds to (possibly also non-explicit) proof complexity lower bounds cannot be ruled out entirely.}


\iddo{Motivation. Why is it important? Take from FSTW16...}
\iddo{How the technique extends FSTW16?}


\medskip

\nind\emph{Proof sketch of }\autoref{thm:first-cond-lower-bound-Q}:  First, we show in \autoref{cor:from-Q-IPS-to-Z-IPS} that it is enough to consider IPS refutations over \Z\ instead of \Q. An IPS refutation over \Z\ is a proof of a nonzero integer $M$ instead of $-1$. \iddo{(***show here the IPS condition)}  Let $S_n:=\sum_{i=1}^n 2^{i-1}x_i$~ so that BVP is $S_n+1=0$, and assume that the IPS refutation of BVP is written as follows (this can be assumed without loss of generality by a result of \cite{FSTW16}):

\begin{equation}\label{eq:intro:BVP-IPS-ref}
Q(\vx)\cd (S_n + 1) + \sum_{i=1}^n H(\vx)\cd(x_i^2-x_i) = M\,.
\end{equation}   \vspace{-5pt}

\nind Since the IPS refutation is over \Z\ we know in particular that $Q(\vx)$ is an integer polynomial. Let us consider now only \bits\ assignments to \autoref{eq:intro:BVP-IPS-ref}. Since under \bits\ assignments the boolean axioms $x_i^2-x_i$ vanish we get from \autoref{eq:intro:BVP-IPS-ref}:
\vspace{-9pt}

\begin{equation}\label{eq:intro:E2}
Q(\vx)\cd (S_n + 1) =M\,.
\end{equation}

\nind Observe that the image of $S_n+1$ under boolean assignments is the set of  all possible natural numbers between $1$ to $2^n$. In other words, for every number $b\in[2^n]$, there exists an assignment $\valpha\in\bits^n$, such that $(S_n+1)(\valpha)=b$. Since $Q(\vx)$ is an integer polynomial, it evaluates to an integer under every \bits\ assignment.  Therefore, by \autoref{eq:intro:E2} $M$ is a product of every natural number between $1$ to $2^n$.
This already brings us close to the conditional lower bound: we assume contra-positively that there is a polynomial-size constant-free circuit that computes $Q(\vx)$, which  implies that there exists a polynomial-size constant-free and variable-free circuit that computes $M$ (because fixing any boolean assignment to the variables we get such a circuit over over \Z\ computing $M$). We then show that if there exists a $\poly(n)$-size constant-free circuit for $M\in\Z$, such that $M$ is divisible by every number in $[2^n]$, then there exists a $\poly(n)$-size circuit that computes $(2^n)!$.

Consider the $\poly(n)$-size circuit for $M^{2^n}$ that is obtained by $n$ repeated squaring of $M$. Since $M$ is divided by every natural number in $[2^n]$ it is in particular divisible by every \emph{prime} number in $[2^n]$. It is possible to show that the power of every prime number in the prime factorisation of $(2^n)!$ is at most $2^n$, from which we can conclude that $M^{2^n}$ is an integer product of $(2^n)!$. We thus obtain a constant-free  $\poly(n)$-size circuit for a nonzero integer product of $(2^n)!$. From this it is easy to show that for every $m$ with $2^{n-1}\le m\le 2^n$ there is a $\poly(n)$-size constant-free circuit computing a nonzero integer product of $m!$, hence that sequence $(c_m\cd m!)_{m=1}^\infty$ admits a $\log^{O(1)}m$-size family of constant-free circuits, in contrast to the Shub-Smale hypothesis. \qed
\vspace{-5pt}


\para{Rational field lower bounds.}
We now consider IPS operating over the field of rational functions in the (new) indeterminate $y$, denoted $\Q(y)$ (\ref{def:Q[y]}). This  allows us to formulate a very interesting  version of the binary value principle. Roughly speaking, this version expresses the fact that the BVP is ``almost always'' unsatisfiable.  More precisely, consider the linear  equation  $\sum_{i=1}^n a_ix_i = y$, for integer $a_i$'s, and $y$ the new indeterminate. This equation is unsatisfiable for most $y$'s, when $y$ is substituted by an  element from \Q. 
We show that once we have an IPS refutation over $\Q(y)$ of this equation  we can substitute  $y$ by any constant but a finite number of rational numbers and get a valid IPS refutation over \Q\ of the original BVP. Thus \emph{an
IPS refutation over $\Q(y)$ of $\sum_{i=1}^n a_ix_i = y$ can be viewed as a single refutation  for all but finitely many values of $y\in\Q$.}

We show that while for polynomially bounded coefficients $a_i$ there are small $\Q(y)$-IPS refutations of 
 $\sum_{i=1}^n a_ix_i = y$, for  $\sum_{i=1}^n 2^{i-1}x_i = y$, there are no small refutations, assuming the \tconj: 



\begin{theorem*}[\ref{thm:second-lower bound}]\iddo{maybe make it less formal here...}
Suppose a system of polynomial equations $F_0(\vec{x}) = F_1(\vec{x}) = F_2(\vec{x}) =  \cdots = F_n(\vec{x})= 0$, $F_i \in \mathbb{Q}(y)[x_1, \ldots, x_n]$, where $F_0(\vec x) = y + \sum_{i = 1}^{i = n} 2^{i - 1} x_i$ and $F_i(\vec x) = x_i^2 - x_i$, has an IPS-LIN$_{\mathbb{Q}(y)}$ certificate $H_0(\vec{x}), \ldots, H_{n}(\vec x)$, where each $H_i(\vec x)$ can be computed by a $\mathbb{Q}(y)[x_1, \ldots, x_n]$-algebraic circuit of size $\poly(n)$ \iddo{Is it constant free circuit?}. Then, the $\tau$-conjecture is false.
\end{theorem*}

\hirsch{re is a draft for the introduction. It assumes that the simulation of CPS in IPS is mentioned before lower bounds. If it is not, we'll have to make a workaround by promising it first.}\iddo{Okay, I see. Indeed, it comes after. So just mention it comes after.}




Roughly speaking, the lower bound proof extracts denominators from the refutation and obtains a small circuit that has all $n$-bit nonnegative integers as its roots and thus cannot exist under the $\tau$-conjecture.

\subsubsection{Algebraic versus Semi-Algebraic Proofs}\label{sec:Algebraic-versus-Semi-Algebraic-Proofs}

We  exhibit the importance of the binary value principle by showing that it captures in a manner made precise the strength of semi-algebraic reasoning in the regime of strong (to very strong) proof systems, and formally those systems that can efficiently reason about bit arithmetic.
Note that already Frege system can reason about bit arithmetic
(see \cite{Goe90} following \cite{Bus87}); however, this alone is not sufficient to simulate semi-algebraic systems.
Specifically, we show that short refutations of the binary value principle would bridge the gap between very strong algebraic reasoning captured by
the ideal proof system and its semi-algebraic analogue
that we introduce in this work, which we call the Cone Proof System (CPS for short).


Whereas IPS is devised to capture derivations in the \emph{ideal} of initial given polynomials, CPS is defined so  to exhibit derivations in the \emph{cone} (\ref{def:cone})  of these polynomials. The cone proof system  establishes that a collection of polynomial equations  $\assmp:=\{f_i=0\}_i$ and polynomial inequalities $\ineqassmp:=\{h_i\ge 0\}_i$ are unsatisfiable over 0-1 assignments (or over real-valued assignments, when desired). In the spirit of  IPS \cite{GP14} we define  a refutation in CPS as a \emph{single} algebraic circuit. This circuit computes a polynomial that results from positive-preserving operations such as addition and product applied between the inequalities \ineqassmp\ and themselves, as well as the use of nonnegative scalars and arbitrary squared polynomials. In order to simulate in CPS the free use of equations from \assmp\ we incorporate in the set of inequalities \ineqassmp\ the inequalities $f_i\ge 0$ and $-f_i\ge 0$ for each $f_i=0$ in \assmp\ (we show that this enables one to add freely products of the polynomial $f_i$ in CPS proofs, namely working in the ideal of \assmp\ (in addition to working in the cone of \ineqassmp); see \ref{sec:Boolean-CPS-Simulates-Boolean-IPS}).

We first formalize the concept of a cone as an algebraic circuit. Let $C$ be a circuit and $v$ be a node in $C$. We call $v$ a \demph{squaring gate} if $v$ is a product gate of which two incoming edges are emanating from the \emph{same} node $z$, that is $v=z^2$.
%


\begin{definition*}[\autoref{def:conic-circuit-for-vx}, \vy-conic circuit] Let $R$ be an ordered ring. We say that an algebraic circuit $C$ computing a polynomial over $R[\vx,\vy]$ is a \demph{conic circuit with respect to \vy}, or \demph{$\vy$-conic}  for short,  if for every negative constant or a variable $x_i\in\vx$, that appears as a leaf $u$ in $C$, the following holds: every path $p$ from $u$ to the output gate of $C$ contains a squaring gate.
\end{definition*}

Informally, a \vy-conic circuit is a circuit in which we assume that the \vy-variables are nonnegative, and any other input that may be negative (that is, a negative constant or an \vx-variable) must be part of a squared sub-circuit.
CPS is defined roughly in the same way as   IPS only that instead of circuits we use conic circuits:

\begin{definition*}[\autoref{def:cps}, (boolean)\ Cone Proof System (CPS)]
Consider a collection of polynomial equations $\assmp:=\{f_i(\vx)=0\}_{i=1}^m$, and a collection of polynomial inequalities $\ineqassmp:=\{h_i(\vx)\ge  0\}_{i=1}^\ell$, where all polynomials are from $\R[x_1,\ldots,x_n]$. Assume that the following \demph{boolean axioms} are included in the assumptions: \assmp\ includes $x_i^2-x_i=0$, and  $\ineqassmp$ includes the inequalities $x_i\ge0$ and $1-x_i\ge 0$, for every variable $x_i\in\vx$. Suppose further that $\ineqassmp$ includes (among possibly other inequalities) the two inequalities $f_i(\vx)\ge0$ and $-f_i(\vx)\ge0$ for every equation $f_i(\vx)=0$ in $\assmp$ (including the equations $x_i^2-x_i=0$).
A \demph{CPS proof of $p(\vx)$ from $ \assmp$ and $\ineqassmp$}, showing    that $\assmp,\ineqassmp$ semantically imply the polynomial inequality $p(\vx)\ge 0$
over $0$-$1$ assignments, is an algebraic circuit $C(\vx,\vy)$ computing a polynomial in $\R[\vx,y_1,\dots,y_\ell]$, such that:\footnote{Note that formally we do not make use of the assumptions \assmp\ in CPS, as we assume always that the inequalities that correspond to the equalities in \assmp\ are present in \ineqassmp. Thus, the indication of \assmp\ is done merely to maintain clarity and distinguish (semantically) between two kinds of assumptions: equalities and inequalities.}\vspace{-5pt} \begin{enumerate}
 \item $C(\vx,\vy)$ is a \vy-conic circuit;
and\vspace{-5pt}
\item $C(\vx,\ineqassmp)=p(\vx)$,
\end{enumerate}\vspace{-5pt}
where equality 2 above is a formal polynomial identity (that is, $C(\vx,\ineqassmp)$ computes the polynomial $p(\vx)$) in which the left hand side means that we substitute $h_i(\vx)$ for $y_i$, for all $i=0,\dots,\ell$.
The \demph{size} of a CPS proof is the size of the circuit $C$. The variables $\vy$ are the \emph{placeholder} \emph{variables} since they are used as a placeholder for the axioms. A CPS proof of $-1$ from $\assmp,\ineqassmp$ is called a \demph{CPS refutation  of \assmp, \ineqassmp}.
\end{definition*}
\mar{ Chec if needed: In what follows, we will write ``conic'' instead of ``\vy-conic'' where the meaning of \vy\ is clear from the context.}

To refute CNF formulas  in CPS we use
the algebraic translation of CNFs (\autoref{def:algebraic-transl-CNF}) into a set of polynomial equalities (we can equally express CNFs as  inequalities; see \autoref{prop:CPS-from-equtional-CNF-to-inequalities-CNF}). 
The real version of CPS, called \demph{real CPS}, is defined  similar to CPS only without the boolean axioms.

\mar{\iddo{Make sure you reference the propositional version of SoS. As otherwise, the general algebraic version of SoS does not have the $x_i\ge 0 $ assumptions in the inequalities; hence, it seems slight weaker or different than CPS}
\hirsch{SoS is not referenced right here. The simulation does mention boolean (not propositional) version.}}
\begin{remark}
Formally, CPS proves only consequences from an initial set of inequalities \ineqassmp\ and not equalities \assmp. However, we are not losing any power doing this. First, observe that
an assignment satisfies  \assmp, \ineqassmp\ iff it satisfies \ineqassmp\ (in the case of boolean CPS an assignment that satisfies either $\assmp$ or $\ineqassmp$ must be a \bits\ assignment).
Second, we encode \emph{equalities} $f_i(\vx)=0\in\assmp$ using the two inequalities $f_i(\vx)\ge 0$ and $-f_i(\vx)\ge 0$ in \ineqassmp. As shown in \autoref{thm:CPS-sim-IPS}  \emph{this way we can derive any polynomial in the \emph{ideal} of \assmp, and not merely in the cone of \assmp,} as is required for equations (and similar to the definition of SoS), with at most a polynomial increase in size (when compared to IPS).

\end{remark}
In contrast to IPS where a short refutation for \bvpn\ would imply strong computational consequences, the binary value principle is trivially refutable in CPS (as well as  in \sos):

\begin{proposition*}[\autoref{prop:CPS-proof-of-BVP}]
CPS admits a linear size refutation of the binary value principle \bvpn.
\end{proposition*}

We show that IPS and  CPS simulate each other if there  exist small IPS refutations of  the binary value principle. This provides a characterisation of semi-algebraic reasoning in terms of the binary value principle. In what follows, \ipszs\ and \cpszs\ stand for \emph{boolean} versions of IPS and CPS, where both are proof systems for refuting unsatisfiable sets of polynomial equalities (not necessarily CNFs) and where  the `$^\star$' superscript means  that  possible values that are computed along the IPS or CPS proofs (as  circuits) are not super-exponential (when the input variables range over \bits), namely, that the bit-size of these values are polynomial in the proof size (see \autoref{sec:Algebraic-versus-Semi-Algebraic-Proof-Systems}).



\mar{Explain that BVP is then a good candidate to separate semia and algebraic proof systems in the strong(est) proof systems known..."We identify a simple instance that captures the hardness of semi-algebraic proofs over algebraic ones in the very strong regiem of IPS, etc..." ...\\
}

\begin{corollary*}[\autoref{cor:IPS-cond-equiv-CPS}, BVP characterizes the strength of boolean CPS]
\begin{enumerate}[leftmargin=*]

\item 
Constant-free \ipszs\ is polynomially equivalent to constant-free \cpszs\ iff constant-free \ipszs\ admits $\poly(t)$-size refutations of \bvpt.\vspace{-6pt}

\item 
Constant-free \ipsqs\ is polynomially equivalent to  constant-free \cpsqs\ iff for every positive integer $M$ constant-free \ipsqs\ admits $\poly(t,\tau(M))$-size refutations of \bvptm.
\end{enumerate}
\end{corollary*}

\begin{proof}[Proof idea for part 1]
%
%

($\Leftarrow$) To show that \ipszs\ simulates \cpszs\ assuming short refutations of \bvpt\ we proceed as follows: let $C(\vx,\assmp)=-1$ be the \cpszs\ refutation of \assmp. Then, as a polynomial identity  $C(\vx,\assmp)=-1$ is basically freely provable in \ipszs. We now use  the ability of IPS to do efficient bit arithmetic, that we demonstrate formally as follows. Define $\val(\vw)=w_1+2w_2+\dots 2^{n-2}w_{n-1}-2^{n-1}w_n$ to be  the value of an integer number given by the $n$ boolean bits $\vw$ in the two's complement scheme (where $w_n$ is the sign bit). 
Our \emph{main novel technical contribution  here} is the following  result connecting  the value of a polynomial to its bit vector expressed as a function of the variables: \vspace{-2pt} 
\begin{lemma*}[\ref{lem:main-binary-value-lemma}; informal] For any circuit $f$, IPS has a $\poly(|f|)$-size proof of\vspace{-5pt} 
\begin{equation}\label{eq:intro:main-binary-value-lemma}
\val\left(\biti 1(f)\cdots \biti n(f)\right)=f,
\end{equation}
where $\biti i (f)$ is the polynomial that computes the $i$th bit of the number computed by $f$ as a function of the variables $\vx$ to $f$ that range over \bits\ values. 
\end{lemma*}

Denote $C(\vx,\assmp)$ by $C$ for short. By \ref{eq:intro:main-binary-value-lemma} we have $C=\val\left(\biti 1(C)\cdots \biti n(C)\right)=-1 $. Since $C$ is a conic circuit and thus preserves positive signs we can prove that the sign bit $\biti n(C)=0$. We are thus left with the need to refute that the value of a positive number written in binary $\biti 1(C)\cdots \biti {n-1}(C)$ is non-negative, which is  efficiently provable in  \ipszs\ by assumption.
\end{proof}

\vspace{-10pt}

\para{The relative strength of proof systems.} \autoref{fig:ps-diagram} provides an illustrative picture of the relative strength of algebraic and semi-algebraic proof systems, which gives  context to our results. Note that CPS is among the strongest concrete  proof systems for boolean tautologies to be formalized to date:  it simulates IPS (\ref{thm:CPS-sim-IPS}) which is already very strong. Like IPS it can prove freely polynomial identities (\ref{fact:zero-poly-ips-proof}), and so it ``subsumes'' in this sense such identities (accordingly, CPS proofs needs the full power of \coRP\ to be verified). It is unclear whether even ZFC (as a proof system for propositional logic) can simulate CPS (it is not hard to show that this would imply that polynomial identity testing is in \NP).
Indeed, we are unaware of any concrete propositional proof system (even those that are merely  \coRP-verifiable) that can simulate CPS.


Grigoriev \cite{Gri01-CC} showed that algebraic proofs like PC cannot simulate semi-algebraic proofs like \sos\ because symmetric subset-sum instances such as  $x_1+\dots+x_n=-1$ require linear degrees (and exponential monomial size) \cite{IPS99}\iddo{anything other paper?}, and Forbes et al.~\cite{FSTW16} extended these lower bounds on symmetric subset-sum instances to  stronger algebraic proof systems, namely to subsystems of IPS. Our work (\ref{thm:first-cond-lower-bound-Q}) extends this gap  further, showing that even the strongest algebraic proof system known to date IPS cannot fully simulate even a weak proof system like \sos, assuming Shub-Smale hypothesis. \iddo{It's unclear where to put that is statement, that follows from simple linear upper bound on BVP in SoS. In previous subsection?}


Exponential size lower bounds for semi-algebraic proof systems are known since \cite{GHP02}, and such bounds for propositional versions of static Lovasz-Schrijver and constant degree Positivstellensatz systems were  proved in \cite{IK06}. Beame, Pitassi and Segerlind \cite{BPS07} started the study of lower bounds for
semantic threshold systems, that include in particular
tree-like Lov\'asz-Schrijver systems. This line of research culminated in \cite{GP18}, where strong lower bounds were  proved using critical block sensitivity, a notion introduced in \cite{HN12}.

\begin{SCfigure}[1]
\begin{centering}
\captionsetup{width=1.\linewidth}
\includegraphics[width=0.5\textwidth,height=0.42\textheight]{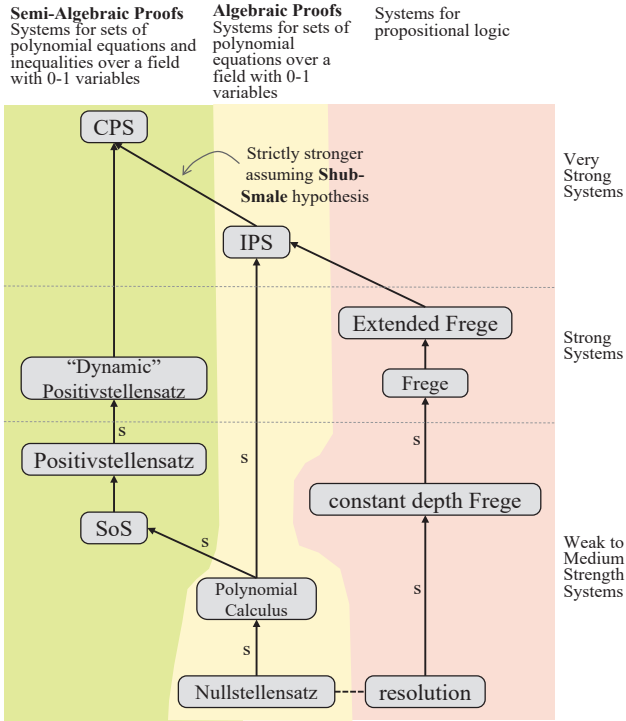}
\end{centering}
\caption{\small Relative strength of propositional proof systems (partial). An arrow $Q\to P$ means that $P$ simulates $Q$. While $Q\stackrel{s}{\rightarrow}P$ means ``strictly stronger'', i.e., $P$ simulates $Q$ but $Q$ does not simulate $P$. Dashed line $Q - - - P $ means that $Q$ and $P$ are incomparable: $P$ cannot simulate $Q$ and $Q$ cannot simulate $P$. %
The three colored-shaded vertical blocks indicate proof systems for languages of increasing expressiveness (from right to left): systems for propositional logic, for polynomial equations with 0/1 variables  (including encodings of propositional logic) and both polynomial equations and inequalities with 0/1 variables. The \emph{informal} qualifications of strength mean roughly  the following:\ weak systems are those we know super-polynomial lower bounds against, and their strength and limitations are quite well understood via feasible interpolation results and random CNFs lower bounds. Medium strength systems are those with some known lower bounds, but their strength is less well understood; e.g., feasible interpolation is not known for them. \iddo{sos is weak then!!!} Strong systems are those with no known lower bounds. Very strong proof systems are those strong systems whose verification is done in \coRP, and they can prove freely any polynomial identity
 (written as an algebraic circuit).}\label{fig:ps-diagram}
\end{SCfigure}
 \mar{There is a linear-size \sos\ refutation of \bvpn\ (see XXX\hirsch{?}) }

\iddo{Write: The main binary value lemma may be interesting by itself. \\ Lower bounds for the BVP in PC and C-IPS; for large coefficients: ResLin.}

\subsection{Conclusions}

 This work demonstrates that a simple subset-sum  principle, written as a linear equation,
captures, in the boolean case (i.e., when variables range over \bits), the apparent advantage of semi-algebraic proofs over algebraic proofs in the following sense:
it is necessary for any boolean algebraic proof system that simulates full boolean semi-algebraic proofs to admit short refutations of the principle;
and if the algebraic proof system is  strong enough to be able to efficiently perform basic bit arithmetic,
this condition is also \emph{sufficient} to achieve such a simulation. To formalize these results we introduce a very strong proof system CPS that derives polynomials in the cone of initial axioms instead of in the ideal.

We show that  CPS is expected to be stronger than even the very strong algebraic  Ideal Proof System  (IPS) formulated by Grochow and Pitassi in \cite{GP14}, since  our subset-sum principle is hard for IPS assuming the hardness of computing factorials \cite{SS95}. 
We establish a related lower bound on IPS refutation-size based on the \tconj\ \cite{SS95}. These lower bounds  push forward  the paradigm introduced  by Forbes et al.~\cite{FSTW16}: whereas \cite{FSTW16} showed how to obtain restricted IPS lower bounds for certain subset-sum instances, based  on known lower bounds against restricted circuit classes, we show how to obtain \emph{general} IPS lower bounds based on specific hardness assumptions from algebraic complexity.\footnote{Note again that extending the paradigm in \cite{FSTW16} to IPS operating with general circuits must result in conditional lower bounds, as long as explicit  super-polynomial algebraic circuit lower bounds are  not known.}


\subsection{Relation to other Work}\label{sec:relation-to-other-work}


\para{Bit arithmetic and semi-algebraic proofs.} In \ref{sec:Reasoning-about-Bits-within-Algebraic-Proofs} we show how to reason about the bits of polynomial expressions within algebraic proofs. Bit arithmetic in proof complexity was used before in Frege systems (see \cite{Goe90} following \cite{Bus87}). Independently of our work \cite{HT18-dagstuhl}, Impagliazzo et al.~\cite{IMP19} considered the possibility to \emph{effectively} simulate \emph{weak}   semi-algebraic proofs using medium-strength algebraic proofs. They have  considered expressing and reasoning with the bits of algebraic expressions, as we do in \ref{sec:Reasoning-about-Bits-within-Algebraic-Proofs}. However, their proof methods and results are fundamentally different from ours:
first, they work in the weak proof systems regime, while we work in the strong systems regime. I.e., they aim to effectively simulate \emph{weak}  proof systems like constant degree sum-of-squares (in which polynomials are written as sum of monomials), while we aim to simulate \emph{very strong} proof systems such as CPS (essentially, \ps\ written as algebraic circuits). Second, they use a different way to express bits in their work. This is done in order to be able to reason about bits with bounded-depth algebraic circuits, while we do not need this mechanism. Third, they  show only \emph{effective }simulation and not simulation (namely, before the algebraic proofs can simulate a system of polynomial equations  or inequalities, the equations and inequalities need to be pre-processed, that is, translated ``off-line`` to their bit-vector representation). Fourth, they do not consider the \val\ function nor the binary value principle, while our work shows that essentially this is a necessary ingredient in a full simulation of strong semi-algebraic proof systems. In fact, we have the following:

Assuming the Shub-Smale hypothesis, our results \emph{rule out the possibility that 
even a very strong algebraic proof system such as IPS simulates (in contrast to the weaker notion of an effective simulation) even a weak semi-algebraic proof system like  constant degree \sos\ measured by monomial size}. In other words, assuming Shub-Smale hypothesis, we rule-out the possibility that the arguments  in \cite{IMP19} (or any other argument) can yield a simulation of constant degree \sos\ by algebraic proofs operating with constant depth algebraic circuits (depth-$d$ PC in \cite{IMP19}\footnote{Here we use the fact that IPS simulates depth-$d$ PC.}). 
 It remains however open whether depth-$d$ PC simulates constant degree \sos\ \emph{for the language of unsatisfiable CNF formulas} or for unsatisfiable sets of linear equations with small coefficients.

\para{Subset-sum lower bounds in  proofs complexity.} Different instances of the subset sum problem have been considered as hard instances for algebraic proof systems. For example, Impagliazzo et al.~\cite{IPS99} provided an exponential size lower bound on the symmetric subset sum instance $x_1+\dots+x_n=n+1$, for boolean $x_i$'s in the polynomial calculus refutation system.
Grigoriev \cite{Gri01-CC} proved that the version $\sum_{i=1}^n x_i=r$ for a non-integer $r\approx\frac{n}2$ requires linear degrees to refute in \ps,
and \cite{GHP02} later transformed this idea into an exponential-size lower bound for both Positivstellensatz and static high-degree Lovasz-Schrijver proof systems.
Moreover, as already mentioned, our lower bounds can be seen as an extension to the case of general IPS refutations of the approach introduced by Forbes et al.~\cite{FSTW16}. 


The work of Part and Tzameret \cite{PT18} established unconditional exponential lower bounds  on the size of resolution over linear equations refutations of the binary value principle, over any sufficiently large field \F, denoted Res(lin$_\F$).
The proof techniques in \cite{PT18} are completely different from the current work, but these results demonstrate  that using instances with large coefficients in proof complexity provides new insight into the complexity of proof systems.

 \iddo{Can we show that every subset-sum has a VNP polynomial proof? I\ think yes. Hence we get a new proof of shub-smale $==>$ VP vs VNP !!}


\iddo{SoS: Berkholtz?}

\iddo{Note: proving IPS lower bounds on BVP (binary Value Principle) will separate P from VP over C in the sense of Blum-Shub-Smale! Not only Valiant's mode.}

\mar{\textit{From Raghavendra-Weitz}: The main appeal of SoS certificates for polynomial optimization is that the existence of a degree $d$ SoS certificate can be formulated as the feasibility of a semidefinite program (SDP). This is the degree d SoS relaxation first introduced by Shor [Sho87], and expanded upon by later works of Nesterov [Nes00], Grigoriev and Vorobjov [GV02], Lasserre [Las00, Las01] and Parrilo [Par00]. (see, e.g., [Lau09, BS14] for many more details).}

\section{Preliminaries}\label{sec:prelims}

\iddo{Remember to add the new observations: if we dont use the VAL function (ie, use the Impagliazzo et al. technique we end up with a proof that PIT is in NP (or even P). }
\hirsch{I do not know about the connection to PIT...}
\iddo{I will write about it....maybe}

\subsection{Notation}\label{sec:notation}
For a natural number we let $[n]=\set{1,\ldots,n}$.
Let $R$ be a ring. Denote by $ R[x_1,\ldots,x_n] $ the ring of 
multivariate polynomials with coefficients from $ R $ and  variables $ x_1,\ldots,x_n $. We usually denote by $\vx$ the vector of variables $ x_1,\ldots,x_n $.  We treat polynomials as \emph{formal }linear combination of monomials, where a monomial is a product of variables. Hence, when we talk about the \textit{zero polynomial} we mean the polynomial in which the coefficients of all monomials are zero.  Similarly, two polynomials are said to be \emph{identical} if their monomials have the same coefficients. The \emph{number of monomials} in a polynomial $f$ is the number of monomials with nonzero coefficients denoted $\monomsize f$. The \emph{degree} of a multivariate polynomial (or total degree) is the maximal sum of variable powers in a monomial with a nonzero coefficient in the polynomial.
We write $\poly(n)$ to denote a polynomial growth in $n$, namely a function that is upper bounded by $cn^{c}$, for some constant $c$ independent of $n$. Similarly, $\poly(n_1,\dots,n_s)$ for some constant $s$, means a polynomial growth that is at most  $kn_1^{c_1} \cdots n_s^{c_s}$, 
for $k$ and  $c_{ji}$'s that are constants independent of $n_1,\dots,n_s$.

For $S$ a set of polynomials from $R[x_1,\ldots,x_n]$, we denote by $\langle S \rangle $ the \emph{ideal generated by $S$}, namely the minimal set containing $S$ such that
if $f,g\in\langle  S\rangle $ then also $\alpha f + \beta g \in \langle  S\rangle$, for any $\alpha,\beta\in R$.

\subsection{Algebraic Circuits}\label{sec:algebraic_circuits}
Algebraic circuits  over some fixed chosen field or ring $R$ compute polynomials in $R[x_1,\ldots,x_n] $ via addition and multiplication gates, starting from the input variables $\vx$ and constants from the field.
More precisely, an \emph{algebraic circuit} $C$ is a finite directed acyclic graph
where edges are directed from leaves (that is,  in-degree 0 nodes) towards the output nodes (that is out-degree 0 nodes).
By default, there is a single output node.
\emph{Input nodes}  are in-degree 0 nodes that are labeled with a variable from $x_1,\dots,x_n$;
every other in-degree zero node is labelled with a scalar element in $R$.
All the other nodes have in-degree two (unless otherwise stated) and  are labeled with either $+$ or $\times$.
An in-degree 0 node is said to \emph{compute}  %
%
the variable or scalar that   labels  itself. A $+$ (or $\times$) gate is said to compute the addition (product, resp.) of the polynomials computed by its incoming nodes.
The \demph{size} of an algebraic  circuit $C$ is the number of nodes in it denoted $|C|$,
and the \emph{depth} of a circuit is the length of the longest directed path in it. Note that the size of a field coefficient in this setting
\mar{\iddo{Do we use this setting?}\hirsch{I think no, I suggest to remove the 1-weight model.} \iddo{Update: we do use this setting in the simulation part. I specifically made statements about both constant-free and non constant free circuits. Also, this measure does apply anyhow to constant-free circuits, so we need to say this.}\hirsch{Should we give it an explicit name for further referral; for example, ``unit-cost'' similarly to as they speak about random-access machines.}\iddo{Yes unit-cost sounds good}} is  1 irrespective of the value of the coefficient.
Sometimes it is important to consider the size of the coefficients appearing in the circuit (for instance, when we are concerned with the computational complexity of problems pertaining to algebraic circuits we need to have an efficient way to represent the circuits as bit strings). For this purpose it is standard to define a \demph{constant-free} algebraic circuit to be an algebraic circuit in which the only constants used
are $0, 1, -1$. Other constants must be built up using algebraic operations, which
then count towards the size of the circuit.

An algebraic circuit is said to be a \emph{multi-output} circuit if it has more than one output node, namely, more than one node of out-degree zero.
Given a single-output algebraic circuit $F(\vx)$ we denote by $\widehat F(\vx)\in R[\vx]$ the \emph{polynomial} computed by $F(\vx)$.
We define the \emph{degree} of a circuit $C$ (similarly a node) as the total degree of the polynomial $\hat C$  computed by $C$, denoted $\deg(C)$.




We will also use circuits that have division gates; when we need them, we  define them explicitly.

\para{Algebraic Complexity Classes.}\label{sec:Algebraic-Complexity-Classes}
We  now recall some  basic notions from algebraic
complexity (for more details see \cite[Sec.~1.2]{SY10}).
Over a ring $R$, $\cc{VP}_{R}$ (for
``Valiant's \P'') is the class of families $f=(f_n)_{n=1}^{\infty}$ of formal polynomials $f_n$ such that $f_n$ has $\poly(n)$ input variables, is of $\poly(n)$ degree, and can be computed by algebraic circuits over $R$ of $\poly(n)$ size. $\cc{VNP}_{R}$ (for ``Valiant's
\NP'') is the class of families $g$ of polynomials $g_n$ such that $g_n$ has $\poly(n)$ input variables and is of $\poly(n)$ degree, and can be written as
\[
g_n\left(x_1,\dotsc,x_{\poly(n)}\right) = \sum_{\vec{e} \in \{0,1\}^{\poly(n)}} f_n(\vec{e}, \vec{x})
\]
for some family $(f_n) \in \cc{VP}_{R}$.
A major question in algebraic complexity theory is whether the permanent polynomial can be computed by algebraic circuits of polynomial size. Since the permanent is complete for $\VNP$ (under a suitable concept of algebraic reductions that are called p-projections), showing that no polynomial-size circuit can compute the permanent amounts to showing \VP$\neq$\VNP\ (cf.~\cite{Val79:ComplClass,Val79-permanent,Val82}).

Similarly, we can consider the \emph{constant-free} versions of \VP\ and \VNP: we denote by \VPZ\ and \VNPZ\ the class of polynomial-size and polynomial-degree \emph{constant-free} algebraic circuits and the class of \VNP\ polynomials as above in which the family of polynomials $(f_n) \in \VPZ$. In these definitions of \VPZ\ and \VNPZ\ we assume also that no division gate occur in the circuits, hence \VPZ\ and \VNPZ\ compute polynomials over \Z. We shall also consider in \ref{sec:lower-bounds} constant-free circuits \emph{over \Q}: these will be constant-free circuits in which constant sub-circuits (and \emph{only} constant sub-circuits) may contain division gates.


\subsection{The $\tau$-Conjecture and Shub-Smale Hypothesis}\label{sec:tau-conj} 
Here we explain  several important assumptions and conjectures that are known to lead to strong complexity lower bounds and complexity class separations, all of which are  relevant to our work. See for example Smale's list of ``mathematical problems for the next century" \cite{Sma98} for a short description and discussion about these problems.
\begin{definition}[$\tau$-function \cite{SS95}]\label{def:tau-function}
 Let $f\in\Z[\vx]$ be a multivariate polynomial over \Z. \mar{CHECK!} Then $\tau(f)$ is the minimal size of a constant-free  algebraic circuit that computes $f$ (that is, a circuit where the only possible constants that may appear on leaves are $1,0,-1$).
\end{definition}

When we focus on constant polynomials, that is, numbers $n\in\Z$, $\tau(n)$ is the minimal-size circuit that can construct $n$ from $1$ using additions, subtractions
and multiplications (but not divisions; note that subtraction of a term $A$ can be constructed by $-1\cd A$).

We say that a family of (possibly constant) polynomials $(f_n)_{n\in\N}$ is \emph{\textbf{easy}} if $\tau(f_n)=\log ^{O(1)}n$, for every $n>2$, and \emph{\textbf{hard}} otherwise.\footnote{We put the condition  $n>2$ instead of $n\ge 1$, because unlike \cite{SS95} we do not add the constant 2 to the constants available in the circuit, hence to keep the same known upper bounds of $\tau$ we skip the cases $n=1,2$.}

The following are some known facts about $\tau(\cd)$:
\begin{itemize}
\item $(2^n)_{n\in\N}$ is \emph{easy}. For instance, if $n$ is a power of 2 then  $\tau(2^n)=\log n+3$, where $\log$ denotes the logarithm in the base 2. We start with 3 nodes to build $2=1+1$ and then by $\log n$ repeated squaring we arrive at $((2^2)^2)^2\dots)^2=2^{2^{\log n}}=2^n$. 

\item $(2^{2^n})_{n\in\N}$ is \emph{hard}. \iddo{How to show this? Ref?}\hirsch{added:}\iddo{Actually, I'm not sure how it is proved precisely...I\ get the idea though... } This is clear from the straightforward upper bound on the largest integer that can be computed with $k$ multiplication/addition/subtraction gates. 

\item
A simple known upper bound on $\tau$ is this \cite{MS96}: for every integer $m>2$, $\tau(m)\le 2\log m$. 
This is  shown by considering the binary expansion of $m$.
\item For every integer $m$, the following lower bound is known $\tau(m)\ge \log \log m$ \cite{MS96}.
\end{itemize}

While $(2^n)_{n\in\N}$ is easy and $(2^{2^n})_{n\in\N}$  is hard, it is not known whether $(n!)_{n\in\N}$ is easy or hard, and as seen below, showing the hardness of $\tau(m_n\cd n!)$, for every sequence $(m_n\cd n!)_{n\in\N}$ with $m_n\in\Z$ any nonzero integers, has very strong consequences.



Blum, Shub and Smale \cite{BSS89} introduced an algebraic version of  Turing machines that has access to a field $K$ (Poizat observed that their model can be defined as  algebraic circuits in which \emph{selection} gates $s(z,x,y)$ can be used; where a selection gate outputs $x$ in case $z=0$ and $y$ in case $z=1$). In this model one can formalise and study a variant of the \P\ vs.~\NP\ problem for languages solvable by polynomial-time machines with access to  $K$, denoted $\P_K$, versus nondeterministic polynomial-time machines with access to $K$, denoted $\NP_K$. 

The following is a condition put forth by Shub and Smale \cite{SS95} (cf.~\cite{Sma98}) towards separating $\P_\C$ from $\NP_\C$, for \C\ the complex numbers:

\begin{Shub-Smale-hypothesis}[\cite{SS95,Sma98}]
For every nonzero integer sequence $(m_n)_{n\in\N}$, the sequence $(m_n\cd n!)_{n\in\N}$ is hard.
\end{Shub-Smale-hypothesis}
Shub and Smale, as well as B\"urgisser, showed the following consequences of the Shub-Smale hypothesis:
\begin{theorem}[\cite{SS95,Bur09}]\label{thm:shub-smale}
\begin{enumerate}
\item If the Shub-Smale hypothesis holds then $\P_\C\neq\NP_\C$.

\item If  the Shub-Smale Hypothesis holds then $\VPZ\neq\VNPZ$. In other words, Shub-Smale Hypothesis implies that the permanent  does not have polynomial size constant-free algebraic circuits over \Z.
\end{enumerate}
\end{theorem}

It is open whether the Shub-Smale hypothesis holds. What is known is that if Shub-Smale hypothesis does \emph{not} hold then factoring of integers can be done in (nonuniform) polynomial time  (cf.~Blum
et al.~\cite[p.126]{BCR+98} and \cite{Che04}). \iddo{remember to write this interesting corollary in the intro or Corollaries..}
\bigskip

Another related important assumption in algebraic complexity is the \emph{$\tau$-conjecture}. Let $f\in\Z[x]$ be a univariate polynomial with integer coefficients, denote by $z(f)$ the number of distinct integer roots of $f$.


\begin{tauconj}[\cite{SS95,Sma98}]
There is a universal constant $c$, such that for every univariate polynomial $f\in\Z[x]$:
$$ (1+\tau(f))^c \ge z(f)\,.$$
\end{tauconj}


The consequences of the \tconj\ are similar to the Shub-Smale Hypothesis:
\begin{theorem}[\cite{SS95,Bur09}]\label{thm:tau-conj-consequences}
If the \tconj\ holds then both
$\P_\C\neq\NP_\C$ and $\VPZ\neq\VNPZ$ hold.
\end{theorem}

\subsection{Basic Proof Complexity} In the standard setting  of propositional proof complexity, a \emph{propositional proof system} \cite{CR79} is  a polynomial-time predicate $V(\pi,x)$ that verifies purported proofs $\pi$  (encoded naturally in, say, binary)
for propositional formulas $x$ (also encoded naturally in binary),
such that $\exists\pi\ (V(\pi,x)=\textsf{true})$ iff $x$ is a tautology.\footnote{Historically, Cook and Reckhow \cite{CR79} defined a propositional proof systems as a  polynomial-time computable surjective mapping of bit strings (encoding purported proofs) \emph{onto} the set of propositional tautologies (encoded as bit-strings as well). This is equivalent to the definition of propositional proof systems we presented, up to polynomial factors.}
Hence, a propositional proof system is a complete and sound proof system for propositional logic in which a proof can be checked for correctness in polynomial time (though, note that a proof $\pi$ may be exponentially larger than the tautology $x$ it proves).


When considering algebraic proof systems that operate with algebraic circuits, such as IPS,
it is common to relax the notion of a propositional proof system,
so to require that the relation $V(\pi,x)$ is in probabilistic polynomial time,
instead of deterministic polynomial time
(since polynomial identities can be verified in \coRP,
while not known to be verified in \P).

Furthermore, the language that a given proof system proves, namely the set of instances that the proof system proves to be tautological, or always satisfied, can be different from the set of propositional tautologies. First, we can consider a propositional proof system to be a \emph{refutation system} in which a proof establishes that the initial set of axioms (e.g., clauses) is \emph{unsatisfiable}, instead of always satisfied (i.e., tautological). For most cases, considering a propositional proof system to be a refutation system preserves all properties of the proof system, and thus the notions of refutation and proofs are used as synonyms. Second, we can define a proof system to be  complete and sound for languages different or larger than unsatisfiable propositional formulas. For instance, in algebraic proof systems we usually consider   proof systems that are sound and complete for the language  of  unsatisfiable sets of polynomial equations.

For the purpose of comparing the relative complexity of different proof systems we have the concept of  \emph{simulation}: given two proof systems $P,Q$ for the \emph{same} language,
we say that $P$ \demph{simulates} $Q$
if there is a function $f$ that maps $Q$-proofs to $P$-proofs of the same instances
with at most a polynomial blow-up in size.
If $f$ can be computed in polynomial time,
this is called a \demph{p-simulation}.
If $P$ and $Q$ simulate each other we say that $P$ and $Q$ are \emph{polynomially-equivalent}.
If   $P$ and $Q$ are two proof systems for \emph{different} languages, prima facie we cannot compare their strength via the notion of simulation. However, if both $P$ and $Q$ prove (or refute) propositional instances like formulas in conjunctive normal form, or boolean tautologies, while encoding them in different ways (namely, they use different representations for essentially the same propositional formulas), we can fix a polynomial-time computable translation from one representation to the other. Under this translation we can  consider $P$ and $Q$ to be proof systems for the \emph{same} language,  allowing us to use the notion of simulation between $P$ and $Q$.
\iddo{Then what is the difference between this an Effective Simulation by Pitassi Santhanam, actually?}
\hirsch{Difference 1: $p$-simulation is for a specific pre-agreed translation of the language of $P$ into the language of $Q$,
effective simulation may choose any transformation. \iddo{still, what's the difference? the p-time transformation  then can be said to be "the agreed one", no?\hirsch{Put another way: p-simulation translates proofs (for the same Boolean formula) only; effective simulation translates formulas as well. We write that we can agree on some translation and then consider the existence (or nonexistence) of p-simulation \textbf{for this particular translation}. P-S quantifies over such translations (in particular it is hard to imagine how to prove the nonexistence of effective simulation).}}Difference 2: effective simulation is parameterized. This is how I understand Sect.2 of Pitassi-Santhanam.\iddo{Okay... in this case the IMP19 is also a ``simulation".}}

\subsection{Algebraic Proofs}

Grochow and Pitassi~\cite{GP14}  suggested the following algebraic proof system  which is essentially a Nullstellensatz proof system (\cite{BeameIKPP96}) written as an algebraic circuit. A proof in the  Ideal Proof System is given as  a \emph{single} polynomial. We provide below the \emph{boolean} version of  IPS (which includes the boolean axioms), namely the version that establishes the unsatisfiability over 0-1 of a set of polynomial equations.  In what follows we follow the notation in \cite{FSTW16}:



\begin{definition}[(boolean) Ideal Proof System (IPS),
Grochow-Pitassi~\cite{GP14}]\label{def:IPS} Let $f_1(\vx),\ldots,f_m(\vx),p(\vx)$ be a collection of polynomials in $\F[x_1,\ldots,x_n]$ over the field \F. An \demph{IPS proof of $p(\vx)=0$ from $\{f_j(\vx)=0\}_{j=1}^m$}, showing that $p(\vx)=0$ is semantically  implied from the assumptions $\{f_j(\vx)=0\}_{j=1}^m$ over $0$-$1$ assignments, is an algebraic circuit $C(\vx,\vy,\vz)\in\F[\vx,y_1,\ldots,y_m,z_1,\ldots,z_n]$ such that (the equalities in what follows stand for  formal polynomial identities\footnote{That is, $C(\vx,\vnz,\vnz)$ computes the zero polynomial and $C(\vx,f_1(\vx),\ldots,f_m(\vx),x_1^2-x_1,\ldots,x_n^2-x_n)$ computes the polynomial $p(\vx)$.}):
        \begin{enumerate}
                \item $C(\vx,\vnz,\vnz) = 0$; and\vspace{-5pt}
                \item $C(\vx,f_1(\vx),\ldots,f_m(\vx),x_1^2-x_1,\ldots,x_n^2-x_n)=p(\vx)$.
        \end{enumerate}
        The \demph{size of the IPS proof} is the size of the circuit $C$. If $C$ is assumed to be constant-free, we refer to the size of the proof as the \demph{size of the constant-free IPS proof}.
%
The variables $\vy,\vz$ are  called the \emph{placeholder} \emph{variables} since they are used as placeholders for the axioms. An IPS proof  $C(\vx,\vy,\vz)$ of  $1=0$ from $\{f_j(\vx)=0\}_{j\in[m]}$ is called  an \demph{IPS refutation} of $\{f_j(\vx)=0\}_{j\in[m]}$ (note that in this case  it must hold that  $\{f_j(\vx)=0\}_{j=1}^m$ have no common solutions in $\bits^n$).
\end{definition}

Notice that the definition above adds the  equations $\{x_i^2-x_i=0\}_{i=1}^n$, called the set of \demph{boolean axioms} denoted  $\vx^2-\vx$, to the system $\{f_j(\vx)=0\}_{j=1}^m$. This allows  to refute over $\bits^n$ unsatisfiable systems of equations.
Also, note that the first equality in the definition of IPS means that the polynomial computed by $C$ is in the ideal generated by $\overline y,\overline z$, which in turn, following the second equality, means that $C$ witnesses the fact that $1$ is in the ideal generated by $f_1(\vx),\ldots,f_m(\vx),x_1^2-x_1,\ldots,x_n^2-x_n$ (the existence of this witness, for unsatisfiable set of polynomials, stems from the Nullstellensatz theorem \cite{BeameIKPP96}).

In order to use IPS as a propositional proof system (namely, a proof system for propositional tautologies), we need to
fix the  encoding of  clauses as algebraic circuits.
\begin{definition}[algebraic translation of CNF formulas]\label{def:algebraic-transl-CNF}
Given a CNF formula in the variables $\vx$,  every clause $\bigvee_{i\in P} x_i \lor \bigvee_{j\in N} \neg{x_j}$ is translated into  $\prod_{i\in P} (1-x_i)\cdot \prod_{j\in N} x_j=0$. (Note that these terms are    written as algebraic circuits as displayed, where products are not multiplied out.)
\end{definition}
Notice that in this way a 0-1 assignment to a CNF is satisfying iff the assignment is satisfying all the equations in the algebraic translation of the CNF.


Therefore, using \autoref{def:algebraic-transl-CNF} to encode CNF formulas, boolean IPS is considered as a propositional proof system for the language of unsatisfiable CNF formulas, sometimes called \demph{propositional IPS}. We say that an IPS proof is an \demph{algebraic IPS} proof, if we do not use the boolean axioms \ba\ in the proof. \emph{As a default when referring to IPS we mean the boolean IPS version}.




\subsubsection{Conventions and Notations for IPS Proofs}\label{sec:IPS-conventions} An IPS proof over a specific field or ring is sometimes denoted IPS$_\F$ noting it is over \F. For two algebraic circuits $F,G$, we define the \textit{size of the equation $F=G$} to be the total circuit size of $F$ and $G$, namely, $|F|+|G|$.

Let $\overline {\mathcal F}$ denote a set of polynomial equations $\{f_i(\vx)=0\}_{i=1}^m$, and let $C(\vx,\vy,\vz)\in\F[\vx,\vy,\vz]$ be an IPS proof of $f(\vx)$ from $\overline {\mathcal F}$ as in \autoref{def:IPS}. Then we write $C(\vx,\overline {\mathcal F},\ba)$ to denote the circuit $C$ in which $y_i$ is substituted by $f_i(\vx)$ and $z_i$ is substituted by the boolean axiom $x_i^2-x_i$. By a slight abuse of notation we also call  $C(\vx,\overline {\mathcal F},\ba)=f(\vx)$ an IPS proof of $f(\vx)$ from $\overline {\mathcal F}$ and $\ba$ (that is, displaying $C(\vx,\vy,\vz)$ \emph{after} the substitution of the placeholder variables $\vy,\vz$ by the axioms in $\overline {\mathcal F}$ and $\ba$, respectively).

For two polynomials $f(\vx),g(\vx)$, \emph{an IPS proof of $f(\vx)=g(\vx)$} from the assumptions $\overline {\mathcal F}$ is an IPS proof of $f(\vx)-g(\vx)=0$ (note that in case $f(\vx)$ and $g(\vx)$ are identical as polynomials this is trivial to prove; see \autoref{fact:zero-poly-ips-proof}).

We denote by $C: \overline {\mathcal F}\ipsprf s p=0$ (resp.~$C: \overline {\mathcal F}\ipsprf s p=g$) the fact that $p=0$ (resp.~$p=g$) has an IPS proof $C(\vx,\vy,\vz)$ of size $s$ from assumptions $\overline {\mathcal F}$. We  may also suppress   ``$=0$" and write simply $C:\overline {\mathcal F}\ipsprf s p$ for $C:\overline {\mathcal F}\ipsprf s p=0$. Whenever
we are only interested in claiming the existence of an IPS proof of size $s$ of $p=0$ from $\overline {\mathcal F}$ we suppress the $C$ from the notation. Similarly, we can suppress the size parameter $s$ from the notation.  If $F$ is a circuit computing a polynomial $\hat F\in\F[\vx]$, then we can talk about \emph{an IPS proof $C$ of $F$ from assumptions $\overline {\mathcal F}$}, in symbols $C:\overline {\mathcal F}\ipsprf {} F$, meaning an IPS proof of $\hat F$. Accordingly, for two circuits $F,F'$ such that $\hat F=\hat F'$, we may speak about an \emph{an IPS proof $C$ of $F$ from assumptions $\overline {\mathcal F}$} to refer to an IPS proof of $F'$ from assumptions \assmp.

\subsection{Semi-Algebraic Proofs}\label{sec:semi-algebraic-proofs-prelim}

The \emph{Positivstellensatz} proof system, as defined by Grigoriev and Vorobojov \cite{GV02}, is a refutation system for establishing the unsatisfiability over the reals \R\ of a system consisting of both polynomial equations $\assmp=\{f_i(\vx)=0\}_{i\in I}$ and polynomial inequalities $\ineqassmp=\{h_j(\vx)\ge 0\}_{j\in J}$, respectively.
It is based on a restricted version of the Positivstellensatz theorem \cite{Kri64,Ste74}. In order to formulate it, we need to define the notion of a cone, as in \cite{GV02}, which serves as a non-negative closure of a set of polynomials, or informally the notion of a ``positive ideal". Usually the cone  is defined as the set closed under non-negative linear combinations of polynomials (cf.~\cite{BPT13-SDP-book}), but following \cite{GV02} we are going to use a more general formulation which is sometimes called \emph{the sos cone}.


\begin{definition}[cone]\label{def:cone}
Let $\ineqassmp\subseteq R[\vx]$ be a set of polynomials over an ordered ring $R$.
Then \emph{the cone of \ineqassmp}, denoted $\cone(\ineqassmp)$,
is defined to be the smallest set $S\subseteq R[\vx]$ such that:\vspace{-4pt}
\begin{enumerate}
\item
$\ineqassmp \subseteq S$;\vspace{-5pt}

\item for any
polynomial $s\in R[\vx]$, $s^2\in S$; \vspace{-5pt}
\item\label{it:const} for any positive constant $c>0$, $c\in S$;\vspace{-5pt}
\item if $f,g\in S$, then both $f+g \in S$ and $f\cd g\in S$.
\end{enumerate}
\end{definition}

Note that we have formulated the cone for any ordered ring (item~\autoref{it:const} would be superfluous for reals). This is because we are going to use this notion in the context of $\Z$ and $\Q$
(although the Positivstellensatz theorem does not hold for these rings,
it is still possible to use Positivstellensatz refutations in the presence of  the boolean axioms, namely as a refutation system for instances unsatisfiable over 0-1 value).

Note also that every sum of squares (that is, every sum of squared polynomials $\sum_i s_i^2$) is contained in every cone. Specifically, $\cone(\emptyset)$ contains every  sum of squares.

Similar to the way
\mar{\iddo{It's unclear whether the algebraic NS proof system is by Beame, Impagliazzo, Kraj\'{i}\v{c}ek, Pitassi and Pudl{\'a}k---it is?}\hirsch{At least they claim so in the paper about the system for propositional logic; I guess noone talked about polynomials as ``proofs'' before.}}
the Nullstellensatz proof system \cite{BeameIKPP96} establishes the  unsatisfiability of sets of polynomial equations based on the Hilbert's Nullstellensatz theorem \cite{Hil78-translated-work} from algebraic geometry, the \ps\ proof system\ is based on the  Positivstellensatz theorem from semi-algebraic geometry:

\begin{theorem}[Positivstellensatz theorem \cite{Kri64,Ste74}, restricted version]
\mar{No module part ($g_i\neq 0$), thus restricted}
Let $\assmp:=\{f_i(\vx)=0\}_{i\in I}$ be a set of polynomial equations and let  $\ineqassmp:=\{h_j(\vx)\ge 0\}_{j\in J}$ be a set of polynomial inequalities, where all polynomials are from $\R[x_1,\ldots,x_n]$. There exists  a pair of polynomials $f\in \langle \{f_i(\vx)\}_{i\in I} \rangle $ and $h\in \cone(\{h_j(\vx)\}_{j\in J})$ such that $f+h=-1$  if and only if there is no assignment that satisfies both $\assmp$ and $\ineqassmp$.
\end{theorem}


The \ps\ proof system is now natural to define. We shall distinguish between the \emph{real} \ps\ in which variables are meant to range over the reals and \emph{boolean} \ps\ in which variables range over \bits.
\begin{definition}[real Positivstellensatz proof system (real PS) \cite{GV02}]\label{def:PS}
Let $\assmp:=\{f_i(\vx)=0\}_{i\in I}$ be a set of polynomial equations and let  $\ineqassmp:=\{h_j(\vx)\ge 0\}_{j\in J}$ be a set of polynomial inequalities, where all polynomials are from $\R[x_1,\ldots,x_n]$. Assume that \assmp, \ineqassmp\ have no common real solutions.
A \emph{Positivstellensatz} \emph{refutation} of \assmp, \ineqassmp\ is a collection of polynomials $\{p_i\}_{i\in I}$ and $\{s_{i,\zeta}\}_{i,\zeta}$ (for $i\in\N$, $\zeta\subseteq J$ and $I_\zeta\subseteq\N$) in $\R[x_1,\ldots,x_n]$ such that the following formal polynomial identity holds:
\begin{equation}\label{eq:ps}
\sum_{i\in I} p_i\cd f_i + \sum_{\zeta\subseteq J} \left(\prod_{j\in \zeta} h_j \cd \left(\sum_{i\in I_\zeta} s_{i,\zeta}^2\right)\right)=- 1\,.
\end{equation}%
 The \textbf{monomial size} of a \ps\ refutation is the combined total number of monomials in $\{p_i\}_{i\in I}$ and $\sum_{i\in I_\zeta} s_{i,\zeta}^2$, for all $\zeta\subseteq J$, that is, $\sum_{i\in I}\monomsize {p_i}+\sum_{\zeta\subseteq J}\monomsize {\sum_{i\in I_\zeta} s_{i,\zeta}^2}$.
\footnote{The definition of size measure for \ps\ and \sos\ proofs is slightly less standard than degree measure (see discussion in \cite{AH19}). We define the monomial size measure of \ps\ proofs to count the monomials in $p_i$ and $s_{i,\zeta}^2$, while ignoring the monomials in the initial axioms in \assmp, \ineqassmp. This choice of definition is closer to the definition of size of IPS proofs, which ignores the size of the initial axioms. }
\end{definition}

Note that Grigoriev et al.~\cite{GHP02} defined the size of \ps\ proofs slightly differently: they included in the size of proofs both the number of monomials and  the size of the coefficients of monomials written in binary
(while this does not matter for their lower bounds).
This is more natural when considering \ps\ as a propositional proof system (which is polynomially verifiable).


In order to use \ps\ as a refutation  system for  collections of equations \assmp\ and inequalities \ineqassmp\ that are
unsatisfiable over 0-1 assignments, we need to include simple  so-called boolean axioms.
This is done in slightly  different ways in different works  (see for example \cite{GHP02,AH19}).
One way to do this, which is the way we follow,
is the following:

\begin{definition}[(boolean) \ps\ proof system (boolean PS)]\label{def:prop-PS}
A \demph{boolean \ps\ proof} from a set of polynomial equations \assmp, and polynomial inequalities \ineqassmp, is an algebraic  \ps\ proof in which the following  \demph{boolean axioms} are part of the axioms: the polynomial equations $x^2_i-x_i=0$ (for all $i\in[n]$) are included in \assmp, and the polynomial inequalities $x_i\ge 0,~ 1-x_i\ge 0$ (for all $i\in[n]$) are included in \ineqassmp. \end{definition}
In this way, \assmp, \ineqassmp\ have no common 0-1 solutions iff there exists a boolean \ps\ refutation of \assmp, \ineqassmp.
 %
%
Eventually, to define the boolean \ps\ as a propositional proof system for the unsatisfiable CNF formula we consider CNF formulas to be encoded as polynomial equalities according to  \autoref{def:algebraic-transl-CNF}. This version is sometimes called \demph{propositional Positivstellensatz}.
\emph{As a default when referring to \ps\ we mean the boolean \ps\ version}.
\medskip

In recent years, starting mainly with the work of Barak, Brandao, Harrow, Kelner, Steurer and Zhou \cite{BarakBHKSZ12}, a special case of the \ps\ proof system has gained much interest due to its application in complexity and algorithms (cf.~\cite{OZ13}). This is the \demph{sum-of-squares} proof system (\textbf{SoS}), which is defined as follows:

\begin{definition}[sum-of-squares proof system]\label{def:sos-proof-system} A \demph{sum-of-squares proof} (SoS for short) is a \ps\ proof in which in \ref{eq:ps} in \autoref{def:PS} we restrict the index sets $\zeta\subseteq J$ to be \emph{singletons}, namely $|\zeta|=1$, hence, disallowing arbitrary products of inequalities within themselves. The real, boolean and propositional versions of \sos\ are defined similar to \ps.
\end{definition}
For most interesting cases \sos\ is also complete (and sound) by a result of  Putinar \cite{Put93}.

\subsubsection{Dynamic \ps}\label{sec:Stronger-Semi-Algebraic-Proof-Systems}
Here we  follow Grigoriev, Hirsch and Pasechnik  \cite{GHP02} to define what is, to the best of our knowledge, the most general propositional Positivstellensatz- (or SoS-) based semi-algebraic proof system defined to date.
It can be viewed as the generalization of (dynamic) Lovasz-Schrijver proof systems \cite{LS91,Lov94}
that have been put in the context of propositional proof complexity by Pudl\'ak \cite{Pud99},
and constitutes essentially a dynamic version of propositional \ps\
(the proof size is measured by the total number of monomials appearing in the proof).


The translation of propositional formulas here is different from the algebraic translation (\autoref{def:algebraic-transl-CNF}). For higher degree proof systems, \autoref{def:algebraic-transl-CNF} and the definition that follows can be reduced to one another (within the proof system, as long as both translations can be written down efficiently); however, we provide \autoref{def:semi-algebraic-transl-CNF} for the sake of consistency with earlier work.

\begin{definition}[semi-algebraic translation of CNF formulas]\label{def:semi-algebraic-transl-CNF}
Given a CNF formula in the variables $\vx$,  every clause $\bigvee_{i\in P} x_i \lor \bigvee_{j\in N} \neg{x_j}$ is translated into  $\sum_{i\in P} x_i+ \sum_{j\in N} (1-x_j)\ge 1$.
\end{definition}
Notice that in this way a 0-1 assignment to a CNF formula is satisfying iff the assignment  satisfies all the inequalities in the semi-algebraic translation of the CNF formula.

\begin{definition}[\LSInfty~\cite{GHP02}]\label{def:LSInfty}
Consider a boolean formula in conjunctive normal form and translate it
into inequalities as in \autoref{def:semi-algebraic-transl-CNF}.
Take these inequalities as axioms,
add the axioms $x\ge0$, $1-x\ge0$, $x^2-x\ge0$, $x-x^2\ge0$
for each variable $x$. Allow also $h^2\ge0$
as an axiom, for any polynomial $h$ of degree at most $d$.
An {\rm LS}$^d_{*,+}$ proof
of the original formula is a derivation of $-1\ge 0$
from these axioms using the following rules:
\[{{f\ge0,\quad g\ge0}\over{f+g\ge0}}~~~~~~~~~~~{{f\ge0}\over{\alpha f\ge0}}~\textup{(for $\alpha$ a nonnegative integer)} ~~~~~~~~~{{f\ge0,\quad g\ge0}\over{f\cdot g\ge0}}.\]
In particular, {\rm LS}$^\infty_{*,+}$ is such a proof
without the restriction on the degree.
Note that we have to write polynomials as sums of monomials
(and not as circuits or formulas), so the verification
of such proof is doable in deterministic polynomial-time.
\end{definition}

The proof of the following simulation follows by definition and we omit the details:
\begin{proposition}\label{prop:LSInfty-simulates-PS}
\LSInfty\ simulates boolean \ps.
\end{proposition}

\section{Conditional IPS Lower Bounds}\label{sec:lower-bounds}
\subsection{IPS Lower Bounds under Shub-Smale Hypothesis}\label{sec:Shub-Smale-lower-bound}
Here we provide a super-polynomial conditional lower bound on the size of (boolean) IPS refutations of the binary value principle  over the rationals  based on the Shub-Smale Hypothesis (\autoref{sec:tau-conj}).



The conditional lower bound is first established  for constant-free IPS proofs over \Z\ and then we extract a lower bound over \Q\ as a corollary using \autoref{cor:from-Q-IPS-to-Z-IPS} below.
Notice that we can consider IPS proofs also over rings, and not only fields. In the case of IPS over \Z\ we cannot anymore assume that \emph{refutations} are proofs of the polynomial $1$, rather we define refutations in IPS over \Z\ to be proofs of any nonzero constant polynomial (cf.~\cite[Definition 2.1]{BussIKPRS96}):

\begin{definition}[\ipsz]
An \emph{\ipsz\ proof of $g(\vx)\in\Z[\vx]$ from a set of assumptions $\assmp\subseteq\Z[\vx]$} is an IPS proof of $g(\vx)$ from \assmp, as in \autoref{def:IPS}, where $\F=\Z$ and all the constants in the IPS proof are from \Z. An \emph{\ipsz\ refutation of \assmp} is a proof of $M$, for $M\in\Z\setminus\{0\}$. (The definition is similar for the  boolean and algebraic  IPS versions.)
\end{definition}



We will need to define a constant-free circuit over \Q.

\begin{definition}\label{def:const-free-q}
A \emph{constant-free circuit over \Q}
is a constant-free algebraic circuit as in \autoref{sec:algebraic_circuits}
that has additionally
division gates $\div$, where $u\div v$ means that the polynomial computed by $u$ is divided by the polynomial computed by $v$, such that for every division gate $u\div v$ in $C$ the circuit $v$ contains \emph{no variables} and computes a nonzero constant.
 A \emph{constant-free IPS proof over \Q\ }is an IPS proof written with a constant-free circuit over \Q.
\end{definition}

The following proposition  is proved by a simple induction on the circuit size, using sufficient products to cancel out the denominators in the circuit over \Q, turning  it into a circuit over \Z.

\begin{proposition}[from \Q-circuit to \Z-circuit]
\label{prop:Q-circuit-to-Z-circuit}
Let $C$ be a size-$s$ constant-free circuit over \Q. Then there exists a size $\le 4s$
constant-free circuit computing $M\cd \widehat C$, for some $M\in \Z\setminus\{0\}$, with $\tau(M)\le 4s$.
\end{proposition}

%
%
%

\begin{proof}
We choose any topological order $g_1,g_2,\ldots,g_i,\ldots,g_{|C|}$ on the gates of the constant-free circuit $C$ over \Q\ (that is, if $g_j$ has a directed path to $g_k$ in $C$ then $j<k$) and proceed by induction on $|C|$ to  eliminate rationals from the circuit (identifying the gate $g_i$ with the sub-circuit of $C$ for which $g_i$ is its root). \smallskip

\nind\textbf{Induction statement:} Let $g_1,\dots,g_s$ be the topologically ordered gates of  a constant-free circuit $C$  over \Q, where $s=|C|$.
Then, there exists a division-free constant-free circuit consisting of the corresponding topologically ordered  gates $g_{11},\ldots,g_{1a_1},g_{21},\ldots,g_{2a_2},\dots,g_{s1},\dots,g_{s a_s}$, such that for \emph{every} $i\le s$:
\begin{enumerate}
\item $a_i\le 4$ and $g_{i a_i}$ computes the polynomial $M_i\cd g_i$, for some nonzero integer $M_i$ (again, identifying the gate $g_{i a_i}$ with the sub-circuit for which it is a root);
\label{it:IPSQ-main-cond}
\item
The integer $M_i$ is  constructed as a part of the circuit (except for the trivial case $M_i=1$). More precisely, there exists a division-free constant-free and variable-free (sub-)circuit $g_{j,\ell}$, for $j\le i,\ell\le 4$ that computes $M_i$. In particular $\tau(M_i)\le 4i$;
\label{it:M-constructible-already}
\item Sub-circuits with no division gates remain intact: if $g_i$ is a division-free constant-free circuit  then $M_i=1$ and $g_i$ is a sub-circuit of the new circuit. That is, $g_{i1}=g_i$, $a_i=1$, and all gates in $g_1,\dots,g_i$ that are part of the  sub-circuit $g_i$ in $C$ occur also as  gates $g_{j\ell}$ (for some $j\le i,\ell\le 4$) in the new circuit $g_{i1}$.\label{it:weird-condition}
\end{enumerate}

\smallskip

\Base  $g_i$ is a variable or a constant in $\{-1,0,1\}$. Hence, we put $g_{i1} := g_i$, $a_i=1$, and $M_i=1$.

\induction In the case of a binary gate $g_i = g_j\circ g_\ell$, for $\circ\in\{\times,+, \div \}$ (where $j,\ell<i$),  by induction hypothesis we already  have division-free constant-free circuits  $g_{ja_j}$ and $g_{\ell a_\ell}$ computing the polynomials  $M_j g_j$ and $M_\ell g_\ell$, respectively, for some integers $M_j,M_\ell$ that are also computed as part of the circuit.

\smallskip

\case 1 $g_i$ is a division gate computing $g_j / g_\ell$, where, by definition of circuits over \Q,  $g_\ell$ is a division-free constant-free circuit computing a nonzero constant that contains no variables.

By induction hypothesis \autoref{it:IPSQ-main-cond} we have already constructed the two gates $g_{ja_j}$ and $g_{\ell a_\ell}$, where $g_{ja_j}$ computes the polynomial $M_j g_j$ for some nonzero integer $M_j$. By \autoref{it:M-constructible-already}, $M_j$ is already computed by one of the gates in the circuit. Finally, since $g_\ell$ does not have division gates by definition,   \autoref{it:weird-condition} means that  $g_{\ell a_\ell}=g_\ell$ (and $a_\ell=a_1$ and $M_\ell=1$), and specifically $g_\ell$ is a constant-free variable-free circuit.

We put $g_{i2}:=g_{ja_j}$ and $g_{i1}:= M_j\cd g_{\ell a_\ell}=M_j\cd g_\ell$ (that is, $g_{i1}$ is a product gate that connects to the two previously constructed gates that compute the two integers $M_j$ and $g_\ell$),  $a_i=2$ and $M_i= M_j g_\ell$ is an integer. Hence, $g_{i2}$ computes the polynomial $g_{ja_j}=M_j  g_j=g_\ell\cd ((M_jg_j)/g_\ell)=g_\ell\cd M_j \cd(g_j/ g_\ell)=(g_\ell\cd M_j) \cd g_i=M_i\cd g_i$ and  $M_i= M_j g_\ell$ is an integer that is computed (as a constant-free variable-free circuit) by the gate $g_{i1}$ as required.


\case 2 $g_i= g_j\cd g_\ell$.
In this case $a_i=2$ and  $M_i = M_j M_\ell$, and we put $g_{i2}:=g_{ja_j}\cd g_{\ell a_\ell}$ and $g_{i1}:=M_i\cd M_j$, where $M_i, M_j$ are two integers that are computed already in the circuit (with a constant-free division-free and variable-free sub-circuit).

\case 3 $g_i= g_j+g_\ell$.
In this case $a_i=4$,  $M_i = M_j M_\ell$, and we put $g_{i4}:=M_\ell\cd g_{ja_j}+M_j\cd g_{\ell a_\ell}$, namely, we add three gates $g_{i2},g_{i3},g_{i4}$ (two products, both of which connects to previous gates, and one addition to add these two products). Finally, we put $ g_{i1}:=M_i\cd M_j$, where $M_i, M_j$ are two integers that are computed already in the circuit (with a constant-free division-free and variable-free sub-circuit).
\end{proof}

An immediate corollary of \autoref{prop:Q-circuit-to-Z-circuit} is:

\begin{corollary}[from \ipsq\ to \ipsz]
\label{cor:from-Q-IPS-to-Z-IPS}
Boolean IPS$_\Z$ simulates boolean IPS$_\Q$, in the following sense: if there exists a size-$s$ constant-free boolean IPS proof over \Q\ from \assmp\ of $H$,
for   $\assmp$  a set of assumptions written as constant-free algebraic circuits over \Z\ and
$H$  a constant-free algebraic circuit over \Z,  then there exists a size $\le 4s$
constant-free boolean IPS$_\Z$ proof of $M\cd H$, for some $M\in \Z\setminus\{0\}$, such that  $\tau(M)\le 4s$.
\end{corollary}

 \begin{theorem}\label{thm:first-cond-lower-bound-Q}
Under the Shub and Smale Hypothesis, there are no $\poly(n)$-size constant-free (boolean) IPS refutations of the binary value principle \bvpn\ over \Q.
\end{theorem}

\begin{proof}
Given \autoref{cor:from-Q-IPS-to-Z-IPS}, it suffices to prove the statement for constant-free (boolean) IPS$_\Z$.

We proceed to prove the contrapositive. Suppose that the binary value principle $1 + \sum_{i = 1}^{i = n} 2^{i - 1} x_i=0$ has a constant-free  IPS$_\Z$ refutation (using only the boolean axioms)  of size $\poly(n)$. We will show that there is a  sequence of nonzero natural numbers $c_m$ such that $\tau(c_m m!) \le (\log m) ^ {c}$, for all $m\ge2$, where $c$ is a constant independent of $m$. In other words, we will show that $(c_m m!)_{m=1}^\infty$ is easy.

Assume that $C(\vx,y,\vz)$ is the polynomial-size constant-free boolean IPS$_\Z$ refutation of $1 + \sum_{i = 1}^{i = n} 2^{i - 1} x_i=0$ (here we  only have a single placeholder variable $y$ for the single non-boolean axiom, that is,  the binary value principle). For simplicity, denote  $G(\vx)=1 + \sum_{i = 1}^{i = n} 2^{i - 1} x_i$, $F_i(\vx)=x_i^2-x_i$, and $\overline {F}(\vx)=\ba$.

We know that there exists an integer constant $M \neq 0$ such that

\begin{equation}\label{eq:C(.)=M}
C\left(\vx, G(\vec x),\overline F(\vx)\right) = M\,.
\end{equation}

For every integer $0 \le k < 2^n$ we denote by $\overline b_k:=(b_{k 1}, \ldots, b_{k n})\in\bits^n$ its (positive, standard) binary representation, that is,   $k = \sum_{i = 1}^{i = n}b_{k i} 2^{i - 1}$. Then, $F_i(\vec b_k) = 0$ and $G(\vec b_k) = 1 + k$, for all $1\le i \le n$, $0 \le k < 2^n$. Hence, by \ref{eq:C(.)=M}:

\begin{equation}\label{eq:lb1}
C(b_{k1}, \ldots, b_{kn}, 1+k, \vec 0) = M,\text{~~for every integer $0 \le k < 2^n$}.
\end{equation}
\begin{claim}\label{cla:M-div-every-prime}
 $M$ is divisible by every prime number less than $2^n$.
\end{claim}
\begin{proofclaim}
For a fixed $0\le k<2^{n}$ and its binary representation $ b_{k1}, \ldots, b_{kn}$, consider $g(y) = C(b_{k1}, \ldots, b_{kn}, y, \vec 0)$ as  a univariate polynomial in $\Z[y]$. Then, $g(1+k) = M$ by \ref{eq:lb1}, and $g(0) = 0$ holds since  $C(b_{k1}, \ldots, b_{kn}, 0, \vec 0) = 0$, by the definition of IPS. Because $g(0)=0$, we know that $g(y)=y\cd g^\star(y)$, for some $g^\star(y)\in\Z[y]$, meaning that $g(1+k)=(1+k)\cd g^\star(1+k)=M$. Since $g^\star(y)$ is an integer polynomial, this implies that $M$ is a product of $1+k$.

Overall, this argument shows that for every $1\le p\le 2^n$, $M$ is divisible by  $p$, and in particular  $M$ is divisible by every prime number less than $2^n$.
\end{proofclaim}


Note that once we substitute the all-zero assignment $\overline 0$ into \ref{eq:C(.)=M},
we obtain a constant-free algebraic circuit of size $\poly(n)$ with no variables computing $M$,
thus $\tau(M)=\poly(n)$.
Then we can compute $M^{2^n}$ using a constant-free algebraic circuit of size $\poly(n)$ by taking $M$ to the power $2$, $n$ many times
(that is, using $n$ repeated squaring),
yielding $\tau(M^{2^n})=\poly(n)$.

\begin{claim}\label{cla:prime-factorisation}
The power of every prime factor in $(2^n)!$ is at most $2^n$.
\end{claim}
\begin{proofclaim}
We show that for every number $k\in\N$, the power of every prime factor of $k!$ is at most $k$. Let $ p_1^{t_1} \cdots p_r^{t_r}$ be the prime factorisation of $k!$, namely $k!= p_1^{t_1} \cdots p_r^{t_r}$ where each $p_i$ is a prime number and $p_i \neq p_j$, for all $i\neq j$. To compute $t_i$ we consider the $k$ products $k, (k-1),\dots, 1$, in $k!=k\cd(k-1)\cdots 1$, out of which only each $p_i$th number is divisible by $p_i$, hence only $\lfloor \frac{k}{p_i} \rfloor$ numbers are divisible by $p_i$. Consider now only these $\lfloor \frac{k}{p_i} \rfloor$ numbers in $k!$ which are divisible by $p_i$, and write them as  $p_i\cd \lfloor \frac{k}{p_i} \rfloor ,p_i\cd (\lfloor \frac{k}{p_i} \rfloor-1) ,\dots, p_i\cd 1$. Now we need once again to factor out the $p_i$ products in $\lfloor \frac{k}{p_i} \rfloor , \lfloor \frac{k}{p_i} \rfloor-1 ,\dots, 1$. Hence,  as before, we conclude that in these  ${\lfloor \frac{k}{p_i} \rfloor}$ numbers only $\lfloor \frac{{\lfloor \frac{k}{p_i} \rfloor}}{p_i}\rfloor =  {\lfloor \frac{k}{p_i^2} \rfloor}$ are divisible by $p_i$. Continuing in a similar fashion we obtain the equation  $t_i= \lfloor \frac{k}{p_i} \rfloor + \lfloor \frac{k}{p_i ^ 2} \rfloor + \lfloor \frac{k}{p_i ^ 3} \rfloor + \dots \le \frac{k}{p-1}$.
\end{proofclaim}

Consider the $\poly(n)$-size circuit for $M^{2^n}$ that exists by assumption. Since $M$ is divisible by every prime number between 1 and $2^n$, and since every prime factor of $(2^n)!$ is clearly at most $2^n$, we get that $M^{2^n}$ is divisible by the $2^n$-th power of \emph{each} prime factor of $(2^n)!$. By \autoref{cla:prime-factorisation} the power of every prime factor of $(2^n)!$ is at most $2^n$, and so  $M^{2^n}$ is divisible by $(2^n)!$. We conclude that there are nonzero numbers $c_n\in\N$  such that the sequence $\{c_n\cd(2^n)!\}_{n=1}^\infty$ is computable by a sequence of constant-free algebraic circuits of size $\poly(n)$, that is,   $\tau(c_{n}\cd(2^n)!) \le n^c$ for some constant $c$ independent of $n$.
It remains to show that not only the  multiples  of factorials \emph{of powers of $2$ }are easy, but also the  multiples  of factorials of \emph{all} natural numbers are easy.


 For every  natural number $m$, let $n\in\N$ be such that  $2^{n - 1} \le m \le 2^n$. Because $(2^n)!$ is clearly divisible by $m!$, there exists some $c_m\in\N$, such that $c_{n} \cd(2^n)!=c_m\cd m!$, where $c_n$ is the natural number for which we have showed the existence of $\poly(n)$-size constant-free circuit computing $c_{n}\cd(2^n)!$. Hence, this same circuit also computes  $c_m \cd m!$, meaning  that $\tau(c_m \cd m!) \le n^b \le (\log (2 m))^b \le (\log m)^c$, for some constants $b$ and $c$ independent of $m$.
%
\end{proof}

\para{Why does  an IPS lower bound on BVP not lead to Extended Frege  lower bounds?}
Given that IPS (of possibly exponential degree) simulates Extended Frege (EF) \cite{GP14,PT16}, it is interesting to consider why our conditional IPS lower bound for the BVP does not imply a conditional EF lower bound. Simply put, the answer is that the BVP is not a propositional tautology (or a direct translation of one). More precisely, there is no apparent way to translate the BVP into a propositional tautology for which a short EF proof translates into a short IPS refutation of the BVP.
In other words, we can encode the BVP as a propositional tautology stating that the carry-save addition of the $n$ numbers in the BVP has sign-bit 0 (and hence the addition is positive), but the problem is that there is no apparent way to efficiently derive in IPS this encoding from the BVP principle itself (!), \emph{because from an  equation like $f=0$ we cannot in general efficiently derive in IPS  that the sign-bit of $f$ is zero}, as we now explain.

One can think of the following translation of the BVP into a propositional tautology: we consider the addition of $n$ numbers $2^{i-1}x_i$, for $x_i\in\bits$ and $i=1,\dots,n$. Each $2^{i-1}x_i$ is written as a bit vector $\mathbf v_i$ of at most $n$ bits, in the two's complement notation. Each bit in $\mathbf v_i$ can be written as a polynomial-size boolean circuit in the single boolean variable $x_i$.  Using carry-save addition we can construct a polynomial-size in $n$ boolean circuit  $C$ computing the sign-bit  of the addition of these $n$ bit-vectors $\sum_{i=1}^n \mathbf v_i$ (this is done as in \autoref{sec:Reasoning-about-Bits-within-Algebraic-Proofs}). Now, the BVP can be  encoded as the tautology $C\equiv \false$ (namely, the sign bit of the addition is logically equivalent to \emph{false} (equivalently, 0); note that since $2^{i-1}\ge 0$, for all $i$, this is indeed a tautology).

Apparently, there is a polynomial-size in $n$ EF proof of $C\equiv\false$ (using  basically the same ideas as in \autoref{sec:Reasoning-about-Bits-within-Algebraic-Proofs}). The question is whether we can turn this short EF proof into a short IPS refutation of the BVP. And apparently the answer   is ``no!". The reason is that for this to work we first need to derive in IPS from the BVP the (arithmetization of the boolean circuit) $C\equiv\false$. But such a  derivation is already morally equivalent to refuting the BVP itself. In other words, there is no apparent way to efficiently derive  from the given BVP equation  $\sum_{i=1}^n 2^{i-1}x_i+1=0$ any statement expressing a specific property pertaining to a  \emph{single} bit in the bit vector representation of $\sum_{i=1}^n 2^{i-1}x_i+1$ (as a function of the input boolean variables \vx), and specifically no apparent way to derive $C\equiv\false$. Our main technical \autoref{lem:main-binary-value-lemma} in \autoref{sec:Reasoning-about-Bits-within-Algebraic-Proofs}, shows only that we can efficiently derive in IPS from the BVP equation a statement  about the \emph{collective value of the bits}, namely that if  $\sum_{i=1}^n 2^{i-1}x_i+1$ is denoted by $f$ we have the following:
\medskip

~~\begin{minipage}[c]{0.9\textwidth}
\textbf{if}
~~(i) the bits of $f$ (computed as polynomial-size circuits of the input variables $\vx$) are $z_1,\dots,z_m$; and
(ii) we know that $f=0$;

\textbf{then}
~~(iii) $\sum_{i=1}^{m-1} 2^{i-1}z_i-2^{m-1}z_m=0$ (where the left hand side is the value of the bit-vector $z_1,\dots,z_m$ that represents an integer in the two's complement scheme).
\end{minipage}
\medskip

\nind\emph{But from the equation  in (iii) we apparently cannot conclude anything about the \emph{individual} sign-bit $z_m$}.


\subsection{IPS over Rational Functions and the $\tau$-Conjecture}\label{sec:rational-field-lower-bounds}

Here we deal with IPS operating over the field of rational functions in the (new) indeterminate $y$. This will allow us to formulate an interesting  version of the binary value principle. Roughly speaking, this version expresses the fact that the BVP is ``almost always'' unsatisfiable.  More precisely, consider the equation  $\sum_{i=1}^n 2^{i-1}x_i = y$. This equation is unsatisfiable for most $y$'s, when $y$ is substituted by an  element from \Q. In the setting of IPS refutations over the field of rational functions in the indeterminate $y$, refuting $\sum_{i=1}^n 2^{i-1}x_i = y$ would correspond to refuting $\sum_{i=1}^n 2^{i-1}x_i = M$, for all $M\in\Q$ but a finite set of  numbers from \Q\ (see below).


We shall  prove a super-polynomial lower bound on   $\sum_{i=1}^n 2^{i-1}x_i = y$, over the fields of rational functions in the indeterminate $y$,  subject to the $\tau$-conjecture. 


\begin{definition}[$\mathbb{Q}$-rational functions]\label{def:Q[y]}
Denote by $\mathbb{Q}(y)$ the field of \emph{$\mathbb{Q}$-rational functions in $y$}, that is, all functions $f(y):\Q\to\Q$ such that there exist $P(y)\in\Q[y]$ and nonzero $Q(y)\in\Q[y]$
with $f(y) = \frac{P(y)}{Q(y)}$.
\end{definition}

In particular, in this system one can consider refutations of $\sum_{i=1}^n 2^{i-1}x_i+y=0$, where $x_i$ are boolean variables (the boolean axioms $x_i^2-x_i=0$ are included in the initial axioms). In this section we will be using the concept of a \emph{linear} IPS refutation (proved to be polynomially equivalent to general IPS, at least in the unit-cost model where each coefficient appearing  in an algebraic circuit is considered to contribute only $1$ to the overall size of the  circuit), defined in Forbes et al.~\cite{FSTW16}:
\begin{definition}[\cite{FSTW16}]\label{def:IPS-LIN}
An \emph{IPS-LIN${}_{\mathbb{Q}(y)}$-certificate} of the unsatisfiability of
a system of polynomial equations $F_1(\vec{x}) = F_2(\vec{x}) =  \cdots = F_m(\vec{x})= 0$
is a set of polynomials $(H_1(\vec x), \ldots H_m(\vec x))$, where each $H_i(\vec x) \in \mathbb{Q}(y)[x_1, \ldots, x_n]$, such that $F_1(\vec{x}) \cdot H_1(\vec{x}) + \dots + F_m(\vec{x}) \cdot H_m(\vec{x}) = 1$ (as a formal polynomial equation).
\end{definition}


We assume that the  $F_j$'s include the boolean axioms $x_i^2-x_i$
for every variable $x_i$.
The system is complete for this case, as discussed in the next subsection.

Note that once we have an IPS-LIN${}_{\mathbb{Q}(y)}$-certificate of a system of equations that include the boolean axioms and the equation  $\sum_i x_ia_i=y$, we can substitute for $y$ any constant except for the finite number of roots of the denominators of $H_i$'s and get a valid IPS-LIN${}_{\mathbb{Q}}$ refutation.
\emph{Thus an IPS-LIN${}_{\mathbb{Q}(y)}$-certificate can be viewed as a single proof for all but finitely many values of $y$.}

To show this concept is meaningful, we first show a short IPS-LIN$_{\mathbb{Q}(y)}$ proof of
$\sum_{i=1}^n a_i x_i=y$ for small scalars $a_i$. Then we demonstrate
a lower bound for $a_i=2^{i-1}$ modulo the $\tau$-conjecture.

We start with precise definitions of the complexity of IPS-LIN${}_{\mathbb{Q}(y)}$-proofs and related completeness issues.

\subsubsection{Complexity Considerations}

To compute elements of $\mathbb{Q}(y)$, we extend the definition of a constant-free circuit
by allowing the use of gates for $y$. The definition of a constant-free circuit over $\mathbb{Q}(y)$
thus mimics \autoref{def:const-free-q}, but we allow now the constant $y$ in addition to $-1,0,1$
(note that $y$ is indeed a constant in terms of polynomials in $\mathbb{Q}(y)[x_1,\ldots,x_n]$).

Note that the system we consider is complete for the boolean case,
that is, for every inconsistent (over $\{0,1\}$) system
of polynomial equations involving coefficients in $\mathbb{Q}(y)$
that contains the boolean equation $x^2-x=0$ for every variable $x$,
there is an IPS-LIN${}_{\mathbb{Q}(y)}$ certificate.
Indeed, the system remains inconsistent in the algebraic closure of $\mathbb{Q}(y)$
(as every solution must satisfy $x^2-x=0$),
and thus by Hilbert's Nullstellensatz
the linear system that has $H_i$'s coefficients as variables and expresses that $H_i$'s form a valid certificate,
has a solution. Since the coefficients of this linear system are in $\mathbb{Q}(y)$, so must be (some) solution.

\begin{remark}
In \cite{FSTW16} Forbes et al.~prove that IPS is polynomially equivalent to IPS-LIN when the scalars are given for free (that is, do not count towards the proof size).
We believe that a similar transformation can be  made for the constant-free model to establish the equivalence between IPS${}_{\mathbb{Q}(y)}$ and IPS-LIN${}_{\mathbb{Q}(y)}$;
however, we did not verify the details.
\end{remark}

\subsubsection{Upper Bound}
\begin{proposition}
Suppose we have  a system of polynomial equations $F_0(\vec{x}) = F_1(\vec{x}) = F_2(\vec{x}) =  \cdots = F_n(\vec{x})= 0$, $F_i \in \mathbb{Q}(y)[x_1, \ldots, x_n]$, where $F_0(\vec x) = y + \sum_{i = 1}^{i = n} a_i x_i$, $a_i \in \mathbb{N}$ and $F_i(\vec x) = x_i^2 - x_i$. Then there is an IPS-LIN$_{\mathbb{Q}(y)}$ certificate of this system consisting of $H_0(\vec{x}), \ldots, H_{n + 1}(\vec x)$, where each $H_i(\vec x)$ can be computed by a constant-free algebraic circuit over $\mathbb{Q}(y)$ of size $\poly(a_1 + \cdots + a_n)$.
\end{proposition}

\begin{proof}
We will construct our proof by induction:

\bfseries Base case: \mdseries suppose $G_{0, t}(\vec x) = y + t$, $t \in \mathbb{N}, t \le \sum_{i = 1}^{n} a_i$. Then we can take $H_{0, 0, t}(\vec{x}) = \frac{1}{y + t}$ and $H_{0, i, t} = 0$ where $1 \le i \le n$, $i \in \mathbb{N}$ as an IPS-LIN${}_{\mathbb{Q}(y)}$ certificate for a system of polynomial equations $G_{0, t}(\vec{x}) = F_1(\vec{x}) = F_2(\vec{x}) =  \cdots = F_n(\vec{x})= 0$.

\bfseries Induction step: \mdseries suppose we have already built certificates $H_{k, 0, t}(\vec x), \ldots, H_{k, n, t}(\vec x)$ for the systems of polynomial equations $G_{k, t}(\vec{x}) = F_1(\vec{x}) = F_2(\vec{x}) =  \cdots = F_n(\vec{x})= 0$ where $G_{k, t}(\vec x) = y + t + a_1 x_1 + \ldots a_k x_k$, $t \in \mathbb{Z}$, $0\le t \le \sum_{i = k + 1}^{n} a_i$. Now we will build certificates $H_{k + 1, 0, t}(\vec x), \ldots, H_{k + 1, n, t}(\vec x)$ for the systems of polynomial equations $G_{k + 1, t}(\vec{x}) = F_1(\vec{x}) = F_2(\vec{x}) =  \cdots = F_n(\vec{x})= 0$ where $G_{k  + 1, t}(\vec x) = y + t + a_1 x_1 + \ldots a_{k + 1} x_{k + 1}$, $t \in \mathbb{Z}$, $0\le t \le \sum_{i = k + 2}^{n} a_i$. There are the following cases:
\begin{enumerate}
    \item If $i > k + 1$, then we will take $H_{k + 1, i, t}(\vec x) = 0$.
    \item If $i = k + 1$, then we will take $H_{k + 1, i, t}(\vec x) = a_{k + 1} (H_{k, 0, t}(\vec x) - H_{k, 0, t + a_{k + 1}}(\vec x))$.
    \item If $0 \le i < k + 1$, then we will take $H_{k + 1, i, t}(\vec x) = x_{k + 1} H_{k, i, t + a_{k + 1}}(\vec x) + (1 - x_{k + 1}) H_{k, i, t}(\vec x)$.
\end{enumerate}

The main idea of this construction is the case analysis for $x_{k+1}=0$, $x_{k+1}=1$, that is,
$$
(y + t + a_1 x_1 + \ldots + a_{k + 1} x_{k + 1}) x_{k + 1} - a_{k + 1} (x_{k + 1}^2 - x_{k + 1}) = (y + t + a_{k + 1} + a_1 x_1 + \ldots a_{k} x_{k}) x_{k + 1}
$$
and
$$
(y + t + a_1 x_1 + \ldots + a_{k + 1} x_{k + 1}) (1 - x_{k + 1}) + a_{k + 1} (x_{k + 1}^2 - x_{k + 1}) = (y + t + a_1 x_1 + \ldots a_{k} x_{k}) (1 - x_{k + 1}).
$$

which means that (using the induction hypothesis)
\begin{multline*}
((y + t + a_1 x_1 + \ldots a_{k + 1} x_{k + 1}) x_{k + 1} - a_{k + 1} (x_{k + 1}^2 - x_{k + 1}))H_{k, 0, t + a_{k +1}}(\vec x) +\\ +(x_1^2 - x_1)x_{k + 1}H_{k, 1, t + a_{k + 1}}(\vec x) + \ldots+
(x_k^2 - x_k)x_{k + 1}H_{k, k, t + a_{k + 1}}(\vec x) = x_{k + 1}
\end{multline*}
and
\begin{multline*}
((y + t + a_1 x_1 + \ldots a_{k + 1} x_{k + 1}) (1 - x_{k + 1}) + a_{k + 1} (x_{k + 1}^2 - x_{k + 1}))H_{k, 0, t}(\vec x) +\\ +(x_1^2 - x_1)(1 - x_{k + 1})H_{k, 1, t }(\vec x) + \ldots+
(x_k^2 - x_k)(1 - x_{k + 1})H_{k, k, t }(\vec x) = 1 - x_{k + 1}
\end{multline*}

Summing up the equations for both cases, due to the fact that $(1 - x_{k + 1}) + x_{k + 1} = 1$ we
get $G_{k+1,t} H_{k+1,0,t} + \sum_{i=1}^n F_i H_{k+1,i,t}=1$.

On each step of our induction we create no more than $\poly(a_1 + \cdots + a_n)$ new gates computing algebraic circuits for $\mathbb{Q}(y)[x_1, \ldots, x_n]$-polynomials $H_{k, 0, t}(\vec x), \ldots, H_{k, n, t}(\vec x)$. Thus, we can take $H_0(\vec{x}) = H_{n, 0, 0}(\vec x), \ldots, H_n(\vec{x}) = H_{n, n, 0}(\vec x)$ to conclude our proof.

\end{proof}

\subsubsection{Lower Bound}

\begin{lemma}
Suppose we have a constant-free circuit $C$ over $\mathbb{Q}(y)$ of size $M$
computing a polynomial in $\mathbb{Q}(y)[x_1,\dots,x_n]$
that is a rational function
$f(y, x_1, \ldots, x_n)$.
Then there are two
constant-free circuits over $\mathbb{Z}$
of size less than $3M$ computing polynomial functions
$P(y, x_1, \ldots, x_n) \in \mathbb{Z}[y, x_1,\dots,x_n]$ and $Q(y) \in \mathbb{Z}[y]$
such that $f(y, x_1, \ldots, x_m) = \frac{P(y, x_1, \ldots, x_n)}{Q(y)}$.
\end{lemma}
\begin{proof}
Consider any topological order $g_1, \ldots, g_M$ on the gates of $C$.
We will gradually rewrite our circuit starting from $g_1$.
Assume that we have already done the job for $g_1,\ldots,g_k$,
that is, for each $i \le k$ there are appropriate algebraic circuits
for polynomial functions $P_i(y, x_1, \ldots, x_n) \in \mathbb{Z}[y, x_1, \ldots, x_n]$ and $Q_i(y) \in \mathbb{Z}[y]$
such that $g_i = \frac{P_i}{Q_i}$. We now augment these circuits
to compute the polynomials for $g_k$.

Here are all possible cases:
\begin{enumerate}
    \item $g_1$ is a variable $x_j$, then $P = x_j$, $Q = 1$.
    \item $g_1$ is a constant from $\mathbb{Q}(y)$ (that is, $0,-1,1,y$), then $P$ computes this constant, and $Q=1$.
    \item $g_{k + 1} = \frac{g_i}{g_j}$, where $i, j \le k$. In this case $P_i \in \mathbb{Q}(y)$ because of the structure of our $\mathbb{Q}(y)[x_1,\dots,x_n]$-circuit.
    Then $Q_{k + 1} = Q_{i} P_{j}$ and $P_{k + 1} = P_{i} Q_{j}$ and sizes of the circuits for $P_{k + 1}$ and $Q_{k + 1}$ are less than $3 \cdot (k + 1)$.
    \item $g_{k + 1} = g_i \cdot g_j$, where $i, j \le k$.  Then $Q_{k + 1} = Q_{i} Q_{j}$ and $P_{k + 1} = P_{i} P_{j}$ and sizes of the circuits for $P_{k + 1}$ and $Q_{k + 1}$ are less than $3 \cdot (k + 1)$.
    \item $g_{k + 1} = g_{i} + g_{j}$. Then $P_{k + 1} = P_{i} Q_{j} + P_{j} Q_{i}$ and $Q_k = Q_i Q_j$ and sizes of the circuits for $P_{k + 1}$ and $Q_{k + 1}$ are less than $3 \cdot (k + 1)$.
\end{enumerate}
We can conclude our proof by taking $P_M$ and $Q_M$ as $P$ and $Q$, respectively.
\end{proof}

\begin{theorem}\label{thm:second-lower bound}
Suppose a system of polynomial equations $F_0(\vec{x}) = F_1(\vec{x}) = F_2(\vec{x}) =  \cdots = F_n(\vec{x})= 0$, $F_i \in \mathbb{Q}(y)[x_1, \ldots, x_n]$, where $F_0(\vec x) = y + \sum_{i = 1}^{i = n} 2^{i - 1} x_i$ and $F_i(\vec x) = x_i^2 - x_i$, has an IPS-LIN$_{\mathbb{Q}(y)}$ certificate $H_0(\vec{x}), \ldots, H_{n}(\vec x)$, where each $H_i(\vec x)$ can be computed by a $\poly(n)$-size constant-free algebraic circuit over $\mathbb{Q}(y)$. Then, the $\tau$-conjecture is false.
\end{theorem}

\begin{proof}
Based on the above lemma, we can say that there are polynomials $P_i(y, x_1, \ldots, x_n) \in \mathbb{Z}[y, x_1, \ldots, x_n]$ and $Q_i(y) \in \mathbb{Z}[y]$ such that $H_i = \frac{P_i}{Q_i}$ for every $i$. Also we know that
$$
(y + x_1 + \dots + 2^{n - 1} x_n) \frac{P_0}{Q_0} + (x_1 ^ 2 - x_1) \frac{P_1}{Q_1} + \dots + (x_n ^ 2 - x_n) \frac{P_n}{Q_n} = 1
$$

So we can derive that

\begin{equation}\label{eq:final-rational-lb}(y + x_1 + \dots + 2^{n - 1} x_n) P_0 \prod_{j = 1}^{n} Q_j + (x_1 ^ 2 - x_1) P_1 Q_0 \prod_{j = 2}^{n} Q_j + \dots + (x_n ^ 2 - x_n) P_n  \prod_{j = 0}^{n - 1} Q_j =  \prod_{j = 0}^{n} Q_j
\end{equation}

Denote $ \prod_{j = 0}^{n} Q_j$ by $Q(y)$. From the above lemma we know that there is a constant-free circuit
over $\mathbb{Z}$ of size $\poly(n)$ for $Q(y)$.
Furthermore, for any integer $y$ such that $0\ge y > -2^n$,
there are values for $x_i$ (namely, the bit expansion of $-y$) such that the left hand side  of \ref{eq:final-rational-lb} is zero,
and hence $Q(y) = 0$. However, it contradicts the $\tau$-conjecture.
\end{proof}


\section{The Cone Proof System}\label{sec:CPS}
Here we define a very strong semi-algebraic proof system under the name Cone Proof System (CPS for short).
Similarly to other semi-algebraic systems, CPS establishes that a collection of polynomial equations  $\assmp:=\{f_i=0\}_i$ and polynomial inequalities $\ineqassmp:=\{h_i\ge 0\}_i$ are unsatisfiable over 0-1 assignments (or over real-valued assignments, when desired (\ref{def:real-CPS})). In the spirit of the Ideal Proof System (IPS) of Grochow and Pitassi \cite{GP14} we are going to define  a refutation in CPS as a \emph{single} algebraic circuit. Specifically, a CPS refutation is a circuit $C$ that computes a polynomial that results from \emph{positive-preserving} operations such as addition and product applied between the inequalities \ineqassmp\ and themselves, as well as the use of nonnegative scalars and arbitrary squared polynomials. In order to simulate in CPS the free use of equations from \assmp\ we incorporate in the set of inequalities \ineqassmp\ the inequalities $f_i\ge 0$ and $-f_i\ge 0$ for each $f_i=0$ in \assmp\ (we show that this enables to add freely products of the polynomial $f_i$ in CPS proofs, namely working n the ideal of \assmp; see \ref{sec:Boolean-CPS-Simulates-Boolean-IPS}).

We need to formalise the concept of a cone as an algebraic circuit. For this we first introduce the  notion of a \emph{squaring gate}: let $C$ be a circuit and $v$ be a node in $C$. We call $v$ a \demph{squaring gate} if $v$ is a product gate whose two incoming edges are emanating from the \emph{same} node. Therefore, if we denote by $w$ the single node that has two outgoing edges to the squaring gate $v$, then $v$ computes $w^2$ (that is, the square of the polynomial computed at node $w$).

The following is a definition of a circuit computing polynomials in the cone of the \vy\  variables:


\begin{definition}[\vy-conic circuit] \label{def:conic-circuit-for-vx} Let $R$ be an ordered ring. We say that an algebraic circuit $C$ computing a polynomial over $R[\vx,\vy]$ is a \demph{conic circuit with respect to \vy}, or \demph{$\vy$-conic}  for short,  if for every negative constant or a variable $x_i\in\vx$, that appears as a leaf $u$ in $C$, the following holds: every path $p$ from $u$ to the output gate of $C$ contains a squaring gate.
\end{definition}

Informally, a \vy-conic circuit is a circuit in which we assume that the \vy-variables are nonnegative, and any other input that may be negative (that is, a negative constant or an \vx-variable) must be part of a squared sub-circuit.
Here are examples of \vy-conic circuits (over \Z):
$y_1$, $~y_1 \cd y_2$, ~$3+2y_1$, ~$(-3)^2$, ~$x_1^2$, ~$(3\cd -x_1+1)^2$, $(x_1y_2+y_1)^2$, ~ $y_1+\dots+y_n$. On the other hand, $-1$, ~$x_1$,~ $x_1\cd y_2$, ~$-1\cd y_1+4$ ~are examples of non \vy-conic circuits.


 Note that if the \vy-variables of a \vy-conic circuit are assumed to take on non-negative values, then a \vy-conic circuit computes only non-negative values. It is evident that \vy-conic circuits can  compute all and only polynomials that are in the cone of the \vy\ variables. In other words, if $\vy$ are the variables $y_1,\dots,y_m$, then there exists a \vy-conic circuit  $C(\vx,\vy)$ that computes the polynomial $p(\vx,\vy)$ iff $p(\vx,\vy)\in\cone(y_1,\dots,y_m)\subseteq R[\vx,\vy]$. Similarly, if $\overline f(\vx)$ is a sequence of polynomials $f_1(\vx),\dots,f_m(\vx)$, then there exists a \vy-conic circuit  $C(\vx,\vy)$ such that $C(\vx,\overline f(\vx))=p(\vx)$ iff $p(\vx)$ computes a polynomial in $\cone\left(\overline f(\vx)\right)\subseteq R[\vx]$.

Deciding if a given circuit is \vy-conic is in  deterministic polynomial-time (see \autoref{cla:verify-conic-circuits} in \autoref{sec:basic-prop-CPS}).







Similar to IPS, we start by defining a boolean version for the Cone Proof System (\ref{def:cps}), which is a refutation system for sets of polynomial equations and  inequalities with no \bits\ solutions. It is easy to define the corresponding \emph{real version} of CPS that refutes sets of polynomial equations and inequalities that are unsatisfiable \emph{over the reals}. This is done simply by taking out the boolean axioms from the system (\ref{def:real-CPS}).

\emph{By default, when referring to CPS we will be speaking about the boolean version}.

\begin{definition}[(boolean)\ Cone Proof System (CPS)]\label{def:cps}
Consider a collection of polynomial equations $\assmp:=\{f_i(\vx)=0\}_{i=1}^m$, and a collection of polynomial inequalities $\ineqassmp:=\{h_i(\vx)\ge  0\}_{i=1}^\ell$, where all polynomials are from $\R[x_1,\ldots,x_n]$. Assume that the following \demph{boolean axioms} are included in the assumptions: \assmp\ includes $x_i^2-x_i=0$, and  $\ineqassmp$ includes the inequalities $x_i\ge0$ and $1-x_i\ge 0$, for every variable $x_i\in\vx$. Suppose further that $\ineqassmp$ includes (among possibly other inequalities) the two inequalities $f_i(\vx)\ge0$ and $-f_i(\vx)\ge0$ for every equation $f_i(\vx)=0$ in $\assmp$ (including the equations $x_i^2-x_i=0$).
A \demph{CPS proof of $p(\vx)$ from $ \assmp$ and $\ineqassmp$}, showing    that $\assmp,\ineqassmp$ semantically imply the polynomial inequality $p(\vx)\ge 0$
over $0$-$1$ assignments, is an algebraic circuit $C(\vx,\vy)$ computing a polynomial in $\R[\vx,y_1,\dots,y_\ell]$, such that:\footnote{Note that formally we do not make use of the assumptions \assmp\ in CPS, as we assume always that the inequalities that correspond to the equalities in \assmp\ are present in \ineqassmp. Thus, the indication of \assmp\ is done merely to maintain clarity and distinguish (semantically) between two kinds of assumptions: equalities and inequalities.}\vspace{-5pt} \begin{enumerate}
 \item $C(\vx,\vy)$ is a \vy-conic circuit;
and\vspace{-5pt}
\item $C(\vx,\ineqassmp)=p(\vx)$,
\end{enumerate}\vspace{-5pt}
where equality 2 above is a formal polynomial identity\footnote{That is, $C(\vx,\ineqassmp)$ computes the polynomial $p(\vx)$.} in which the left hand side means that we substitute $h_i(\vx)$ for $y_i$, for $i=0,\dots,\ell$.

The \demph{size} of a CPS proof is the size of the circuit $C$. The variables $\vy$ are the \emph{placeholder} \emph{variables} since they are used as a placeholder for the axioms. A CPS proof of $-1$ from $\assmp,\ineqassmp$ is called a \demph{CPS refutation  of \assmp, \ineqassmp}.
\end{definition}
In what follows, we will write ``conic'' instead of ``\vy-conic'' where the meaning of \vy\ is clear from the context.

In order to refute propositional formulas in conjunctive normal form (CNF) in CPS we use
the algebraic translation of CNFs (\autoref{def:algebraic-transl-CNF}), which is expressed as a set of polynomial equalities. We show in \ref{prop:CPS-from-equtional-CNF-to-inequalities-CNF} that CPS can efficiently translate CNF formulas written as polynomial equalities to the standard way in which CNF formulas are written as polynomial \emph{inequalities}.
The real version of CPS is defined as follows:

\begin{definition}[Real CPS]\label{def:real-CPS}
The \emph{real CPS} system is defined similarly to boolean CPS except that boolean axioms are \emph{not} included in the assumptions. That is, \assmp\ does \emph{not} include $x_i^2-x_i=0$, and  $\ineqassmp$ does \emph{not} include the inequalities $x_i\ge0$ and $1-x_i\ge 0$ (for  variables $x_i\in\vx$).
\end{definition}

\para{Remark about CPS.}

\begin{enumerate}
\item
CPS should be thought of as a way to derive valid polynomial inequalities from a set of polynomial equations and inequalities from $\R[\vx]$. Loosely speaking, it is a circuit representation of the \ps\ proof system (\autoref{def:PS}), though in CPS the assumptions \assmp, \ineqassmp\ (more precisely, placeholder variables of which) may have powers greater than one. That is, whereas \ref{eq:ps} is \emph{multilinear} in the $h_i$ variables, CPS is not.
\iddo{Make sure you reference the propositional version of SoS. As otherwise, the general algebraic version of SoS does not have the $x_i\ge 0 $ assumptions in the inequalities; hence, it seems slight weaker or different than CPS}
\hirsch{SoS is not referenced right here. The simulation does mention boolean (not propositional) version.}
\item We add the boolean axioms $x_i^2-x_i\ge 0$, $x_i-x_i^2\ge 0$, $x_i\ge0$ and $1-x_i\ge 0$ to \ineqassmp\ as a default.
Hence, the system can refute any set of inequalities (and equalities) that is unsatisfiable over 0-1 assignments.

\item Formally, CPS proves only consequences from an initial set of inequalities \ineqassmp\ and not equalities \assmp. However, we are not losing any power doing this. First, observe that:
\begin{quote}An assignment satisfies  \assmp, \ineqassmp\ iff it satisfies \ineqassmp\ (in the case of boolean CPS an assignment that satisfies either $\assmp$ or $\ineqassmp$ must be a 0-1 assignment). \end{quote}Second, we encode equalities $f_i(\vx)=0\in\assmp$ using the two inequalities $f_i(\vx)\ge 0$ and $-f_i(\vx)\ge 0$ in \ineqassmp. As shown in \autoref{thm:CPS-sim-IPS}  \emph{this way we can derive any polynomial in the \emph{ideal} of \assmp, and not merely in the cone of \assmp,} as is required for equations (and similar to the definition of SoS), with at most a polynomial increase in size (when compared to IPS).

To derive polynomials in the ideal of \assmp\ we need to be able to  multiply  $f_i$ and $-f_i$ (from \ineqassmp) by \emph{any} (positive) polynomial in the \vx\ variables. There are two ways to achieve this in boolean CPS: the first, is to use the boolean axiom $x_i\ge 0$ in \ineqassmp. This allows to product $f_i$ and $-f_i$ by any polynomial in the \vx-variables. The second way, the one we use in \ref{prop:single-minus-gate-at-top} to show that CPS simulates IPS in \ref{thm:CPS-sim-IPS}, is  different and does not necessitate the addition of the axiom $x_i\ge 0$ to \ineqassmp. Since the second way does not use the boolean axiom  $x_i\ge 0$ in \ineqassmp\ we can use it in  real CPS, hence allowing the derivation of polynomials in the ideal of \assmp\ within real CPS.




\end{enumerate}



\bigskip

To exemplify a proof in CPS we provide the following simple proposition:

\begin{proposition}\label{prop:CPS-proof-of-BVP}
CPS admits a linear size refutation of the binary value principle \bvpn.
\end{proposition}
\begin{proof}
To simplify notation we put $S:=\sum_{i=1}^n 2^{i-1}\cd x_i +1$.
Let \assmp$:=\left\{ S =0,x_1^2-x_1=0,\dots,x_n^2-x_n=0 \right\}$. Then by the definition of CPS \ineqassmp\ contains the following correspondent $4n+2$ axioms ($4n$ boolean axioms, and two axioms for the single non-boolean axiom in \assmp):
\begin{multline*}
\ineqassmp:=\left\{x_1\ge 0, \dots, ~x_n\ge 0, ~-S\ge 0, ~S \ge 0, ~x_1^2-x_1 \ge 0,\dots,x_n^2-x_n\ge 0,  \right. \\
\left. -( x_1^2-x_1)\ge 0,\dots,~-(x_n^2-x_n)\ge 0, ~ 1-x_1\ge 0, \dots, ~1-x_n\ge 0\right\}.
\end{multline*}
Therefore, the CPS refutation of the binary value principle is defined as the following \emph{\vy-conic} circuit:
\begin{equation}\label{eq:CPS-bvp-proof}
C(\vx,\vy):= \left(\sum_{i=1}^n 2^{i-1}\cd y_i\right)  + y_{n+1} ,
\end{equation}
where the placeholder variables $y_1,y_2,\dots,y_{4n+2}$ correspond to the axioms in \ineqassmp\ in the order they appear above. Observe indeed that $C(\vx,\ineqassmp)= C(\vx,x_1,\dots,x_n,-S,\dots)=\left(\sum_{i=1}^n 2^{i-1}\cd x_i\right)  + (-S) =-1$.
\end{proof}

Observing the CPS refutation  in \ref{eq:CPS-bvp-proof} we see that it is in fact already an \sos\ refutation:
\begin{corollary}\label{cor:sos-short-refutation-of-BVP}
\sos\ admits a linear monomial size refutation of the binary value principle \bvpn.  \end{corollary}



\subsection{Basic Properties of CPS and Simulations}
\label{sec:basic-prop-CPS}
CPS is a very strong proof system. In fact, of all proof systems with randomized polynomial-time verification, given concretely, \iddo{What about WF of Jerabek?}\hirsch{I guess IPS simulates it. I do not know if it is written anywhere.}\iddo{unsure..since WF has non-tautological axioms apparently...}\hirsch{I was thinking about WF as EF operating with circuits; the particular rules of showing circuits ``similarity'' are not important since IPS allows to rewrite circuits for the same polynomial.} \iddo{no, WF is something else. You meant CF}to the best of our knowledge CPS is the strongest to have been defined   to this date.
CPS simulates IPS as shown below, while we show that IPS simulates CPS only under the condition that there are short IPS refutations of the binary value principle.

\begin{proposition}[CPS is sound and complete]\label{prop:CPS-complete-and-sound} Let $R$ be an ordered ring.
CPS (resp.,~real CPS) is a complete and sound proof system for the language of sets of polynomial equations and inequalities over $R$ that have no 0-1 (resp.~$R$-solutions) solutions. More precisely, given two sets of polynomial equalities and inequalities \assmp, \ineqassmp, respectively, where all polynomials are from $R[x_1,\dots,x_n]$, there exists a CPS (resp.~real CPS) refutation of \assmp, \ineqassmp, iff there is no \bits\  assignment (resp.~$R$-assignment) satisfying both \assmp\ and \ineqassmp\ (iff there is no \bits\  assignment (resp.~$R$-assignment) satisfying \ineqassmp).
\end{proposition}
\begin{proof}
The completeness of  boolean CPS follows from the simulation of propositional \ps\ below (\autoref{thm:CPS-simulates-PS}). The soundness of boolean CPS stems from the following. Assume that $C(\vx,\vy)$ is a CPS refutation of \assmp, \ineqassmp. Recall that an assignment satisfies  \assmp, \ineqassmp\ iff it satisfies \ineqassmp.  Assume by a way of contradiction that  \valpha\ is  a 0-1 assignment to $\vx$ that satisfies \assmp, \ineqassmp. The circuit $C(\vx,\vy)$ is \vy-conic and hence   $C(\valpha,\ineqassmp(\valpha))$ is non-negative assuming that the inputs to the \vy\ variables (that is, $\ineqassmp(\valpha)$) are non-negative. Since $\valpha$ satisfies \ineqassmp\ we know that indeed $h_i(\valpha)\ge 0$, for every $h_i(\vx)\in\ineqassmp$.
Therefore, $\widehat C(\valpha,\ineqassmp(\valpha))\ge 0$, which contradicts our assumption that $C(\valpha,\ineqassmp(\valpha))=-1$.

The completeness for real CPS follows by similar arguments.
\end{proof}

 \begin{proposition}\label{prop:CPS-verified-in-RP}
A CPS proof (either real of boolean) can be checked for correctness in probabilistic polynomial-time.
\end{proposition}
\begin{proof}
Similar to IPS, we can verify condition 2 in \autoref{def:cps}, that is $C(\vx,\ineqassmp)=p(\vx)$, in probabilistic polynomial-time (formally, in  $\mathsf{coRP}$). For condition 1 we need to check that $C$ is a \vy-conic circuit, which can be done in \P\ via the following claim:
\begin{claim}\label{cla:verify-conic-circuits}
There is a polynomial-time algorithm to determine if a circuit $C(\vx,\vy)$ is a \vy-conic circuit or not. \end{claim}

\begin{proofclaim}
We say that a directed path from a leaf $u$ in $C$ holding either a negative constant or an \vx\ variable to the output gate of $C$ is \emph{bad} if the path does not contain any squaring gate.

For each leaf $u$ in $C$ holding either a negative constant or an \vx\ variable we can determine the following property in \NL: there \emph{exists} a bad path from $u$ to the output gate of $C$. This algorithm is in \NL\ simply because nondeterministically we can go along a directed path from $u$ to the output gate and check that no squaring gate was encountered along the way (we only need to record the current node and the current length of the path so to know when to terminate). This means that the complement problem of deciding that there does \emph{not} exist a bad path from $u$ to the output gate is in $\coNL$ which is known to be contained in \P.

Our algorithm thus checks that each of the leaves holding negative constants do not possess any bad path to the output gate, which can be done in polynomial-time by the argument above.
%
\end{proofclaim}
\end{proof}

As a corollary of \autoref{prop:CPS-complete-and-sound} and \autoref{prop:CPS-verified-in-RP} we get that, similar to IPS,  if CPS is p-bounded (namely, admits polynomial-size refutations for every unsatisfiable CNF formula)
then coNP is in MA, yielding in particular the polynomial hierarchy collapse to the third level (cf.~\cite{Pit98,GP14}).

\medskip

\subsubsection{CPS Simulates IPS}\label{sec:Boolean-CPS-Simulates-Boolean-IPS}
We now show that boolean  CPS simulates boolean IPS
for the language of \bits-unsatisfiable sets of polynomial \emph{equations} over any ordered ring. Similarly, real CPS simulates algebraic IPS over \Q. We translate an input equality $f_i(\vx)=0$
into a pair of inequalities $f_i(\vx)\ge 0 $ and $-f_i(\vx)\ge 0$,
Note that an IPS proof is written as a general algebraic circuit (computing an element of an ideal),
while a CPS proof is written as a more restrictive algebraic circuit, namely as a \vy-conic circuit
(computing an element of a cone).
This means that a priori we cannot (obviously) multiply an inequality by an arbitrary polynomial in CPS. We thus demonstrate how to do it when we have opposite-sign inequalities.
In order to do this, we represent an arbitrary polynomial
as the difference of two nonnegative expressions.

\begin{proposition}[minus gate normalisation]\label{prop:single-minus-gate-at-top}
Let $G(\vx)$ be an algebraic circuit computing a polynomial in the $\vx$ variables over \Q. Then, there is an algebraic circuit of the form $G_P(\vx)-G_N(\vx)$ computing the same polynomial as $G(\vx)$ where $G_P$ and $G_N$ are $\emptyset$-conic.
The size of $G_P$, $G_N$ is at most linear in the size of $G$.
\end{proposition}

\begin{proof}
This is somewhat reminiscent of Strassen's conversion
of a circuit with division gate to a circuit with only a single division gate at the top \cite{Str73}. We are going to break inductively each node into a pair of nodes computing the positive and negative parts of the polynomial computed in that node.
%
%
Formally, we define the circuits  $G_P,G_N$ (that may have common nodes) by induction on the size of $G$ as follows:

\case 1 $G=x_i$, for $x_i \in \vx$. Then, $G_P:=\frac12(x_i^2+1), G_N:=\frac12(x_i-1)^2$.

\case 2 $G=\alpha$, for $\alpha$ a constant in the ring. Then \vspace{-5pt}
\[
\begin{matrix}
G_P:=\alpha,~G_N:=0, & \text{if } \alpha\ge 0; \\
G_P:=0,~G_N:=\alpha, & \text{if } \alpha<0.
\end{matrix}
\]

\case 3 $G=F+H$. Then, $G_P:=F_P+H_P$ and $G_N:=F_N+H_N$.

\case 4 $G=F\cd H$. Then, $G_P:=F_P\cd G_P + F_N\cd G_N$ and $G_N:=F_P\cd G_N + F_N\cd G_P$.

The size of both $G_P, G_N$ is $O(|G|)$, namely linear in the size of $G$. This is because we only add constantly many new nodes in $G_P , G_N$ for any original node in $G$; note that since we  construct a new \emph{circuit} computing the same polynomial as $G$, we can re-use nodes computed already, in case 4: for example, $F_P$ is the same node used in $G_P$ and $G_N$ (hence, indeed, the number of new added nodes for every original node in $G$ is constant). \end{proof}

\begin{theorem}\label{thm:CPS-sim-IPS}
Real CPS simulates algebraic IPS as a proof system for the language of unsatisfiable sets of polynomial equations over $\mathbb{Q}$. In other words, there exists a constant $c$ such that for any polynomial $p(\vx)$ and a set of polynomial equations \assmp,
if $p(\vx)$ has an IPS proof of size $s$ from \assmp\ then there is a CPS proof of $p(\vx)$ from \assmp\ of size at most $s^c$. Furthermore, boolean CPS simulates boolean IPS (for any ordered ring).\end{theorem}
\begin{remark}
It is easy to see that fractional $\mathbb{Q}$ coefficients are not needed in the case
of boolean systems, as Case 1 in \autoref{prop:single-minus-gate-at-top} above
simplifies to $G_P:=x_i, G_N:=0$ when $x_i$'s are nonnegative. This is the reason boolean CPS simulates boolean IPS over any ordered ring.
\end{remark}

Specifically, if \assmp\ is a set of polynomial equations with no 0-1 satisfying assignments and suppose that there is an  IPS refutation of \assmp\ with size $s$, then there is a CPS refutation of \assmp\ with size at most $s^c$. 

\begin{proof}[Proof of \autoref{thm:CPS-sim-IPS}]

%
%
We are going to simulate both the boolean and the algebraic versions of IPS. The proof in both cases is the same.

 Assume that $C(\vx,\vy)$ is the IPS proof of $p(\vx)$ from $\assmp=\{f_i(\vx)=0\}_{i=1}^{\ell} $, of size $s$, and let $\vy=\{y_1,\dots,y_\ell\}$ be the placeholder variables for the equations in \assmp. We assume for simplicity that if we simulate the \emph{boolean }version of IPS  the boolean axioms \ba \ are  also  part of \assmp\ (while if we simulate the \emph{algebraic} version of IPS these axioms are not part of \assmp).\mar{Put in Prelim: We sometimes write $C(\vx)$ displaying the input variables of a circuit and sometimes we can suppress the explicit mention of the variables.}
%
%
We  use the following claim:
\begin{claim}\label{cla:simple-break}
Let $C(\vx,\vy)$ be a circuit of size $s$, where $\vy=\{y_1,\dots,y_\ell\}$ and such that $C(\vx,\vzero)=0$. Then $C$ can be written as a sum of circuits with only a polynomial increase in size as follows: $C(\vx,\vy) = \sum_{i=1}^\ell y_i\cd C_i(\vx,\vy)$.
\end{claim}

\begin{proofclaim}
We proceed by a standard process to factor out the \vy\ variables one by one. Beginning with $y_1$ we get:
$$
C(\vy,\vx)=y_1\cd \left(C(1,\vy',\vx)-C(0,\vy',\vx)\right)+C(0,\vy',\vx),
$$
where $\vy'$ denotes the vector of variables $(y_2,\dots,y_\ell)$. In a similar manner we factor out the variable $y_2$ from $C(0,\vy',\vx)$. Continuing in a similar fashion we conclude the claim. Notice that the size of the resulting circuit is $O(|C|^2)$, and that in the final iteration of the construction we factor out $y_\ell$ from $C(\overline 0, y_\ell,\vx)$ it must hold that $C(\overline 0, y_\ell,\vx)=y_1\cd \left(C(\vzero,1,\vx)-C(\vzero,0,\vx)\right)+C(\overline 0,0,\vx)=y_1\cd C(\vzero,1,\vx)$, because by assumption $C(\overline 0,0,\vx)=0$.
\end{proofclaim}


By this claim we have
\begin{align}\notag
C(\vx,\vy) & = \sum_{i=1}^\ell y_i\cd C_i(\vx,\vy)\\
& = \sum_{i=1}^\ell y_i\cd C_{i,P}(\vx,\vy)  - \sum_{i=1}^\ell y_i\cd C_{i,N}(\vx,\vy)\,, \label{eq:break-up-C}
\end{align}
where $C_{i,P}(\vx,\vy),C_{i,N}(\vx,\vy)$ are the positive and negative parts of $C_i(\vx,\vy)$, respectively, that exist by \autoref{prop:single-minus-gate-at-top}, written as circuits in which no negative constants occur (we do not need to distinguish between the variables \vx\ and \vy\ here).

We wish to construct now a CPS refutation of \assmp. Our corresponding set of inequalities \ineqassmp\ will consist of  $f_i(\vx)\ge 0,-f_i(\vx)\ge 0$, for every $i\in[\ell]$.
In total, $|\ineqassmp|=2\ell$. Accordingly, our CPS refutation of \assmp, \ineqassmp,
will have $2\ell$ placeholder variables for the axioms in \ineqassmp\ denoted as follows:
$\vy_P$ are the $\ell$ placeholder variables $y_{i,P}$ corresponding to  $f_i(\vx)\ge 0$, $i\in[l]$,
$\vy_N$ are the $\ell$ placeholder variables $y_{i,N}$ corresponding to  $-f_i(\vx)\ge 0$, $i\in[l]$.

Since $C_{i,P}$ and $C_{i,N}$ are $\emptyset$-conic circuits,
$$
\sum_{i=1}^\ell y_{i,P}\cd C_{i,P}(\vx,\vy_P,\vy_N)  + \sum_{i=1}^\ell y_{i,N}\cd C_{i,N}(\vx,\vy_P,\vy_N)
$$
is a $(\vy_P,\vy_N)$-conic circuit.
It constitutes
a  CPS proof of $p(\vx)$ from the assumptions $f_i(\vx)\ge 0,-f_i(\vx)\ge 0$, for $i\in[\ell]$
of size linear in $|C|$ (as before, we denote by \assmp\ the vector $f_1(\vx),\dots,f_\ell(\vx)$):
\begin{align*}
\sum_{i=1}^\ell f_i(\vx)\cd C_{i,P}(\vx,\assmp) & + \sum_{i=1}^\ell (-f_i(\vx))\cd C_{i,N}(\vx,\assmp) \\ & = \sum_{i=1}^\ell f_i(\vx)\cd \left( C_{i,P}(\vx,\assmp)  - C_{i,N}(\vx,\assmp)\right)\\  & = \sum_{i=1}^\ell f_i(\vx)\cd  C_{i}(\vx,\assmp)=   C(\vx,\assmp) = p(\vx).
\end{align*}
\end{proof}

\subsubsection{CPS Simulates \ps\ and SoS}

The following theorem is immediate from the definitions.

\begin{theorem}\label{thm:CPS-simulates-PS}
Real CPS simulates  \ps\ (and hence also \sos) proof system over the same ordered ring.
\hirsch{Actually, we should we very cautious about simulation for abstract rings,
as \emph{\underline{the size is the number of bits}}. However, this is not the main result,
and it is a very straightforward one, so I state it for rings.}\iddo{But we define the size of CPS in the unit-cost model. While your comment is relevant only to the constant-free CPS model, no?}
\end{theorem}
\begin{proof}
This follows immediately from the fact that CPS is a circuit representation of the second big sum in \ref{eq:ps}. More formally, let $\assmp:=\{f_i(\vx)=0\}_{i\in I}$ be a set of polynomial equations and let  $\ineqassmp:=\{h_j(\vx)\ge 0\}_{j\in J}$ be a set of polynomial inequalities, where all polynomials are from $\R[x_1,\ldots,x_n]$. Consider the following Positivstellensatz refutation of \assmp, \ineqassmp, where $\{p_i\}_{i\in I}$ and $\{s_{i,\zeta}\}_{i,\zeta}$ (for $i\in\N$ and $\zeta\subseteq J$) are collections of polynomials  in $\R[x_1,\ldots,x_n]$:
\begin{equation}\label{eq:PS-proof-in-CPS-simulate-PS}
\sum_{i\in I} p_i\cd f_i + \sum_{\zeta\subseteq J} \left(\prod_{j\in \zeta} h_j \cd \left(\sum_{i\in I_\zeta} s_{i,\zeta}^2\right)\right)=- 1\,.
\end{equation}
The size of the \ps\ refutation is  the combined total number of monomials in $\{p_i\}_{i\in I}$ and $\sum_{i\in I_\zeta} s_{i,\zeta}^2$, for all  $\zeta\subseteq J$ (see \autoref{def:PS}).


By definition every $f_i(\vx)=0\in\assmp$ has corresponding two inequalities in \ineqassmp, $f_i(\vx)\ge 0$ and $-f_i(\vx)\ge 0$. Let  the  variables \vy\ (to be used as placeholder variables) be partitioned into three disjoint parts: $\vy=\{y_i\}_{i\in I}\uplus\{y_{i'}\}_{i'\in I'}\uplus\{y_j\}_{j\in J}$, where $\{y_i\}_{i\in I}$ are the placeholder variables for $\{f_i(\vx)\ge 0\}_{i\in I}$ in \ineqassmp, $\{y_{i'}\}_{i'\in I'}$ are the placeholder variables for $\{-f_i(\vx)\ge 0\}_{i\in I}$ in \ineqassmp\ and $\{y_j\}_{j\in J}$ are the placeholder variables for $\{h_j(\vx)\ge 0\}_{j\in J}$ in \ineqassmp.
Assume also that for every $i\in I$, $p_{i,P}$ is the sum of all non-negative monomials in $p_i$ and $p_{i,N}$ is the sum of all negative monomials in $p_i$. Define
$$
C(\vx,\vy):=\sum_{i\in I} p_{i,P}\cd y_i+\sum_{i'\in I'} p_{i,N}\cd y_i  +\sum_{\zeta\subseteq J} \left(\prod_{j\in \zeta} y_j \cd \left(\sum_i s_{i,\zeta}^2\right)\right),
$$
where each of the three big sums is written as a sum of monomials.


Hence,  $C(\vx,\ineqassmp)= -1$ by \ref{eq:PS-proof-in-CPS-simulate-PS} and  the size of $C(\vx,\vy)$ is linear in  $\sum_{i\in I}\monomsize{p_i } +\sum_{\zeta\subseteq J}\sum_{i\in I_\zeta} \monomsize{s_{i,\zeta}^2}$.
\end{proof}
\begin{corollary}
Boolean CPS simulates \sos\ and Positivstellensatz
for inputs that include the boolean axioms.
\end{corollary}

\mar{\iddo{there was a problem here...PS is too weak as it doesn't have the xi 1-xi products....So need to make it like: algebraic CPS simulates PS******\\ Solution I implemented: propositional CPS simulates propositional (and also algebraic) PS. Algebraic CPS simulates algebraic IPS.}
\hirsch{it was not in the formulation, and I have changed the formulation now, and added a corollary}
}



\subsubsection{CPS Simulates \LSInfty\ for CNFs Written as Inequalities}

CPS can simulate the strongest semi-algebraic proof system as defined in \autoref{def:LSInfty}.

\begin{theorem}\label{thm:Boolean-CPS-simulates-LSInfty}
Boolean CPS simulates \LSInfty\ (that is, ``dynamic \ps" from \autoref{def:LSInfty}).
\end{theorem}

Recall that CPS uses the algebraic translation of CNFs (\autoref{def:algebraic-transl-CNF}) as equations
while earlier semi-algebraic systems historically used the semi-algebraic translation of CNFs (\autoref{def:semi-algebraic-transl-CNF}) as inequalities. We will show below that one can be efficiently converted into the other. Modulo this proposition the proof of \autoref{thm:Boolean-CPS-simulates-LSInfty} is almost trivial.

\begin{proof}[Proof sketch of \autoref{thm:Boolean-CPS-simulates-LSInfty}]
It suffices to observe that the derivation rules (adding and multiplying two inequalities, taking a square of an arbitrary polynomial) are the same as the rules of constructing the conic circuit. Therefore, following the $\textrm{LS}^\infty_{*,+}$ proof we construct a conic circuit that, given the axioms on the input, computes -1.
\end{proof}





\begin{proposition}\label{prop:CPS-from-equtional-CNF-to-inequalities-CNF}
There is a polynomial-size propositional CPS proof that starts from the algebraic translation of a  clause as the two inequalities   $\prod_{i\in P} (1-x_i)\cdot \prod_{j\in N} x_j\ge 0$ and  $-\left( \prod_{i\in P} (1-x_i)\cdot \prod_{j\in N} x_j\right)\ge 0$, and derives the semi-algebraic translation of the clause  $\sum_{i\in P} x_i + \sum_{j\in N} (1-x_j)-1\ge 0$.
\end{proposition}

Recall that CPS works  with inequalities, whereas equalities $f=0$ in \assmp\ are interpreted as the two inequalities $f\ge 0$ and $-f\ge 0$ in \ineqassmp. Hence,  \autoref{prop:CPS-from-equtional-CNF-to-inequalities-CNF} suffices to show that a clause given as an equality in \assmp\ can be translated efficiently in CPS to its semi-algebraic translation.

\begin{proof}[Proof of \autoref{prop:CPS-from-equtional-CNF-to-inequalities-CNF}]
We proceed by induction on the number of variables in the clause.

\Base We start with one of the (algebraic) clauses   $x_1$ or $1-x_1$. In the former  case, we start from $-x_1$ which is in $\ineqassmp$ by the definition of CPS, and we need to derive   $(1-x_1)-1$, which is equal to $-x_1$, hence we are done. In the latter case, we start from $-(1-x_1)$ which is $x_1-1$, hence we are done again.

\induction

\case 1    We start from the clause $(1-x_n)\cd\prod_{i\in P}(1-x_i)\cd\prod_{i\in N}x_i$ as a given equation (namely, in \assmp; formally, the two corresponding inequalities are in \ineqassmp), and we need to derive $x_n+\sum_{i\in P}x_i+\sum_{i\in N}(1-x_i)-1$ in CPS. We consider the two cases $x_n=0$ and $x_n=1$, and then use reasoning by boolean cases in CPS. Reasoning by boolean cases in propositional CPS is doable in polynomial-size by \autoref{prop:IPS-cases} which states this for IPS and since propositional CPS simulates IPS by \autoref{thm:CPS-sim-IPS} for the language of polynomial equations \assmp\ (in our case the initial clauses are indeed given as equations, and thus CPS  simulates IPS' reasoning by boolean cases).

In case $x_n=0$, $(1-x_n)\cd\prod_{i\in P}(1-x_i)\cd\prod_{i\in N}x_i = \prod_{i\in P}(1-x_i)\cd\prod_{i\in N}x_i$, from which, by induction hypothesis we can derive in CPS with a polynomial-size proof $\sum_{i\in P}x_i+\sum_{i\in N}(1-x_i)-1$.
Since $x_n=0$ we can add $x_n$ to this expression obtaining $x_n+\sum_{i\in P}x_i+\sum_{i\in N}(1-x_i)-1$, and we are done.

In case $x_n=1$, we have $x_n+\sum_{i\in P}x_i+\sum_{i\in N}(1-x_i)-1 =1+ \sum_{i\in P}x_i+\sum_{i\in N}(1-x_i)-1=\sum_{i\in P}x_i+\sum_{i\in N}(1-x_i)$. But $\sum_{i\in P}x_i+\sum_{i\in N}(1-x_i)$ is easily provable in propositional CPS because we have the axioms $x_i\ge 0$ and $1-x_i\ge 0$ in \ineqassmp, for every variable $x_i$, by definition.

\case 2 We start from the clause $x_n\cd\prod_{i\in P}(1-x_i)\cd\prod_{i\in N}x_i$, and we need to derive $(1-x_n)+\sum_{i\in P}x_i+\sum_{i\in N}(1-x_i)-1$ in CPS. This is similar to Case 1 above with the two boolean sub-cases $x_n=0$ and $x_n=1$ flipped.
\end{proof}

\mar{Non-capitals in [...] in theorems and definitions!}

\section{Reasoning about Bits within Algebraic Proofs}
\label{sec:Reasoning-about-Bits-within-Algebraic-Proofs}

%


\bigskip

In what follows we define a number of circuits implementing arithmetic in the two's complement notation (see below for the details).
Namely, we will define the following polynomial-size circuits:

\begin{description}

\item[$\biti{i}(f)$:] if $f$ is a circuit in the variables $\vx$ then $\biti{i}(f)$ computes the $i$th bit of the integer computed by $f$ (as a function of the input variables $\vx$ where the variables $\vx$ range over 0-1 values).

\item[$\bitv(f)$:] a multi-output operation that computes the bit vector of $f$ (as a function of the input variables $\vx$ where the variables $\vx$ range over 0-1 values).

\item[$\addv(\vy,\vz)$:] a multi-output carry-lookahead circuit that computes the bit vector of the sum of $\vy$ and $\vz$.

\item[$\addi{i}(\vy,\vz)$:] the circuit that computes the $i$th output bit in the carry-lookahead circuit $\addv(\vy,\vz)$.

\item[$\cari{i}(\vy,\vz)$:] the carry for bit $i$ when adding two bit vectors $\vy,\vz$.

\item[$\prdv(\vy,\vz)$:]
the multi-output circuit computing binary multiplication of two bit vectors $\vy$ and $\vz$.

\item[$\prdvp(\vy,\vz)$:]
the multi-output circuit computing binary multiplication of two nonnegative bit vectors $\vy$ and $\vz$.

\item[$\val(\vz)$:] the valuation function that converts $\vz$ encoding an integer in the two's complement representation to its integer value (see below).
\item[$\iabsv(\vx)$:] The multi-output circuit computing the two's complement binary representation of the absolute value of an input integer $\vx$ given in two's complement.
\end{description}

We construct the \biti{i} function by induction on the size of $f$. In general this cannot be done for algebraic variables, but in our case we are assuming that the variables $x_1,\ldots, x_n$ are boolean variables, and this allows us to carry out the constructions below, yielding a circuit of size which is polynomial in the size of the algebraic circuit of $f$ where ring scalars are encoded in binary.

\subsection{Basic Two's Complement Arithmetic}\label{sec:basic-twos-complement}
Integer numbers are encoded in the \textit{two's complement} scheme since this scheme makes standard binary addition work for both positive and negative numbers, which simplifies the construction slightly.
In the two's complement scheme the value represented by the bit string $\overline{w}\in\bits^k$ is determined by a function from $\bits^k$ to \Z\ as follows:
\mar{change k to t}\begin{definition}[the binary value operation \val]\label{def:val}
Given a bit vector $w_0\cdots w_{k-1}$ of variables, denoted $\overline w$, ranging over $0$-$1$ values,
define the following algebraic circuit with integer coefficients\footnote{We assumed that algebraic circuits have fan-in two, hence \val\ is written as a logarithmic depth circuit of addition gates (and product gates at the bottom of the circuit).}:
$$
\val(\overline w):=\sum_{i=0}^{k-2} 2^i\cd w_i\,-2^{k-1}\cd w_{k-1}.
$$
The most significant bit (msb) $w_{k-1}$ is called the \demph{sign bit} of $\overline w$.
\end{definition}
Thus, $\val(\overline w)=\sum_{i=0}^{k-2} 2^i\cd w_i$ in case the sign bit $w_{k-1}=0$ (hence, $\overline w$ encodes a positive integer), and  $\val(\overline w)=\sum_{i=0}^{k-2} 2^i\cd w_i\,-2^{k-1}$, in case the sign bit $w_{k-1}=1$ (hence, $\overline w$ encodes a negative integer).

\medskip
We will represent the integers computed inside the original algebraic circuit by \textit{variable length} bit vectors  (that is, different bit vectors may have different lengths). For each gate in a given circuit we will assign  a number that is sufficiently large to store the bit vector of the integer it computes when the input variables range over 0-1 values; this number will be called the \demph{syntactic length} of the gate (or equivalently, of the circuit whose output gate is this gate). The syntactic length of  a gate is not necessarily the minimal number of bits needed to store a number, since we will find it convenient to use  slightly more bits than is actually required at times. For instance, the product of  two $t$-bit binary numbers can be stored with only $2t-1$ bits, but we will use $2t+3$ bits for a product. Given the syntactic length of algebraic gates such as $+,\times$ as functions of the syntactic length of their incoming edges, we can compute by induction on circuit size the syntactic length of any given gate in a circuit.
\iddo{We need to somehow compute or think how to speak about syntactic length in more accuracy: \textbf{the most important lemma we need is a lemma that bounds from above the syntactic length of a circuit (constant-free or not) given parameters such as size, degree, etc}. Then we need to think of it over \Q\ as I've written in \autoref{def:syntactic-length-over-Q} and the theorems that follow this definition.}
\hirsch{For now, I added a sentence below. We can formulate it as a lemma, but I would not spend more than two lines for its proof.
I have not thought about the $\mathbb{Q}$-version yet.}
It will be straightforward that the syntactic length of a constant-free (integer algebraic) circuit that has $s$ gates and multiplicative depth $D$
(that is, the longest directed path goes through at most $D$ multiplications) is at most $O(s2^D)$ (imagine repeated squaring of $2$ as the worst case),
and it is at most $O(sd)$ for a constant-free circuit that has syntactic degree $d$ (that is, it would compute
a polynomial of total degree at most $d$ if all constants $-1$ are replaced by $1$; surely, $d\le 2^D$).


When we need to make an operation over integers of different bit-length, \emph{we pad the shorter one }(in the two's complement scheme, a number is always padded by its sign bit, and it is immediate  to see that such padding preserves the value of the number as computed by \val).

We will use the boolean connectives $\land,\lor,\oplus$, which stand for AND, OR and XOR, respectively. In order to use boolean connectives inside algebraic circuits, we define the arithmetization of connectives as follows:

\begin{definition}[arithmetization operation $\arit{\cd}$]\label{def:arithmetization}
 For a variable $y_i$, $\arit{y_i}:=y_i$. For the truth values false $\false$ and true $\true$
we put $\arit{\textsf{\false}}:=0$ and $\arit{\textsf{\true}}:=1$.
For logical connectives we define
$\arit{A\land B}:=\arit{A}\cdot \arit{B}$,
$\arit{A\lor B}:=1-(1-\arit{A})\cdot(1-\arit{B})$, and for the {\rm XOR} operation we define
$\arit{A\oplus B}:= \arit{A}+\arit{B}-2\cdot\arit{A}\cdot\arit{B}$.
\end{definition}
In this way, for every boolean circuit $F(\vx)$ with $n$ variables and a boolean assignment $\valpha\in\bits^n$,  $\arit F(\overline \alpha) = 1$ iff $F(\valpha)=\mathsf{true}$.

In what follows, we sometimes omit $\arit{\cd}$ from our formulas and simply write $\land$, $\lor$, $\oplus$
meaning the corresponding polynomials or algebraic circuits.

Addition is done with a carry lookahead adder as follows:

\begin{definition}[$\cari i$, $\addi i $, \addv]\label{def:carryadd}
When we use an adder for vectors of different size,
we pad the extra bits of the shorter one by its sign bit. Suppose that we have a pair of length-$t$ vectors of variables $\vy=(y_0,\dots,y_{t-1}),\vz=(z_0,\dots,z_{t-1})$ of the same size. We first  pad the two vectors by a single additional bit $y_k=y_{k-1}$ and $z_t=z_{t-1}$, respectively (this is the  way to deal with  a possible overflow occurring  while adding the two vectors). Define
\begin{eqnarray}\label{eq:def-carryi}
\cari{i}(\vy,\vz)&:=&
\begin{cases}
(y_{i-1}\land z_{i-1})\lor((y_{i-1}\lor z_{i-1})\land\cari{i-1}(\vy,\vz))
              ,               & i = 1,\ldots,t;\\
              0\,, & i=0\,,
\end{cases}
\end{eqnarray}
and
$$
\addi{i}(\vy,\vz):=
              y_i\oplus z_i \oplus \cari{i}(\vy,\vz)\,,                ~i = 0,\ldots,t.
$$
%
%
Finally, define
\begin{equation}\label{eq:def-ADD}
\addv(\vy,\vz):=\addi{t}(\vy,\vz)\cdots\addi{0}(\vy,\vz)
\end{equation} (that is, $\addv$ is a multi-output circuit with $t+1$ output bits).
\end{definition}

It is worth noting that by \autoref{def:carryadd} we have (where the equality means that the polynomials are identical, though the circuits for them is different):
\begin{equation}\label{eq:car}
\cari{i}(\vy,\vz) = \begin{cases}
\bigvee_{r<i} \left(y_r \land z_r\land \bigwedge_{k=r+1}^{i-1} (y_k\lor z_k)\right),
               & i = 1,\ldots,t;\\
              0\,, & i=0\,.
\end{cases}
\end{equation}



Let $s$ be a bit, and denote by $\ve(s)$ the bit vector in which all bits are $s$ (that is, $\ve(s)=s\cdots s$) and where the length of the vector is understood from the context.
%
In the two's complement scheme inverting a positive number is a two-step process: first flip its bits (that is, XOR with the all-1 vector) and then add 1 to the result.
Hence, in what follows, to \emph{invert a negative} number, and extract its absolute value, we first subtract $1$ and then flip its bits:
\begin{definition}[absolute value operation \iabsv]\label{def:abs}
Let  $\vx$ be a $t$-bit vector representing an integer in two's complement. Let  $s$ be  its sign bit, and let $\vm=\ve(s)$ be the $t$-bit vector all of which bits are $s$. Define $\iabsv(\vx)$ as the multi-output circuit that outputs $t+1$ bits\footnotemark\ as follows (where $\oplus$ here is bit-wise {\rm XOR}): $$\iabsv(\vx):=\addv(\vx,\vm)\oplus\vm.$$
\end{definition}
\footnotetext{Note that since the largest (in absolute value) negative number that can be represented by a $t$-bit binary vector in the two's complement scheme is $2^{t-1}$, while the largest positive number that can be represented in such a way is only $2^{t-1}-1$, we need to store the absolute number of a $t$-bit integer in the two's complement scheme using $t+1$ bits.}


For the sake of the clarity of the proof, we compute the product  of two $t$-bit numbers in the two's complement notation
somewhat less efficiently than it is usually done:
we define the product of  nonnegative numbers in the standard way,
apply it to the absolute values of the numbers
and then apply the appropriate sign bit.
This way we get a slightly greater number of bits
than needed to keep the value. 


To define the multiplication of two $t$-bit integers  in the two's complement notation we first
define an unsigned multiplication operator \prdvp\ which is easy to implement.
It takes two non-negative integers (that is, their sign bit is zero,
and this assumption is required for the circuit to work correctly),
and performs the standard non-negative multiplication by performing
 $i=0,\dots,t-1$ iterations, where the $i$th iteration consists of multiplying the first integer
by the single $i$-th bit of the second integer,
and then padding this product by $i$ zeros to the right.

\begin{definition}[product of two nonnegative numbers in binary \prdvp]\label{def:prodp}
Let $\va$ and $\vb$ be two $t$-bit integers where the sign bit of both $\va,\vb$ is zero.
We define $t$ iterations $i=0,\ldots,t-1$;
the result of the $i$th iteration is defined as the $(t+i)$-length vector $\overline s_i=s_{i,t+i-1}s_{i,t+i-2}\cdots s_{i,0}$ where
\begin{align*}
s_{ij}&:=a_{j-i}\land b_i, &\text{ for $i \le j\le t-1+i$},\\
s_{ij}&:=0 &\text{for $0\le j<i$.}
\end{align*}
(Note that we  use the sign bits $a_{t-1},b_{t-1}$ in this process although we assume it is zero; this is done in order to preserve uniformity with other parts of the construction.)
The  product of the  two nonnegative $t$-bit numbers  is defined  as the sequential addition of all the results in all iterations:
\[
\prdvp(\va,\vb):=\addv\left(\overline s_{t-1},\addv\left(\overline s_{t-2},\ldots,\addv\left(\overline s_1,\overline s_0\right)\right)\ldots\right).
\]
The number of output bits of $\prdvp$ is formally $2t$ including the sign bit.
\end{definition}

\begin{definition}[product of two numbers in binary \prdv]\label{def:prod}
Let $\vy$ and $\vz$ be two $t$-bit integers in the two's complement notation.
Define the product of $\vy$ and $\vz$ by first multiplying the absolute values of the two numbers
and then applying the corresponding sign bit:
\[
\prdv(\vy,\vz):= \addv\left(\prdvp\left(\iabsv(\vy),\iabsv(\vz)\right)\oplus\vm,s\right),
\]
where $s=y_{t-1}\oplus z_{t-1}$ and $\vm=\ve(s)$, with $y_{t-1},z_{t-1}$ the sign bits of $\vy,\vz$ as bit vectors in the two's complement notation, respectively.

Note that the number of bits that $\prdv$ outputs is  $2t+3$:
given a $t$-bit number, its $\iabs$ is of size $t+1$ (including the zero sign bit),
the nonnegative product $\prdvp$ of $\iabs(\vx)$ and $\iabs(\vy)$ has size $2(t+1)$,
bitwise XOR does not change the length, and adding $s$ augments the result by one more bit.
\end{definition}

Note that given the bit vectors $\vx,\vy$ of length $t$, the size of the circuit for $\cari{i}(\vx,\vy)$ is $O(t)$ by \ref{eq:def-carryi}, for  $\addi{i}(\vx,\vy) $ it is $O(t)$ as well,
and for $\addv(\vx,\vy)$ it is still $O(t)$
because in \ref{eq:def-ADD} we can  \emph{re-use} $\cari i $. The size of $\iabsv(\vx)$  is also $O(t)$
(this is addition and linear-size bitwise XOR)
and finally  $\prdv(\vx,\vy)$ is of size  $O(t^2)$:
we perform an addition of $O(t)$ bit vectors of size $O(t)$ each. 



\subsection{Extracting Bits and the Main Binary Value Lemma}\label{sec:extracting-bits}
We are now ready to define the algebraic circuit  \bitv, in which $\biti i$ is the $i$th bit, that extracts the bit vector of the output of a given algebraic circuit (as a function of the input variables, where the variables are considered to range over 0-1).

\begin{definition}[the bit vector extraction operation \bitv]\label{def:bit}
Let $F$ be an algebraic circuit with $t$ its syntactic length.
Assume that $0\le i \le t-1$.
We define \biti{i}$(F)$ to be the circuit computing the $i$th bit of $F$ recursively as follows
(note that \biti{i} is a circuit, that is, in the induction step of the construction we may  \emph{re-use} the same nodes more than once).

\case 1 $F=x_j$ for an input $x_j$. Then, $\biti{0}(F):=x_j$, $\biti{1}(F):=0$ (in this case there are just two bits).

\case 2 $F=\alpha$, for $\alpha\in\Z$. Then, $\biti{i}$ is defined to be the $i$th bit of $\alpha$ in two's complement notation, with at most $t$ bits (i.e., $i<t$).

\case 3 $F=G+H$. Then $\bitv(F):=\addv(\bitv (G),\bitv(H))$, and $\biti i(F)$ is defined to be the $i$th bit of $\bitv(F)$.

\case 4 $F=G\cd H$. Then $\bitv(F):=\prdv(\bitv (G),\bitv(H))$, and $\biti i(F)$ is defined to be the $i$th bit of $\bitv(F)$.

Recall that in the latter two cases the shorter number is padded to match the length of the longer number by copying the sign bit before applying $\addv$ or $\prdv$.
\end{definition}

Note that both  $|\biti{i}(F)|$ and $|\bitv(F)|$ have size $O(t^2\cd|F|)$ (for $t$ the syntactic length of $F$). To understand this upper bound, observe that every node in the circuit for $\bitv(F)$  belongs to  either a sub-circuit computing the $i$th bit of $\addv(\vx,\vy)$ (i.e., $\addi i(\vx,\vy)$) or to a sub-circuit computing the $i$th bit of $\prdv(\vx,\vy)$, for some $0\le i\le t$ and some two vectors of bits $\vx,\vy$ that were already computed by the circuit (since this is a circuit, once the vectors $\vx,\vy$ were computed we can use their result  as many times as we like, without the need to compute them again). Hence, each addition gate in $F$ contributes $O(t)$  nodes to $\bitv(F)$ and each product gate in $F$ contributes $O(t^2)$ nodes to $\bitv(F)$.  This accounts for the size $O(t^2\cd|F|)$ for $\bitv(F)$ (as well as for $\biti{i}(F)$).

For technical reasons we need the following definition:
\begin{definition}[IPS sub-proof; multi-output IPS proofs]\label{def:sub-proof}
Let $C(\vx,\vy)$ be an IPS proof from a set of polynomial equations as assumptions $\overline {\mathcal F}$ of $p(\vx)$ (that is, $C(\vx, \overline {\mathcal F})=p(\vx)$ and $C(\vx,\overline 0)=0$), and suppose that  $C'(\vx,\vy)$ is a sub-circuit of $C(\vx,\vy)$ such that $C'(\vx,\vy)$ is an IPS proof of $g(\vx)$ (that is, $C'(\vx, \overline {\mathcal F})=g(\vx)$ and $C'(\vx,\overline 0)=0$).\footnote{Notice that not all sub-circuits of $C$ are IPS proofs; e.g., if they are polynomials that are not in the ideal generated by $\vy$, they are not sub-proofs.} Then, we say that \emph{$C'(\vx,\vy)$ is a sub-proof of $C(\vx,\vy)$}, and also (by slight abuse of terminology) that \emph{$g(\vx)$ is a sub-proof of the IPS proof $C$ of $p(\vx)$ from $\overline {\mathcal F}$}.

Furthermore, a \emph{multi-output} circuit $C(\vx,\vy)$ is said to be an IPS proof from assumptions \assmp, if each of its output gates computes an IPS proof.
\end{definition}

For example, let the assumptions be $\assmp=\{ x_2, (1+x_1x_2)\}$. The two-output circuit $C(\vx,\vx)$ defined as $(x_1\cd y_1, x_1\cd y_2)$, where $x_1$ is joined by the two sub-circuits $x_1\cd y_1$ and $x_1\cd y_2$, is an IPS proof having two sub-proofs: the first is a sub-proof of  $x_1\cd x_2$ from \assmp, and the second is a sub-proof of $x_1\cd(1+x_1x_2)$ from \assmp.

\begin{lemma}(main binary value lemma)\label{lem:main-binary-value-lemma}
Let $F(\vx)$ be an algebraic circuit over \Z\  in the variables $\vx=\{x_1,\dots,x_n\}$, and suppose that the syntactic length of $F$ is at most $t$. Then, there is an IPS proof of $$F=\val(\bitv(F))$$
of size $\poly(|F|,t)$ (there are no axioms for this IPS proof, except for the boolean axioms). Furthermore, if $F(\vx)$ is constant-free,   the $\poly(|F|,t)$-size  IPS proof is also constant-free.
%
\end{lemma}

\begin{proof}
We will proceed, essentially, by induction on the structure of $F$. For technical reasons (since we work with circuits of which sub-circuits can be re-used) we are going to state our induction hypothesis on an IPS proof that consists, as sub-proofs, of other IPS proofs.

More precisely, let $F_1,\dots,F_k$ be a set of sub-circuits of $F$. We denote by $\lambda(F_1,\dots,F_k)$ the size of the IPS proof we are to construct; this proof will  consist (as sub-proofs) of IPS proofs of $F_i=\val(\bitv(F_i))$, for all $i\in[k]$.
We let $\lambda(\emptyset):=0$.
We shall assume that at every step of the construction $F_1$ is of maximal size, namely there is no $F_i$ that has size bigger than $F_1$ (possibly there are other $F_i$'s with the same size). In this way, we make sure that $F_1$ is not a sub-circuit of any other $F_i$.
The IPS proof is to be constructed by  induction on $|F_1|$  so that in each step of the induction we deal with a single  sub-circuit $F_1$ of $F$, such that $|F_1|>1$. In  the base case of the induction we thus have $\lambda(F_1,\dots,F_k)$ such that all $F_i$'s have size 1, namely, they are all the variables and constant gates that appear in $F$.

 Note that since we treat the input to $\lambda$ as a \emph{set} we discard duplicate $F_i$'s from its input. For example, $\lambda(G,H)=\lambda(G)$ in case $G=H$. (This is convenient in what follows, because $F$ is a circuit and the IPS proof we construct is also a circuit, and hence can re-use multiple times the same IPS sub-proof; see below.)


We proceed by induction on $|F_1|$, the maximal size of a circuit in $F_1,\dots,F_k$, to show  the following: in case all $F_1,\dots,F_k$ are variables or constant nodes we show that
$$
\lambda(F_1,\dots,F_k)\le c_0 k t,$$ for some constant $c_0$.

In case  $F_1=G\circ H$, for $\circ\in\{+,\cd\}$, we  show that
$$
\lambda(F_1,\dots,F_k)\le \lambda(G,H,F_2,\dots,F_k)+(t\cd|F_1|)^{b},
$$
for some constants $b$ independent of $|F_1|$ and $t$. This recurrence relation immediately  implies that $\lambda(F)\le |F|\cd(t\cd|F|)^{b}$, which is polynomial in $|F|$ and $t$ (informally, every node in $F$ contributes the additive term $c_0t$ or $(t\cd|F|)^{b}$ to the recurrence).

\Base All $F_1,\dots,F_k$ are variables or constant nodes.
We construct a multi-output IPS proof $C(\vx,\vy)$, that consists of $k$ disjoint proofs of   $\val(\bitv(F_j))=F_j$, for $j\in[k]$.

\medskip
\case 1 $F_j=x_i$, for $i\in[n]$. Thus, the syntactic length of $F_j$ is 2 and by definition $\val(\bitv(x_i)):=\val(0x_i) := x_i-2^1\cd 0$ (the left equality is by definition of $\bitv$, and the right equality is by definition of $\val$). Hence,  $\val(\bitv(x_i))=x_i$ is a true polynomial identity and so  by \autoref{fact:zero-poly-ips-proof} we have an IPS proof of size precisely the size of the circuit for $x_i-2^1\cd 0-x_i$ which is at most, say, 20. 



\medskip
\case 2 $F_j=\alpha$, for $\alpha \in \Z$. Then, by \autoref{def:bit}, $\val(\bitv(\alpha)):=\sum_{i=0}^{t-2}2^i\alpha_i-2^{t-1}\cd\alpha_{t-1}$, where $\alpha_{t-1}\dots\alpha_0$ is the correct bit vector of $\alpha$\ in the two's complement notation (where  $t$ is the syntactic length of $F_j$).
Hence, $\val(\bitv(\alpha))=\alpha$ is a true polynomial identity of size at most $c_0t$, for some constant $c_0$. By \autoref{fact:zero-poly-ips-proof} we have an IPS proof of $\val(\bitv(\alpha))=\alpha$ of size at most $c_0t$.

Hence, the total size of all the proofs of $\val(\bitv(F_j))=F_j$, for $j\in[k]$, amounts to $\lambda(F_1,\dots,F_k)\le  c_0kt$, as required.

\induction We assume that the syntactic length of $F_1$ is $t.$ We show that, in case $F_1=G+H$,  $\lambda(F_1,\dots,F_k)\le \lambda(G,H,F_2,\dots,F_k)+c_1+(t\cd|F_1|)^{b'}$, for some constants $b'$ and $c_1$, and in case $F_1=G\cd H$ we show that $\lambda(F_1,\dots,F_k)\le \lambda(G,H,F_2,\dots,F_k)+(t\cd|F_1|)^{a}+(t\cd|F_1|)^{b'}$, for some constants $b'$ and $a$ independent of $t$ and $|F_1|$. Thus, choosing a big enough constant $b$, e.g., $b>10\cd\max(a,b')$, will conclude that $\lambda(F)\le |F|\cd(t\cd|F|)^{b}$, and hence will conclude the proof.
\medskip

\case 1 $F_1=G+H$, with $F_1$ of syntactic length $t$. Since the syntactic length of $F_1$ is $t$, the syntactic length of $\bitv(G), \bitv(H)$ is $t-1$ (after padding $\bitv(G), \bitv(H)$ to have the same size). We need to construct an IPS proof consisting of sub-proofs of $\val(\bitv(F_1))=F_1,\dots,\val(\bitv(F_k))=F_k$. 
By induction hypothesis we have an  IPS proof consisting of sub-proofs of $G+H=\val(\bitv(G)) +\val(\bitv(H))$ and $ F_i=\val(\bitv(F_i))$, for $ i=2,\dots,k$, of total size $\lambda(G,H,F_2,\dots,F_k)+c_1$, for some constant $c_1$ (the constant $c_1$  here is needed for the addition of the two proofs; see \autoref{fac:F-G+H-K} in which $c_1=1$). It thus suffices to  prove
$$\val(\bitv(G)) +\val(\bitv(H))=\val(\bitv(F_1))$$ with an IPS proof of size at most $(t\cd|F_1|)^{b'}$, for some  constant $b'$ independent of $t$. 


For simplicity of notation, let us denote the circuits for bits $\bitv(G)$  and $\bitv(H)$, by $\vy$ and $\vz$, respectively, and the syntactic length of $\vy,\vz$ by $r=t-1$.  We proceed slightly informally within IPS as follows (recall that polynomial identities of size $s$ have trivial IPS proofs of size $s$ by \autoref{fact:zero-poly-ips-proof}).
$$
\val(\vy) +\val(\vz) = \sum_{i=0}^{r-2} 2^i(y_i+z_i)-2^{r-1}(y_{r-1}+z_{r-1}).
$$
On the other hand we have (recall the padded bits $y_r=y_{r-1}$, $z_r=z_{r-1}$ in the definition of $\addv$)
\begin{align*}
\val(\bitv(F_1)) & =
    \val\left(
        \addi{0}\left(
                \vy,\vz
                \right)
                \ldots
                \addi{r}\left(
                \vy,\vz
                \right)
        \right)~~~~\text{(by definition of \bitv)}\\
& =
\sum_{i=0}^{r-1} 2^i
    \left(
        z_i\oplus y_i\oplus \cari{i}(\vy,\vz)
    \right)
    -   2^{r}(z_{r-1}\oplus y_{r-1}\oplus \cari{r}(\vy,\vz))\,\\ & ~~~~~~~~~~~~~~~~~~~~~~~~~~~~~~~~~~~~~~\text{(by definition of $\addi i$ and \val)}.
\end{align*}

Thus, to complete the case of addition, it remains to prove the following:

\begin{claim}
There is an IPS proof with size  at most $(r\cd|F_1|)^{b''}$, for a constant $b''$ (independent of $r$, and such that $b'$ will be chosen so that  $b'>b''$) of the equation
\begin{eqnarray*}
 &&\sum_{i=0}^{r-2} 2^i(y_i+z_i)-2^{r-1}(y_{r-1}+z_{r-1})\\
&&=  \sum_{i=0}^{r-1} 2^i
    \left(
        z_i\oplus y_i\oplus \cari{i}(\vy,\vz)
    \right)
    -   2^r(z_{r-1}\oplus y_{r-1}\oplus \cari{r}(\vy,\vz))\,.
\end{eqnarray*}
\end{claim}

\begin{proofclaim}
This is proved by induction on $r$ as follows.

\Base $r =2$. It is easy to see (or can be verified by hand) that in this case the two sides of the claim are equal for every $y_0,z_0,y_1,z_1\in\{0,1\}$.
Given that the number of bits in this case is constant, this is enough to conclude that there is an IPS proof of the above equation (using reasoning by boolean cases as in \autoref{prop:IPS-cases},\mar{TBC in last section} over a constant number of possible truth assignments for $y_0,z_0,y_1,z_1$) of size $(2\cd|F|)^{b''}$, for some constant $b''$.
\mar{It's unclear whether contraction here of cases works! But we may NOT need Fact 22 ! Because if the number of bits is constant we can simply prove this constant size statement in brute force.}\mar{Check the b' here!! We may need to fix this b', and specifically b at the end.}

\induction
To prove this step, notice that using the induction hypothesis we see that the equality we need to prove is
\[
 (z_{r-2}+y_{r-2})-(z_{r-1}+y_{r-1}) = z_{r-2}\oplus y_{r-2}\oplus \cari{r-1}(\vy,\vz) - z_{r-1}\oplus y_{r-1}\oplus \cari{r}(\vy,\vz)\,.
\]
Substituting the definition for $\cari{r-1}$ and $\cari{r}$,
we get a polynomial equation in five variables: $z_{r-2}$, $y_{r-2}$, $z_{r-1}$, $y_{r-1}$, and $C$,
where $C=\cari{r-2}(\vy,\vz)$.
Once it is verified by hand on $\{0,1\}$, we conclude that the circuit size of the proof is polynomial in the size of the circuits provided that these five ``variables'' are indeed boolean. Four of them are boolean by the hypothesis of the lemma, and the equation $C^2-C=0$ for the carry bit $C$ is also easy to derive.
Similarly to the above, we get an IPS proof of size at most  $(r\cd|F|)^{b''}$, for a constant $b'$.  \end{proofclaim}

This concludes Case 1 (i.e., addition) of the induction step of the proof of \autoref{lem:main-binary-value-lemma}.
\bigskip

\case 2 $F_1=G\cdot H,$ with $F_1$ of syntactic length $t$.
We need to construct an IPS proof consisting of sub-proofs of $\val(\bitv(F_1))=F_1,\dots,\val(\bitv(F_k))=F_k$, of size at most   $\lambda(G,H,F_2,\dots,F_k)+(t\cd|F_1|)^a+(t\cd|F_1|)^{b'}$, for constants $a,b'$ independent of $|F_1|$ and $t$. By induction hypothesis we have an  IPS proof consisting of sub-proofs of $G\cd H=\val(\bitv(G))\cd\val(\bitv(H))$
and $ F_i=\val(\bitv(F_i))$, for $ i=2,\dots,k$, of total size $\lambda(G,H,F_2,\dots,F_k)+|F_1|+c_2$, for some constant $c_2$ (the term $|F_1|+c_2$ here is needed for the product of the two proofs $G=\val(\bitv(G))$ and $H=\val(\bitv(H))$; see \autoref{fac:FxH-GxK}). It thus suffices to  prove
$$\val(\bitv(F_1))=\val(\bitv(G)) \cd\val(\bitv(H))$$ with an IPS proof of size at most $(t\cd|F_1|)^{b'}$, for a constant $b'$. Let $r$ denote the syntactic length of $G,H$. Since the syntactic length of $F_1$ is $t$ we have  $t=2r+3$.
%

In what follows, we use the notation from \autoref{def:prod}, namely, $\vy= \bitv(G)$ and $\vz=\bitv(H)$. 
We first prove two simple statements about $\iabsv$.
\begin{claim}\label{cla:neg} Let $\vx$ be a bit vector of length $r$ representing an integer in two's complement and let  $s$ be the sign bit of $\vx$. Then
$\val(\vx)=(1-2s)\cd \val\left(\iabsv(\vx)\right)$ has an IPS proof from the boolean axioms, of size at most  $r^c$, for some constant $c$ independent of $r$.
\end{claim}
\begin{proofclaim}
Recall that the size of $\iabsv(\vx)$ is $O(r)$.
We will apply (slightly informally) \autoref{prop:IPS-cases} for reasoning by boolean cases in IPS as follows. Consider the two cases for the sign bit $s$. In case $s=0$ the claim is not hard to check; we will show only the case $s=1$.

Recall that inverting a negative number via \iabsv\  is done by subtracting $1$ (which is the same as adding the all-one vector) and then inverting all the bits in the resulting vector. Let $\vy$ be a bit vector and $\vone$ be the all-one vector of the same length of $\vy$, then
\begin{equation}\label{eq:something}
\val(\vy \oplus \vone) =
\sum_{i=0}^{r-2} (1-y_i)2^i - (1-y_{r-1})2^{r-1} =
-1-\val(\vy).
\end{equation}
Using this, we have
\begin{align*}
(1-2s)\cd\val(\iabsv(\vx)) &= -1\cd\val(\addv(\vx,\vone)\oplus\vone)~~~\text{(by definition of \iabsv)}\\
&= -1\cd(-1-\val(\addv(\vx,\vone)))~~~ \text{(by \ref{eq:something} above)}\\
&= 1+\val(\addv(\vx,\vone)).
\end{align*}
By the addition case (Case 1 above) we can construct an IPS proof of $\val\left(\addv(\vx,\vone)\right)=\val(\vx)+\val(\vone)=\val(\vx)-1$ of size  at most $r^{b'}$, for some constant $b'$. This concludes the proof since we finally get $1+\val(\addv(\vx,\vone))=\val(\vx)$, where the whole proof is of size at most $r^c,$ for some constant $c$.
\end{proofclaim}
\mar{We need to mention that IPS proofs are closed under substitutions}
\begin{claim}[non-negativeness of $\iabsv$]\label{cla:abspos}
Let $\vx$ be a bit vector of length $r$ representing an integer in two's complement and let $s$ be the circuit computing the sign bit of $\iabsv(\vx)$ according to \autoref{def:abs}. Then $s=0$ has a polynomial-size \IPS\ proof (using only the boolean axioms).
\end{claim}
\begin{proofclaim}
We proceed as before by considering the two cases of the sign of $\vx$. The case of  positive sign is  easy to verify.
In the case of a negative sign we have
$\iabs(\vx)=\addv(\vx,\vone)\oplus\vone$,
where by the definition of \addv, $\vx$ is padded with an additional one bit $x_r=x_{r-1}=1$, and hence the sign bit of $\iabs(\vx)$ is computed as $\cari{r}(\vx,\vone)\oplus 1$ (note that $\addv$ has one more bit than $\vx$).
By (\autoref{eq:car}), $\cari{r}(\vx,\vone)$ is equal to (the arithmetization of) $\bigvee_{i<r} x_i$. Since $x_{r-1}=1$, we can prove in IPS by a simple substitution that the arithmetization of $\bigvee_{i<r} x_i$ is the constant 1, leading to $\cari{r}(\vx,\vone)\oplus 1=0$.
\end{proofclaim}

We consider then the case of the multiplication of nonnegative numbers.
\begin{claim}\label{cla:prdvp}Let $\vy,\vz$ be two bit vectors of length $r$ in the two's complement notation. Then,
$$
\val\left(\prdvp\left(\iabsv(\vy),\iabsv(\vz)\right)\right)=\val\left(\iabsv(\vy)\right)\cd\val\left(\iabsv(\vz)\right)
$$
has an \IPS\ derivation (from the boolean axioms) of size $r^{c}$, for a constant $c$ independent of $r$.
\end{claim}
\begin{proofclaim}
Let $\vy^+$ denote   $\iabsv(\vy)$ and $\vz^+$ denote $\iabsv(\vz)$, both of length $r+1$
(we know from \autoref{cla:abspos} that the sign bits $\vy^+_r,\vz^+_r$ of $\vy^+$ and $\vz^+$, respectively, are zero).
Recall \autoref{def:prod} of \prdv, in which we defined the vector $\overline s_i$ to be the result of multiplying the $i$th bit of $\overline z^+$, denoted $z_i^+$, with   $\vy^+$, and then padding it with $i$ zeros to the right.
First, we show that IPS can prove  that this multiplication step is correct, in the sense that IPS has an $O(r)$-size proof of:
\begin{equation}
\val(\overline s_i)=\val(\vy^+)\cdot z^+_i\cd 2^i\,,\label{eq:prodone}
\end{equation}
for every $i=0,\dots,r$.
Indeed, for every $i=0,\dots,r$, by definition of $\overline s_i$ we have the following polynomial identities:
\begin{align*}
\val(\overline s_i)&=\sum_{j=i}^{r+i-1} (y^+_{j-i}z^+_i) 2^j - (y^+_r\cd z^+_i) 2^{r+i} =
\left(\sum_{j=0}^{r-1} y^+_j2^j\right)\cd z^+_i\cd 2^i \\
& = \val(\vy^+)\cdot z^+_i\cd 2^i\,
\end{align*}
(we have used $y^+_r = z^+_r = 0$ here).


\smallskip
Second, based on the proof of the case of addition (Case 1 above), we can derive
\begin{align}
 \val(\addv(\overline s_{r},& \addv(\overline s_{r-1},\ldots,\addv(\overline  s_0,\overline s_1)\ldots))) \label{eq:prodshift}  \\
&  = \val(\overline s_{r})+\val(\addv(\overline s_{r-1},\ldots,\addv(\overline s_0,\overline s_1)\ldots))) \notag \\
& ~~~~~...\notag
\\
& =  \val(\overline s_{r})+\dots+\val(\overline s_2)+
\val(\addv(\overline s_0,\overline s_1)) \notag
\\
& = \sum_{i=0}^{r} \val(\overline s_i)\,.\label{eq:calculations-end}
\end{align}
Consider line \ref{eq:prodshift}: each \addv\ there contributes $O(r)$ gates. Thus, in total \ref{eq:prodshift} has a circuit of size  $O(r^2)$. 
Since line \ref{eq:prodshift} is of size $O(r^2)$, every step in which we use the addition case of the induction statement (Case 1), takes $r^{c'}$, for some constant $c'>2$ independent of $r$. Hence,  overall we obtain an IPS proof of the equality between \ref{eq:prodshift} and \ref{eq:calculations-end}, of size  $r^{b''}$, for some constant $b''$ independent of $r$.

Using (\autoref{eq:prodone}) and $z^+_r=0$ we conclude with an IPS proof that  \ref{eq:prodshift} above (which by \autoref{def:prod} is $\val\left(\prdvp(\vy^+,\vz^+)\right)$) equals
\[
\val(\vy^+)\cd\left(\sum_{i=0}^{r-1} z^+_i 2^i\right),
\]
which  in turn is equal to  $\val(\vy^+)\cd\val(\vz^+)$,
by definition of \val\ and the fact that $z^+_r=0$. This amounts to an IPS proof of total size $r^{c}$, for a constant $c$ independent of $r$.
%
%
\end{proofclaim}
\medskip

Finally, we arrive at the main case of multiplying two possibly negative integers written in the two's complement scheme, each  with bit vector of length $r$.
Let $s = y_{r-1}\oplus z_{r-1}$ and let $\vm=\ve(s)$ be a
vector of length $r$ in which every bit  is $s$.
Recall that
\[
\prdv(\vy,\vz)= \addv\left(\prdvp\left(\iabsv(\vy),\iabsv(\vz)\right)\oplus\vm,s\right).
\]

\begin{claim}\label{cla:prdv:sign}
$\val{\left(\prdv(\vy,\vz)\right)} = (1-2s)\cd \val{\left(\prdvp\left(\iabsv(\vy),\iabsv(\vz)\right)\right)}$
has an IPS derivation from the boolean axioms of size $r^{c}$, for some constant $c$ independent of $r$.
\end{claim}
\begin{proofclaim}
Consider the following two cases.

\case 1 $s=1$. Note that inverting all bits affects the value of a bit vector as follows: if $\vx$ is a length $k$ bit vector, then
\begin{equation}\label{eq:filling-in}
\val(\vx \oplus \ve(s)) =
\sum_{i=0}^{k-2} (1-x_i)2^i - (1-x_{k-1})2^{k-1} =
-1-\val(\vx).
\end{equation}
Hence, since $s=1$,
\begin{align*}
\val{\left(\prdv(\vy,\vz)\right)}& = \val\left(\addv\left(\prdvp(\iabsv(\vy),\iabsv(\vz))\oplus\vm,s\right)\right)\text{~by definition of \prdv} \\
& = \val\left(\prdvp(\iabsv(\vy),\iabsv(\vz))\oplus\vm\right)+ 1\\& \text{~~~~~~~~~~~~~~~~~~~~~~~~~~~~~~~~~~~~by Case 1 (addition) of induction statement}\\
& =-1- \val\left(\prdvp(\iabsv(\vy),\iabsv(\vz))\right)+ 1 \text{~~~~~~~~~~~~~~by \ref{eq:filling-in}}  \\
& = (1-2s)\cd\val\left(\prdvp(\iabsv(\vy),\iabsv(\vz))\right)\text{~~~~~~~~~since $s=1$}.
\end{align*}

\case 2  $s=0$. This is an easier case, in which  we show  $\val{\left(\prdv(\vy,\vz)\right)}=\val(\prdvp\left(\iabsv(\vy),\iabsv(\vz)\right))$, and so we are done by $s=0$. We omit the details.

Using reasoning by boolean cases in \IPS\ according to \autoref{prop:IPS-cases} we conclude the claim.
\end{proofclaim}

Taking together ~\autoref{cla:prdv:sign},~\autoref{cla:prdvp} and ~\autoref{cla:neg} (for $\vy$ and for $\vz$ of length $t$)
we get the desired equality for the product case, where  $s=y_{t-1}\oplus z_{t-1}$:
\begin{eqnarray*}
&&\val\left(\prdv(\vy,\vz)\right)\\
&&=(1-2s)\cd \val\left(\prdvp\left(\iabsv(\vy),\iabsv(\vz)\right)\right)\\
&&=(1-2s) \cd\val\left(\iabsv(\vy)\right)\cd\val\left(\iabsv(\vz)\right)\\
&&=(1-2y_{r-1})\cd \val\left(\iabsv(\vy)\right) \cdot \left(1-2z_{r-1}\right)\cd \val\left(\iabsv(\vz)\right)\\
&&=\val(\vy)\cdot\val(\vz)\,,
\end{eqnarray*}
where the penultimate equation stems from the polynomial identity  $(1-2y_{t-1})\cd(1-2z_{t-1})=1-2(y_{t-1}\oplus z_{t-1})$.
\smallskip

This concludes the proof of the first part of  \autoref{lem:main-binary-value-lemma}. For the second part, assuming that $F(\vx)$ is constant-free, the proof is identical, noting simply that in the IPS proof we constructed above all coefficients are at most exponential in $n$, and thus by the upper bound $\tau(m)\leq O(\log m )$ for every $m\in\N$, we get a constant-free IPS proof of size $\poly(n)$.
\end{proof}


\section{Algebraic versus Semi-Algebraic Proof Systems}
\label{sec:Algebraic-versus-Semi-Algebraic-Proof-Systems}
Here we show that IPS simulates CPS over \Q\ assuming the existence of small IPS refutations for the generalized binary value principle (and the binary value principle for the case of \Z). Under reasonable conditions we show that in fact IPS is polynomially equivalent to CPS assuming short IPS refutations of the (generalized) binary value principle, hence bridging the gap between algebraic and semi-algebraic proof systems in the regime of very strong proof systems. We work with the \emph{boolean versions} of both CPS and IPS, meaning that the boolean axioms are present.

Moreover, we  demonstrate two kinds of conditional simulations: a (standard) polynomial simulation for the language of unsatisfiable sets \assmp\ of polynomial \emph{equations}, and in \autoref{sec:effective-simulation} an  \emph{effective} \emph{simulation} (in the sense of Pitassi-Santhanam \cite{PitassiSanthanam10}) for the language of unsatisfiable sets containing both equations \assmp\ and inequalities \ineqassmp\ over \Z; similar reasoning works over \Q). Note that we cannot hope to show a (standard) simulation of CPS by IPS for the language containing both polynomial equalities and polynomial inequalities, because  inequalities are not expressible directly as polynomial equations in IPS; hence, for the sake of the second kind of simulation we first translate \ineqassmp\ to bit representation and only then simulate the CPS proof, yielding an effective simulation.

We now prove the simulation for constant-free proofs over \Q, and in \ref{sec:effective-simulation} we will prove the effective simulation (over \Z, which implies the same result over \Q). 


Recall that \ipsq\ and \cpsq\ stand for IPS and CPS proofs over \Q, respectively, and  that by \ref{prop:Q-circuit-to-Z-circuit}, given a constant-free circuit $C$ over \Q\ we can turn it into a constant-free circuit $C'$ over \Z\ computing $M\cd\widehat C$, for some nonzero integer $M$, with $|C'|\le 4|C|$ and $\tau(M)\le 4|C|$.

\begin{definition}[syntactic length of a circuit over \Q]\label{def:syntactic-length-over-Q} The syntactic length of a circuit $C$ over \Q\ is the syntactic length of the corresponding circuit $C'$ over \Z\ constructed from $C$  in \ref{prop:Q-circuit-to-Z-circuit}.
\end{definition}


The main technical theorem of this section is the following:
\begin{theorem}[conditional simulation of constant-free boolean \cpsq\ by constant-free boolean \ipsq]\label{thm:ips_sim_cps}
 Let \assmp\ denote a system of polynomial equations over \Q\ written as constant-free circuits $\{F_i(\vx)=0\}_{i\in I}$ and let $C(\vx,\assmp)=-1$ be a constant-free \cpsq\ refutation of \assmp\ where $C(\vx,\assmp)$ is of size $s$ and syntactic length $t$ (as  in \autoref{def:syntactic-length-over-Q}).\footnotemark ~Assume that the binary value principle \bvptm\ has a
\hirsch{Changed formulation below: $M$ (and thus $r$) depends both on $s$ and $t$. The old one is commented out.
We can simplify it by requiring $r$ to be a polynomial so that we talk about $\poly(s,t)$ only.}
 size $\le r$ constant-free \ipsq\ refutation, for every given positive integer $M$ with $\tau(M)=O(s)$.
 Then, there is a constant-free  \ipsq\ refutation of \assmp\  with size $\poly(s,t,r)$. \iddo{I\ simplified it to r, while also adding the bound on tau(M). I thought about your formulation: it seems that the formulation now is simpler and accurate: although r may depend on other parameters, we consider in this statement *fixed* r, s, t (also I'm not sure r depends on tau(M) or s in fact). And then the bound poly(s,t,r) means that there is a polynomial independent of r,s,t that bounds the size of the constructed proof. Also, we do not know that r is a monotone function of s,t! Maybe for some increased parameters to r, r is actually decreased!  So I don't see the advantage of the complicated bound $\poly\big(s,t,r\left(\poly\left(s,t\right),t\right)\big)$ (also it seems slightly strange to nest the poly's). Commented the previous thm below.}\end{theorem}

\footnotetext{We need to consider also the size of the CPS refutation \emph{after} the substitution  of \assmp\ for the placeholder variables, that is, $C(\vx,\assmp)$, because of the slightly peculiar nature of IPS (similar to CPS) in which the size of a refutation does not include directly  the size of the assumptions it refutes.}

\begin{remark}
\begin{enumerate}
\item
By inspection of the proof of \autoref{thm:ips_sim_cps} one can see that the degree of the simulating IPS refutation can be
exponential in the size of the resulting circuit (clearly, the degree cannot be larger than that).

\item Assuming that indeed propositional IPS simulates propositional CPS, by \autoref{prop:CPS-from-equtional-CNF-to-inequalities-CNF} propositional IPS also simulates any propositional CPS (or \ps/\sos) refutation of CNF formulas given as \emph{inequalities}. This is because if propositional CPS has a short refutation for a CNF given as inequalities (\autoref{def:semi-algebraic-transl-CNF}) then from  \autoref{prop:CPS-from-equtional-CNF-to-inequalities-CNF}, propositional CPS also has a short refutation of the CNF given as equations (\autoref{def:algebraic-transl-CNF}).
\end{enumerate}
\end{remark}


Since the simulation of CPS by IPS in \autoref{thm:ips_sim_cps} depends on the syntactic length $t$ of the simulated CPS proof, if we aim to achieve a (polynomial) simulation we need to bound the syntactic length of the CPS proofs to be at most polynomial in the proof size. We denote this restricted proof system by \cpszs\ and \cpsqs. In other words, a family $\set{\pi_i}_{i=1}^\infty$ of \cpsz\ (resp.~\cpsq) proofs is said to be \emph{a family of \cpszs\ (resp.~\cpsqs) proofs} if there is a constant $c$ such that for every $i\in\N$, the syntactic length of $\pi_i$ is at most $|\pi_i|^c$.
In other words, the maximal value (over \bits-assignments to the variables) of every gate in \cpszs\ proof-sequence   $\set{\pi_i}_{i=1}^\infty$ is  bounded from above by $2^{|\pi_i|^{O(1)}}$.

It is important to note that most known examples of short semi-algebraic proofs of propositional formulas \emph{have polynomial syntactic length},
as the multiplication of arbitrary inequalities is not used, and multiplying by $x$
or by $1-x$ for a variable $x$ increases the syntactic length additively.
The use of division by scalars (for example, in the LS proof of PHP)
can increase the syntactic length in \ref{prop:Q-circuit-to-Z-circuit};
however, as those scalars have at most exponential (actually, polynomial) values,
the transformation from
rational numbers to integers can bring at most $(exp(\poly(n)))^{\textrm{proof-size}}$ factor,
thus a polynomial number of bits.\hirsch{pls check the reasoning!}


\medskip

Recall the terminology in \ref{sec:Shub-Smale-lower-bound}: a refutation in \ipsz\ means a proof of $M$ for some nonzero integer $M$. Further, we say that \ipsz\ simulates \cpsq\ if a size-$s$ \cpsq\ proof of $p$ from assumptions \assmp\ over \Z\ implies that there is a $\poly(s)$-size \ipsz\ proof of $M\cd p$ from \assmp, for some nonzero  integer $M$.

The binary value principle  thus characterizes  exactly the apparent advantage  CPS has over IPS, in the following sense:


\begin{corollary}[BVP characterizes the strength of boolean CPS]\label{cor:capture-semi-algebraic-proofs}
In what follows, IPS and CPS stand for boolean IPS and boolean CPS, respectively, where both are proof systems for refuting unsatisfiable sets of polynomial equalities (not necessarily CNF formulas).
\begin{enumerate}
\item \label{it:cond-sim-corol-z}
Constant-free \ipsz\ simulates  constant-free \cpszs\ iff constant-free \ipsz\ admits $\poly(t)$-size refutations of \bvpt.

\item \label{it:cond-sim-corol-q-cf}
Constant-free \ipsq\ simulates constant-free \cpsqs\ iff for every positive integer $M$,  constant-free \ipsq\ admits $\poly(t,\tau(M))$-size refutations of \bvptm.
\end{enumerate}

\end{corollary}

%
%

\begin{proof}
We show the proof of \autoref{it:cond-sim-corol-q-cf} (which includes all the ideas for the other  case).

\nind ($\Leftarrow$) Assume that for every positive integer $M$  constant-free \ipsq\ admits $\poly(t,\tau(M))$-size refutations of \bvptm. Then specifically for  $\tau(M)=O(s)$ there is a $\poly(t,s)$ upper bound on the size of constant-free \ipsq\ refutations of \bvptm. By \autoref{thm:ips_sim_cps} if there exists a syntactic-length $t$ constant-free \cpsqs\ refutation of \assmp\ then there exists a constant-free IPS refutation of \assmp\ with size $\poly(s,t,r)=\poly(s)$, because $t =\poly(s)$ by assumption and $r=\poly(s,t)$.

\nind ($\Rightarrow$) This follows from the \cpsz\  upper bound on \bvpn\  demonstrated in \autoref{prop:CPS-proof-of-BVP}. More precisely, it suffices to show that given a positive integer $M$ there are constant-free \cpsqs\ refutations of \bvptm\ having $\poly(t,\tau(M))$-size. Using the notation as in the proof of  \autoref{prop:CPS-proof-of-BVP}, we claim that the conic circuit $\frac{1}{M}\cd \left(\sum_{i=1}^t 2^{i-1}\cd y_i\right)  + \frac{1}{M}\cd y_{t+1} $ serves as such a refutation. Indeed, this conic circuit is easily written as an $O(t\cd \log t+\tau(M))$-size  \emph{constant-free} circuit. This is because $\tau(2^{i-1})=\log (i-1)$, for every $i=1,\dots, t$, and $1/M$ is clearly of size $2+\tau(M)$. That this conic circuit is a refutation of \bvptm\ follows immediately from the definition (see the proof of \autoref{prop:CPS-proof-of-BVP}).





The proof of \autoref{it:cond-sim-corol-z} is similar and we omit the details.
\end{proof}

 By \ref{thm:CPS-sim-IPS} CPS simulates IPS, hence when considering IPS proofs of which the syntactic-length grows polynomially in the size of the proofs, \ref{cor:capture-semi-algebraic-proofs} characterizes when IPS and CPS \emph{can simulate each other}. More precisely, similar to \cpszs\ and \cpsqs\ we denote by \ipszs\ and \ipsqs\ the proof systems consisting  of IPS proofs in which the syntactic length grows polynomial in the size of proofs (over \Z\ and \Q, respectively). In other words, a family $\set{\pi_i}_{i=1}^\infty$ of \ipsz\ (resp.~\ipsq) proofs is said to be \emph{a family of \ipszs\ (resp.~\ipsqs) proofs} if there is a constant $c$ such that for every $i\in\N$, the syntactic length of $\pi_i$ is at most $|\pi_i|^c$.

\begin{corollary}[Conditional equivalence of strong algebraic and semi-algebraic proofs]\label{cor:IPS-cond-equiv-CPS}
In what follows, IPS and CPS stand for boolean IPS and boolean CPS, respectively, where both are proof systems for refuting unsatisfiable sets of polynomial equalities (not necessarily CNF formulas).
\begin{enumerate}
\item 
Constant-free \ipszs\ is polynomially equivalent to constant-free \cpszs\ iff constant-free \ipszs\ admits $\poly(t)$-size refutations of \bvpt.

\item 
Constant-free \ipsqs\ is polynomially equivalent to  constant-free \cpsqs\ iff for every positive integer $M$ constant-free \ipsqs\ admits $\poly(t,\tau(M))$-size refutations of \bvptm.
\end{enumerate}
\end{corollary}

\begin{remark}

The results above in \autoref{thm:ips_sim_cps}, \autoref{cor:capture-semi-algebraic-proofs} and \autoref{cor:IPS-cond-equiv-CPS} hold trivially also in the \emph{unit-cost model}, where we consider the size of coefficient in the ring or field to be $1$. More precisely, if we replace the term ``constant-free proof'' with the term ``proof'' the results still hold.
This is because we limit the syntactic length of the original CPS circuit, and the size of circuit families of polynomial  syntactic length in the unit-cost model is smaller or equal than their size in the constant-free model. And if a family of constant-free circuits (proofs) $C_n$ simulates a family of constant-free circuits \emph{with a polynomial syntactic length} $D_n$, then the corresponding circuit family $C_n'$ in the unit-cost model  also simulates the corresponding circuit family $D_n'$ in the unit-cost model (because $|D_n|\le \poly(|D_n'|)$).
\end{remark}

\subsection{Proof of \autoref{thm:ips_sim_cps}}


We need to show that there is an \ipsz\ refutation of \assmp. We first translate the setting to the integers, since this will allow us to use the main binary value \autoref{lem:main-binary-value-lemma} which is stated for \Z, as follows: we take the \cpsq\ refutation, turn it into a  \cpsz\ refutation without increasing the size too much (the syntactic length stays the same by definition), and then simulate this refutation  in \ipsz, that is, construct an \ipsz\ proof from \assmp\ of a nonzero integer $M$. Dividing this \ipsz\ refutation by $M$ we get the desired \ipsq\ refutation of \assmp. We formalize this conversion in the following proposition:

\begin{proposition}[going from constant-free  \cpsq\ to constant-free \cpsz]\label{prop:Q-CPS-to-Z-CPS}
Let \assmp\ denote a system of polynomial equations over \Q\ written as  constant-free circuits $\{F_i(\vx)=0\}_{i\in I}$ and let $C(\vx,\assmp)=-1$
be a constant \cpsq\ refutation of \assmp,
where $C(\vx,\assmp)$ has size $s$ and syntactic length $t$.
Then, there exists a set of polynomial equations over \Z\ denoted  $\assmp^\star=\{F^\star_i(\vx)=0\}_{i\in I}$,
where $F^\star_i(\vx)=M_i\cd F_i(\vx)$ for some nonnegative $M_i\in\Z$, for all $i\in I$, and a constant-free
\cpsz\ proof $C^\star(\vx,\vy)$ from $\assmp^\star$ of $M\cd (-1)$,
for some nonzero $M\in\Z$, such that $C^\star(\vx,\assmp^\star)$ has both size and syntactic length $\poly(s,t)$.
\end{proposition}

\begin{proof}
The proof is identical to the proof of \ref{prop:Q-circuit-to-Z-circuit} (cf.~\ref{cor:from-Q-IPS-to-Z-IPS}).
 Specifically, given a constant-free circuit $D$ over \Q\ the Induction Statement in the proof of  \ref{prop:Q-circuit-to-Z-circuit}  shows that there exists a size at most $4|D|$ constant-free circuit $D^\star$ over \Z\ that computes $M\cd \widehat D$ for some nonzero integer $M$.
Accordingly, we turn \assmp\ into $\assmp^\star$ and $C(\vz,\vy)$ into $C^\star(\vz,\vy)$ in this way. By definition of syntactic length for circuits over \Q\ the syntactic length of $C^\star(\vz,\vy)$ is $t$.
\end{proof}

By \ref{prop:Q-CPS-to-Z-CPS}, to prove \ref{thm:ips_sim_cps} we can assume without loss of generality
that  \assmp\ is a system of constant-free-circuit equations \emph{over \Z} and
that $C(\vx,\assmp)=-M$
is a constant-free \cpsz\ refutation, where $C(\vx,\assmp)$ is of size $s$ and syntactic length $t$. Thus, \emph{from now on we assume that all constant-free circuits and proofs are over \Z.}

\bigskip

Given a multi-output circuit of size $s$, with $m$ output gates, each computing the circuit $H_i$ (for $i\in[m]$), we assume that an algebraic circuit for $\sum_{j=1}^m H_j^2$ is defined to be a sum of  $m$ summands, written as a binary tree of logarithmic in $m$ depth, in which each summand $H_j^2$ is defined as the circuit whose output  is a product gate with its two children connected to the output gate of $H_j$, and where different $H_j$'s can have common nodes (so that the size of the circuit computing  $\sum_{j=1}^m H_j^2$ is linear in $s$).
\begin{lemma}[sign bit of sum of squares is zero]
\label{lem:sign-bit-of-sos}
Consider the circuit $H=\sum_{j\in J} H_j^2$, and let $\biti{t}(H)$ be the sign bit of $\bitv(H)$.
Then $\biti{t}(H)=0$ has a polynomial-size IPS proof (using only the boolean axioms).
\end{lemma}
\begin{proof}
Informally, the idea is to prove the desired equation using only the structure  of sign bits of additions and squares appearing in  top layers only (the layers close to the output gate) of $H$, without looking at the individual structure of the circuits $H_j$'s.

First, we show that the sum of two nonnegative numbers is nonnegative, that is, if a pair of circuits have sign bits that are zero then the sign bit of their addition is also zero, and in symbols:
$$
\biti{t}(F)=0,~\biti{t}(G)=0\ipsprf{\poly(|F|,|G|)} \biti{t+1}(F+G)=0\,,
$$
where the sign bit of $F,G$ is bit $t$ and the sign bit of $F+G$ is  bit $t+1$.

Let $y:=\biti{t}(F)$ and $z:=\biti{t}(G)$, then by \autoref{def:carryadd} the sign bit of $F+G$ is computed as $y\oplus z\oplus\cari{t+1}(\bitv(F),\bitv(G))$, because we have padded $F$ and $G$ by their sign bits $y,z$, respectively, before the addition. Given that $y=0$ and $z=0$ by assumption, we need to prove that $\cari{t+1}(\bitv(F),\bitv(G))=0$. By \autoref{def:carryadd} $\cari{t+1}(\bitv(F),\bitv(G))= (y\land z) \lor ((y\lor z) \land \cdots))$. Since the arithmetic expressions (according to \autoref{def:arithmetization}) for $y\land z$ and $y\lor z$ can be easily proved to be zero (from $y=0$, $z=0$), and the same holds for $0\land \cdots$, we conclude that the sign bit of $F+G$ is zero.
\medskip


To prove that  each of the squares $H_j^2$ are nonnegative, one needs to consider the two cases of the sign bit $x$ of $H_j$ and infer that the sign bit of the square is zero in both cases
using \autoref{prop:IPS-cases}.

Recall that
\[
\prdv(\vy,\vz):= \addv\left(\prdvp\left(\iabsv(\vy),\iabsv(\vz)\right)\oplus\vm,s\right),
\]
where $s=y_{t'}\oplus z_{t'}$ and $\vm=\ve(s)$, with $y_{t'},z_{t'}$ the sign bits of $\vy,\vz$ as bit vectors in the two's complement notation, respectively.

In both cases of the sign of $H_j$, we have $s = 0$ and $\vm=\vzero$ as $y$ and $z$ are equal in our case.
Everything else is identical in both cases:
the sign bit of $\prdvp$ is always zero,
because $\prdvp$ is a consecutive sum of nonnegative numbers
(the sign of each of those numbers $s_i$ from the definition of $\prdvp$
is obtained by $\land$-ing a single bit with the sign of $\iabsv$,
the latter being zero by \autoref{cla:abspos}),
and we have already proved that the sum of nonnegative numbers is nonnegative.
Applying the latter fact once again,
we conclude that the sign of $H_j^2$ is zero in both cases.
\end{proof}

We will need the following simple lemma:
\begin{lemma}\label{lem:F-squared-F-equals-zero}
Let $G$ be an algebraic circuit which is an arithmetization of a boolean circuit $g$ (\autoref{def:arithmetization}). Then, IPS has a polynomial-size in $|G|$ derivation of $G^2-G$ from the boolean axioms.
\end{lemma}

\begin{proof}
This is proved by  induction on $|G|$; see for example \cite[Lemma 4]{GH03}, where this is proved for polynomial calculus over algebraic formulas denoted $\mathcal{F}{\text{-}} PC$.
\end{proof}

Since for any circuit $F$, $\biti{i}(F)$ is  the result of an arithmetization of a boolean circuit we have:
\begin{corollary}\label{cor:bitisbit}
Let $F$ be a circuit, then IPS has a polynomial-size derivation of $\biti{i}(F)^2-\biti{i}(F)$ from the boolean axioms.
\end{corollary}


\begin{lemma}[sign bit of literals is zero]\label{lem:sign-bit-of-lit-is-zero} 
Let $x_i$ be a variable and let $\biti{1}(x_i)$ and $\biti{1}(1-x_i)$ be the sign bits of of $\bitv(x_i)$ and $\bitv(1-x_i)$, respectively. Then $\bitv(x_i)=0$ and $\bitv(1-x_i)=0$ have constant-size IPS proofs (using only the boolean axioms). \mar{Make sure to define the position of the sign bit to starts at 0}
\end{lemma}
\begin{proof}
Observe that indeed the syntactic length of $x_i$ and $1-x_i$ is 2. Now, $\biti{1}(x_i)=0$ holds by definition, since we define $\bitv(x_i)=0x_i$ (\autoref{def:bit}). For $\biti{1}(1-x_i)=0$, this follows by considering the two options $x_i=0$ and $x_i=1$ (where the size of the proofs is constant, since the statement itself is of constant size, namely, it involves only a single variable and a two-bit vector).\mar{maybe make more precise?}
\end{proof}

\begin{lemma}[sign bits of axioms are zero]\label{lem:sign-bits-of-axioms-is-zero}
Given there are polynomial-size IPS proofs of
$\biti{t}(f(\vx))=0$ from $f(\vx)=0$ and the boolean axioms, where $t+1$ is the syntactic length of $f(\vx)$.
\end{lemma}
\begin{proof}
By  \autoref{lem:main-binary-value-lemma} we know that $\val(\bitv (f))=f$, and hence by assumption   $\val(\bitv (f))=0$.
We need to show that under   $\val(\bitv (f))=0$ we can infer $\biti{t}(f)=0$ with a short IPS proof. Note that this inference is a substitution instance of the following inference:
\begin{equation}\label{eq:infer-substitution-instance-of-BVP}
\sum_{i=1}^{t} 2^{i-1}x_i - 2^t x_{t+1}=0 ~\ipsprf{} ~x_{t+1}=0,
\end{equation}
where we substitute $\biti{i-1}(f)$ for $x_i$ ($i=1,\dots,t+1$). By \autoref{fac:IPS-closed-substitution-instance}, IPS proofs are closed under substitution instances (together with the fact that the corresponding substitution instances of the boolean axioms \ba\ are also provable in IPS by  \autoref{cor:bitisbit}) and so it remains to show that under the assumption that  BVP has polynomial-size IPS refutations, \ref{eq:infer-substitution-instance-of-BVP} holds.

To prove \ref{eq:infer-substitution-instance-of-BVP} it suffices to show that the assumptions $x_{t+1}=1$ and $ \sum_{i=1}^{t} 2^{i-1}x_i - 2^t x_{t+1}=0 $ can be refuted with a polynomial-size IPS refutation. \mar{check why it is enough---IPS reasoning proof by contradiction...}

Assuming $x_{t+1}=1$,  $ \sum_{i=1}^{t} 2^{i-1}x_i - 2^t x_{t+1}=0$ becomes  $ \sum_{i=1}^{t} 2^{i-1}x_i - 2^t =0$, and so it remains to show the following:
\begin{claim*}
Under the assumption that \bvpn\ has $\poly(n)$-size IPS refutations, there are polynomial-size IPS refutations of $\sum_{i=1}^{t} 2^{i-1}x_i - 2^t =0$. \end{claim*}

\begin{proofclaim}
Our assumption that there are polynomial-size IPS refutations of BVP$_{t+1}$ $\sum_{i=1}^{t+1} 2^{i-1}x_i + 1=0$, implies that there are short refutation also of its substitution instance $\sum_{i=1}^{t+1} 2^{i-1}(1-y_i) + 1=0$ (again, by \autoref{fac:IPS-closed-substitution-instance} and the fact that the substitution instance of the boolean axioms \ba, are easily provable when substituting $1-y_i$ for $x_i$'s; cf.~Lemma \autoref{lem:F-squared-F-equals-zero}). But $\sum_{i=1}^{t+1} 2^{i-1}(1-y_i) + 1=-(\sum_{i=1}^{t+1}2^{i-1}y_i-2^{t})=0$.
\end{proofclaim}
\end{proof}
%

Up to now, we have shown that for each algebraic circuit in the ``base'' of the conic circuit $C(\vx,\vy)$ comprising a CPS refutation  (namely, the sub-circuits that substitute the placeholder variables \vy, as well as the \vx\ variables themselves), the sign bit can be proved  to be zero in IPS. The following lemma shows that under these assumptions IPS can prove that the conic circuit $C(\vx,\vy)$ itself has a zero sign bit (for simplicity we use only \vx\ variables in the circuit $C(\vx)$ below).

\begin{lemma}[conic circuits preserve zero sign bits]\label{lem:conic-circuit-preserves-zero-sign-bits}
Let $C(\vx)$ be a conic circuit over \Z\ in the variables $\vx=\{x_1,\dots,x_n\}$, let $\overline H:=\{H_i(\vx)\}_{i\ = 1}^n$ be $n$ circuits and suppose that $t$ is the syntactic length of $C(\overline H)$. Then, there is a polynomial-size in $|C(\overline H)|$ IPS proof that the sign bit of $C(\overline H)$ is $0$, that is, of $\biti{t}(C(\overline H))=0$, from the assumptions $\biti{t_i-1}(H_i(\vx))=0$, for all $i\in[n]$, where $t_i$ is the syntactic length of $H_i(\vx)$. \end{lemma}

\begin{proof}
The proof is by induction on the size of $C$. Note that any conic circuit $C$ is one of the following: (1)  a variable $x_i$, (2) a non-negative constant $\alpha$, (3) a square of some (not-necessarily conic) circuit, that is, $C=G^2$, or (4) an  addition $C=G+H$ or product $C=G\cd H$ of two conic circuits $G,H$. Therefore, the base cases of our induction will be the first three cases (1)-(3), and the  induction steps will be the latter case (4).

\Base

\case 1
 $C=x_i$. Then from the assumption that $\biti{t_j-1}(H_j(\vx))=0$
 for all $j\in[n]$, we have that $C(\overline H)=H_i(\vx)$, and so we are done.

\case 2
$C=\alpha$, for a non-negative constant $\alpha$. Then by \autoref{def:bit} $\bitv(\alpha)$  is the actual bits of $\alpha$ in two's complement. Since $\alpha$ is non-negative $\biti {t-1}(C(\overline H))=\biti t(\alpha)=0$, for $t$ the syntactic length of $\alpha$.

%

\case 3
$C=G^2$ for some not-necessarily conic circuit $G$. This case follows from  \autoref{lem:sign-bit-of-sos}.

\induction

\case 1 $C=G+H$. This follows from the claim that the sign bit of the addition of non negative numbers is 0, as shown in the proof of  \autoref{lem:sign-bit-of-sos}.

\case 2 $C=G\cd H$. This follows from the claim that the sign bit of the product of two non-negative integers
is non-negative.
\iddo{TBC! Similar to  \autoref{lem:sign-bit-of-sos}}\hirsch{Should we rather claim and prove both
these cases before \autoref{lem:sign-bit-of-sos}?}\iddo{don't know.}
\end{proof}

We are now ready to conclude the main theorem  of this section.

\begin{proof}[Proof of \autoref{thm:ips_sim_cps}]
By assumption, $C(\vx,\vy)$ is a conic circuit constituting a CPS refutation of \assmp.
We  assume that $\{f_i(\vx)\}_{i\in I}$ can be computed by a sequence of circuits $\{F_i(\vx)\}_{i\in I}$ such that $\sum_{i\in I}|F_i(\vx)| = u$. Hence, by the definition of CPS, we set \ineqassmp\ to be the set of circuits  that consists of $F_i(\vx)$ and $-F_i(\vx)$, for all $i\in I$, as well as the boolean axioms translation $x_i^2-x_i$ and $-x_i^2+x_i$, for all $i\in[n]$, and $x_i$ and $1-x_i$, for all $i\in [n]$.  We thus have $C(\vx,\ineqassmp)=-M$ as a polynomial identity.

Since $C$ is a conic circuit, and the sign bits of all variables $\vx$ and all circuits in \ineqassmp\ can be proved in polynomial size (in $u$) to be 0, by  \autoref{lem:sign-bit-of-lit-is-zero} and \autoref{lem:sign-bits-of-axioms-is-zero}, respectively, we know from \autoref{lem:conic-circuit-preserves-zero-sign-bits} that the sign bit of $C(\vx,\ineqassmp)$ is 0 as well. Since $C(\vx,\ineqassmp)=-M$ is a polynomial identity, by \autoref{fact:zero-poly-ips-proof} $C(\vx,\ineqassmp)+M$ has an IPS proof of size equal to the size of the circuit $C(\vx,\ineqassmp)+M$ itself.
We now proceed to use the short refutation of the BVP to get a short IPS refutation from the fact that the sign  bit of $C(\vx,\ineqassmp)$ is 0 and $C(\vx,\ineqassmp)+M=0$. The following claim suffices for this purpose:

\begin{claim}\label{cla:finally-IPS-refute-C-non-negative-and-C=-1}
Assume that  \bvpnm\ has $\poly(n,\tau(M))$-size IPS refutations
Let $F(\vx)$ be a circuit of syntactic length $t$ and size $s$, such that IPS has a $\poly(s,t)$-size proof of $\biti {t-1}(F(\vx))=0$ (where $\biti {t-1}(F(\vx))$ is the sign bit of $F(\vx)$). Then there is a $\poly(s,t,\tau(M))$ refutation of $F(\vx)+M=0$.
\end{claim}

\begin{proofclaim}
%
%
The size of the circuit $F(\vx)+M$ is $s+\tau(M)+1$. By ~\autoref{lem:main-binary-value-lemma}, $\val(\bitv(F(\vx)+M))=F(\vx)+M=0$
has a polynomial size in $s+\tau(M)+1$ IPS proof from the boolean axioms. By the proof of  \autoref{lem:main-binary-value-lemma} we also have a polynomial-size in $s$ and $\tau(M)$
IPS proof of $$
\val\left(\bitv(F(\vx))\right)+M=0,
$$
namely, a proof of
\[
M+\sum_{i=0}^{t-2} 2^i\cd w_i\,-2^{t-1}\cd w_{t-1}=0\,,\qquad
\textrm{where $w_i:=\biti{i}(F(\vx))$.}
\]
By assumption, $w_{t-1} =0$ has a polynomial-size IPS proof, where $w_{t-1}$ is the sign bit of $F(\vx)$. This leads to\begin{equation}\label{eq:final-substitution-of-BVP}
M+\sum_{i=0}^{t-1} 2^i\cd w_i=0.
\end{equation}

Notice that \ref{eq:final-substitution-of-BVP} is the binary value principle in which variables $x_i$ for $i=1,\dots,t$, are replaced by the circuits  $\biti{i-1}(F(\vx)),$ denoted $w_i$. We assumed that the binary value principle has polynomial-size (in $t$ and $\tau(M)$) refutations (using only the boolean axioms as assumptions). Since IPS proofs are closed under substitutions of variables by circuits (\autoref{fac:IPS-closed-substitution-instance}), there is a $\poly(t,|\bitv(F)|,\tau(M))$-size IPS refutation of \ref{eq:final-substitution-of-BVP} from the substitution instances of the boolean axioms $w_i^2- w_i=0$, for $i=0,\dots,t-1$. Since  for every $i=0,\dots,t-1$, $w_i^2-w_i=0$ has a short IPS proof by  \autoref{cor:bitisbit}, and since $|\bitv(F)|=\poly(t,|F|)$, we conclude that there exists a $\poly(s,t,\tau(M))$-size IPS refutation as desired.
\end{proofclaim}\mar{polish}
\end{proof}

\subsection{Effective Simulation of CPS Refutations with Inequalities}\label{sec:effective-simulation}

We now turn to conditional effective simulation of CPS as a refutation system for \emph{both equalities} \emph{and inequalities} by IPS. Effective simulation means that we are allowed to non-trivially translate the input equalities and inequalities before refuting them in IPS, as long as the translation procedure is polynomial-time and preserves unsatisfiability \cite{PitassiSanthanam10}. Similar to the case of conditional simulation, it is enough to consider only the case of CPS and IPS proofs over \Z\ to conclude it also for \Q. We show  here the case of non-constant-free boolean IPS and boolean CPS proofs over \Z. The case over \Q\ and the cases  of constant-free proofs over \Z\ and \Q\ are similar.\hirsch{Should we be more precise as $M$ would appear then?} \iddo{don't know} \hirsch{Can we just talk about IPS$^*$ and CPS$^*$ as before?}\iddo{Guess so, if you think it's necessary.}

Note that since the construction of the circuit $\biti{i}(\cd)$ (\autoref{sec:extracting-bits}) is mechanical and uniform, there is a straightforward deterministic (uniform) polynomial-time algorithm that receives a set of polynomial inequalities $\ineqassmp=\{H_j(\vx)\ge 0 \}_{j}$ over \Z\ written as algebraic circuits (with coefficients written in binary) and outputs the polynomial equations, written as algebraic circuits, expressing that the sign bit of each  $H_j(\vx)$ is 0 (hence, expressing the inequalities \ineqassmp). This translation of inequalities to polynomial equalities serves as our translation from \ineqassmp\ to the language of polynomial equations that is refutable in IPS. Given an inequality $H_j(\vx)\ge 0$ we denote by $\tr{H_j(\vx)\ge 0}$ this translation; accordingly, we let  $\tr{\ineqassmp} = \{\tr{H_j(\vx)\ge 0}\;:\;H_j(\vx)\in\ineqassmp\}$.

\begin{theorem}[conditional effective simulation of boolean CPS by boolean IPS]\label{thm:effective-sim:IPS-CPS}
Assume that the generalized binary value principle \bvptm\ has  $\poly(t,\tau(M))$-size boolean IPS refutations for every positive integer $M$. Let \assmp\ denote a system of polynomial equations and let \ineqassmp\ denote a system of polynomial inequalities written as circuits $\{H_j(\vx)\ge 0\}_{j\in J}$ (including all the equations in \assmp\ written as inequalities as described in \autoref{def:cps}). Let $C(\vx,\ineqassmp)=-1$ be a  CPS refutation of \assmp\ and \ineqassmp\ where $C(\vx,\ineqassmp)$ has size $s$ and syntactic length $t$. Then, there is a boolean IPS refutation of $\tr{\ineqassmp}$ with size $\poly(s,t)$.\footnote{Equivalently, we can also show that there is a size $\poly(s,t)$ IPS refutation of $\assmp$ and $\tr{\ineqassmp\setminus\assmp}$. But for simplicity we assume that the equalities \assmp\ are also translated via $\tr{\cd}$.}
\end{theorem}

\begin{proof}
This is identical to the proof of \autoref{thm:ips_sim_cps}, only that we do not need to prove separately that the axioms in \ineqassmp\ have all bit-vector representation in which the sign bit is 0, since here this is given to us as an assumption.
\end{proof}


\appendix
\section*{Appendix}
\section{Basic Reasoning in IPS}\label{sec:basic-IPS-reas}
Here we develop basic efficient reasoning in IPS. This is helpful for \autoref{sec:extracting-bits}.

First we show that polynomial identities are proved for free in IPS:
\begin{fact}\label{fact:zero-poly-ips-proof}
If $F(\vx)$ is a circuit in the variables $\vx$ over the field \F\ that computes the zero polynomial, then there is an IPS proof of $F(\vx)=0$ of size $|F|$.
\end{fact}

\begin{proof}[Proof of fact] The IPS proof of $F(\vx)=0$ is simply $C(\vx,\vz):=F(\vx)$ (note that we do not need to use the boolean axioms nor any other axioms in this case). Observe that both conditions 1 and 2 for IPS hold in this case (\autoref{def:IPS}).
\end{proof}

\begin{fact}
Let $F,G,H$ be circuits and $\mathcal F$ be a collection of polynomial equations such that $C:\mathcal F\ipsprf {s_0} F=G$ and $C':\mathcal F\ipsprf {s_1} G=H$. Then, $(C+C') :\mathcal F \ipsprf {s_0+s_1+1} F=H$.
\end{fact}
\begin{proof}[Proof of fact] $C(\vx,\overline {\mathcal F},\overline x ^2-\overline x)+C'(\vx,\overline {\mathcal F},\overline x ^2-\overline x)=F-G+G-H$. 
\end{proof}

\begin{fact}\label{fac:F-G+H-K}
Let $F,G$ be circuits and $\overline {\mathcal F}$ be a collection of polynomial equations such that $C:\overline {\mathcal F}\ipsprf {s_0} F=G$ and $C':\overline {\mathcal F}\ipsprf {s_1} H=K$. Then, $(C+C') :\overline {\mathcal F} \ipsprf {s_0+s_1+1} F+H=G+K$.
\end{fact}
\begin{proof}[Proof of fact] $C(\vx,\overline {\mathcal F},\overline x ^2-\overline x)+C'(\vx,\overline {\mathcal F},\overline x ^2-\overline x)=F-G+H-K$. 
\end{proof}

\begin{fact}\label{fac:FxH-GxK}
Let $F,G$ be circuits and $\overline {\mathcal F}$ be a collection of polynomial equations such that $C:\overline {\mathcal F}\ipsprf {s_0} F=G$ and $C':\overline {\mathcal F}\ipsprf {s_1} H=K$. Assume that there is a circuit with two output gates, of size $s$, with one output gate computing $H$ and the other output gate computing $G$. Then, $\overline {\mathcal F} \ipsprf {s_0+s_1+s+5} F\cd H=G\cd K$.
\end{fact}
\begin{proof}[Proof of fact] Observe that $C(\vx,\overline {\mathcal F},\overline x ^2-\overline x)\cd H+C'(\vx,\overline {\mathcal F},\overline x ^2-\overline x)\cd G =F\cd H-G\cd H + H\cd G-K\cd G= F\cd H-G\cd K$. Hence, the desired proof is the \emph{circuit} $C(\vx,\vy,\vz)\cd H(\vx)+C'(\vx,\vy,\vz)\cd G(\vx)$, which by assumption that there is a circuit of size $s$ computing both $H,G$,  is  at most $s_0+s_1+s+5$ (here, $H,G$ can have common nodes).
\end{proof}

We now wish to show that basic reasoning by \emph{boolean} cases is efficiently attainable in IPS. Specifically, we are going to show that if for a given constant many variables (or even boolean valued polynomials) $V$, for every choice of a fixed (partial) boolean assignment to the variables $V$ a polynomial equation is derivable, then it is derivable regardless (namely, derivable from the boolean axioms alone) in polynomial-size.

\begin{proposition}[proof by boolean cases in IPS]\label{prop:IPS-cases}
Let \F\ be a field. Let $V=\{H_{i}(\vx)\}_{i\in I}$ be a set of circuits  with $|V|=r$, and $\overline {\mathcal F}$ be a collection of polynomial equations such that $\{H_i^2(\vx)-H_i(\vx)=0\}_{i\in I}\subseteq \overline {\mathcal F}$. Assume that for every fixed assignment $\overline \alpha\in\bits^r$ we have $\overline {\mathcal F},\{H_i(\vx)=\alpha_i\}_{i\in I}\ipsprf{s} f(\vx)=0$, then $\overline {\mathcal F}\ipsprf{c^r\cd s} f(\vx)=0$, for some constant $c$ independent of $r$.
\end{proposition}
\mar{Take out $\vx$ in H and f!!}

\begin{proof}
We proceed by induction on $r$.
\Base $r=0$. In this case we assume that $\overline {\mathcal F}\ipsprf{s} f(\vx)=0$ and we wish to show that $\overline {\mathcal F}\ipsprf{c^r\cd s} f(\vx)=0$, for some  constant $c$, which is immediate since $r=0$.
\Induction $r>0$. We assume that  for every fixed assignment $\overline \alpha\in\bits^{r}$ we have $\overline {\mathcal F},\{H_i=\alpha_i\}_{i\in I}\ipsprf{s} f(\vx)=0$, and we wish to show that  $\overline {\mathcal F}\ipsprf{c^r\cd s} f(\vx)=0$, for some constant $c$ independent of $r$.

By our assumption above we know that for every fixed assignment $\overline \alpha\in\bits^{r-1}$ we have:
\begin{gather}
~~~~~~~~\overline {\mathcal F},H_1(\vx)=0, \{H_i(\vx)=\alpha_i\}_{i\in {I\setminus 1}}\ipsprf{s} f(\vx)=0, \text{~~~and}\label{eq:one}
\\
 \overline {\mathcal F},H_1(\vx)=1, \{H_i(\vx)=\alpha_i\}_{i\in {I\setminus 1}}\ipsprf{s} f(\vx)=0.\label{eq:two}
\end{gather}
From \ref{eq:one} and \ref{eq:two}, by induction hypothesis we have for some constant $c$ independent of $r$:
\begin{gather}
~~~~~~~~H_1(\vx)=0, \overline {\mathcal F}\ipsprf{c^{r-1}\cd s} f(\vx)=0,\text{~~~and}\label{eq:three}
\\
 H_1(\vx)=1, \overline {\mathcal F}\ipsprf{c^{r-1}\cd s} f(\vx)=0.\label{eq:four}
\end{gather}
It thus remains to prove the following claim:
\begin{claim}\label{cla:combine-x1=0-and-x1=1}
Under the above assumptions \ref{eq:three} and \ref{eq:four}, we have $\overline {\mathcal F}\ipsprf{c^{r}\cd s} f(\vx)=0$.
\end{claim}
\begin{proofclaim}
By \ref{eq:three} and \ref{eq:four} we have two IPS proofs $C(\vx,\vy,\vz)$ and $C'(\vx,\vy,\vz)$ such that  $C(\vx,\overline {\mathcal F}, H_1(\vx),\ba)=f(\vx)$ and $C'(\vx,\overline {\mathcal F}, 1-H_1(\vx),\ba)=f(\vx)$ (note indeed that $\overline {\mathcal F}, H_1(\vx)$ and $\ba$ are the axioms in the former case, and similarly for the latter case,  where now $1-H_1(\vx)$ replaces the axiom $H_1(\vx)$) each of size $c^{r-1}\cd s$.

By the definition of IPS  $C(\vx,\vy,\vz), C'(\vx,\vy,\vz)$ both compute polynomials that are in the ideal generated by $\vy,\vz$. This means that there are some polynomials $Q_i,P_i,G,M,L_i,K_i$, such that:
\begin{gather*}
\widehat C(\vx,\overline {\mathcal F},H_1(\vx),\ba)=
\sum_i Q_i\cd F_i + \sum_i L_i\cd(x_i^2-x_i)+G\cd H_1(\vx) = f(\vx)\text{~~ and~~}\\ \widehat C'(\vx,\overline {\mathcal F},1-H_1(\vx),\ba)=
\sum_i P_i\cd F_i + \sum_i K_i\cd(x_i^2-x_i)+ M\cd (1-H_1(\vx)) = f(\vx)
\end{gather*}
(here, $\overline {\mathcal F},H_1(\vx)$ is substituted for $\vy$ in the first equation, and $\overline {\mathcal F},1-H_1(\vx)$ is substituted for $\vy$ in the second equation).

Hence, we can multiply these two true polynomial identities by $(1-H_1(\vx))$ and $H_1(\vx)$, respectively, to get the following polynomial identities:
\begin{multline*}
(1-H_1(\vx))\cd \widehat C(\vx,\overline {\mathcal F},H_1(\vx),\ba)=\\
(1-H_1(\vx))\cd\sum_i Q_i\cd F_i + (1-H_1(\vx))\cd\sum_i L_i\cd(x_i^2-x_i)+G\cd H_1(\vx)\cd(1-H_1(\vx)) = (1-H_1(\vx))\cd f(\vx)
\end{multline*}
and
\begin{multline*}
H_1(\vx)\cd \widehat C'(\vx,\overline {\mathcal F},H_1(\vx),\ba)=
H_1(\vx)\cd\sum_i P_i\cd F_i + H_1(\vx)\cd\sum_i K_i\cd(x_i^2-x_i)+ H\cd H_1(\vx)\cd(1-x_1) \\ = H_1(\vx)\cd f(\vx).
\end{multline*}
Each of these two polynomial identities is an     IPS proof from the  assumptions  $\mathcal F=\{F_i\}_i$, the boolean axioms, and the assumption  $H_1(\vx)\cd(1-H_1(\vx))\in\overline  {\mathcal F}$ (more formally, $ (1-H_1(\vx))\cd C$ and $H_1(\vx)\cd C'$ are the \emph{circuits }that constitute these pair of IPS proofs).
Adding these two IPS proofs (note that the addition of two IPS proofs from a set of assumptions is  still an IPS proof  from that set of assumptions) we obtain the desired IPS proof of $f(\vx)$, with size $2\cd c^{r-1}\cd s+c_0\le c^{r}\cd s$, for a large enough constant $c$ independent of $r$.
\end{proofclaim}
This concludes the proof of the proposition.
\end{proof}

\autoref{prop:IPS-cases} allows us to reason by cases in IPS. For example, assume that we know that either $H_i(\vx)=0$ or $H_i(\vx)=1$; namely that we have the assumption $H_i(\vx)\cd(H_i(\vx)-1)=0$. Then, we can reason by cases as follows: if we can prove from $H_i(\vx)=0$ that $A$, with a polynomial-size proof, and from $H_i(\vx)=1$ that $B$, with a polynomial-size proof, then using \autoref{prop:IPS-cases} we have a polynomial-size proof that $A\cd B=0$ from $H_i(\vx)\cd(H_i(\vx)-1)=0$.

As an immediate corollary of \autoref{prop:IPS-cases} we get the same proposition with $H_i(\vx)$'s substituted for variables:

\begin{corollary}\label{cor:IPS-cases}
Let \F\ be a field. Let $V=\{x_{i}\}_{i\in I}$ be a set of variables with $|V|=r$, and $\overline {\mathcal F}$ be a collection of polynomial equations. Assume that for every fixed assignment $\overline \alpha\in\bits^r$ to the variables in $V$ we have $\overline {\mathcal F},\{x_i=\alpha_i\}_{i\in I}\ipsprf{s} f(\vx)=0$, then~ $\overline {\mathcal F}\ipsprf{c^r\cd s} f(\vx)=0$, for some constant $c$ independent of $r$.
\end{corollary}

\begin{fact}[IPS proofs are closed under substitutions]\label{fac:IPS-closed-substitution-instance}
Let $C(\vx,\vy,\vz)$ be an IPS proof of $f(\vx)$ from the assumptions $\{F_i(\vx)\}_{i=1}^m$, and let $\overline H=\{H_i(\vx)\}_{i=1}^n$ be a set of algebraic circuits. Then, $C(\overline H/\vx,\vy,\vz)$ is an IPS proof of $f(\overline H/\vx)$ from $\{F_i(\overline H/\vx)\}_{i=1}^m$, where $\overline H/\vx$ stands for the substitution of $x_i$ by $H_i(\vx)$, for all $i\in[n]$.
\end{fact}

The proof of \autoref{fac:IPS-closed-substitution-instance} is immediate.




\section*{Acknowledgement}
We wish to thank Michael Forbes, Dima Itsykson, Toni Pitassi and Dima Sokolov for useful discussions at various stages of this work.


\small
\bibliographystyle{plain}
\bibliography{PrfCmplx-Bakoma}
\normalsize




\end{document}

%% file: mfmacros.tex
        \usepackage{stmaryrd}
        \usepackage{tikz}
        \usepackage{tocbibind}
        \usepackage{version}
        \usepackage{xparse}
        \usepackage{xspace}
        \usepackage{xstring} 

        \usepackage{algorithm}

        \makeatletter
        \global\let\tikz@ensure@dollar@catcode=\relax
        \makeatother

        \SetSymbolFont{stmry}{bold}{U}{stmry}{m}{n}

        \PackageWarning{miforbes}{overfullrule enabled}
        \newcommand{\ignore}[1]{}

        \newcommand*{\subproofname}{Sub-Proof:}

        \newcommand{\ps}[1]{{\llb #1\rrb}}

        \PackageWarning{miforbes}{get rid of matrix commands}

        \renewcommand{\vec}[1]{\overline{#1}}

        \newcommand{\cG}{\mathcal{G}}

        \newcommand{\cvG}{{\vec{\cG}}}

        \makeatletter
        \newcommand{\va}{{\vec{a}}\@ifnextchar{^}{\!\:}{}}
        \newcommand{\vc}{{\vec{c}}\@ifnextchar{^}{\!\:}{}}
        \newcommand{\vd}{{\vec{d}}\@ifnextchar{^}{\!\:}{}}
        \newcommand{\ve}{{\vec{e}}\@ifnextchar{^}{\!\:}{}}
        
        \newcommand{\vg}{{\vec{g}}\@ifnextchar{^}{\!\:}{}}
        \newcommand{\vh}{{\vec{h}}\@ifnextchar{^}{\!\:}{}}
        \newcommand{\vi}{{\vec{i}}\@ifnextchar{^}{\!\:}{}}
        \newcommand{\vj}{{\vec{j}}\@ifnextchar{^}{\!\:}{}}
        \newcommand{\vk}{{\vec{k}}\@ifnextchar{^}{\!\:}{}}
        \newcommand{\vl}{{\vec{\ell}}\@ifnextchar{^}{\!\:}{}}
        \newcommand{\vm}{{\vec{m}}\@ifnextchar{^}{\!\:}{}}
        \newcommand{\vn}{{\vec{n}}\@ifnextchar{^}{\!\:}{}}
        \newcommand{\vo}{{\vec{o}}\@ifnextchar{^}{\!\:}{}}
        \newcommand{\vp}{{\vec{p}}\@ifnextchar{^}{\!\:}{}}
        \newcommand{\vq}{{\vec{q}}\@ifnextchar{^}{\!\:}{}}
        \newcommand{\vr}{{\vec{r}}\@ifnextchar{^}{\!\:}{}}
        \newcommand{\vs}{{\vec{s}}\@ifnextchar{^}{\!\:}{}}
        \newcommand{\vt}{{\vec{t}}\@ifnextchar{^}{\!\:}{}}
        \newcommand{\vu}{{\vec{u}}\@ifnextchar{^}{\!\:}{}}
        \newcommand{\vv}{{\vec{v}}\@ifnextchar{^}{\!\:}{}}
        \newcommand{\vw}{{\vec{w}}\@ifnextchar{^}{\!\:}{}}
        \newcommand{\vy}{{\vec{y}}\@ifnextchar{^}{\!\:}{}}
        \newcommand{\vz}{{\vec{z}}\@ifnextchar{^}{\!\:}{}}

        \newcommand{\vA}{{\vec{A}}\@ifnextchar{^}{\!\:}{}}
        \newcommand{\vB}{{\vec{B}}\@ifnextchar{^}{\!\:}{}}
        \newcommand{\vC}{{\vec{C}}\@ifnextchar{^}{\!\:}{}}
        \newcommand{\vD}{{\vec{D}}\@ifnextchar{^}{\!\:}{}}
        \newcommand{\vE}{{\vec{E}}\@ifnextchar{^}{\!\:}{}}
        \newcommand{\vF}{{\vec{F}}\@ifnextchar{^}{\!\:}{}}
        \newcommand{\vG}{{\vec{G}}\@ifnextchar{^}{\!\:}{}}
        \newcommand{\vH}{{\vec{H}}\@ifnextchar{^}{\!\:}{}}
        \newcommand{\vI}{{\vec{I}}\@ifnextchar{^}{\!\:}{}}
        \newcommand{\vJ}{{\vec{J}}\@ifnextchar{^}{\!\:}{}}
        \newcommand{\vK}{{\vec{K}}\@ifnextchar{^}{\!\:}{}}
        \newcommand{\vL}{{\vec{L}}\@ifnextchar{^}{\!\:}{}}
        \newcommand{\vM}{{\vec{M}}\@ifnextchar{^}{\!\:}{}}
        \newcommand{\vN}{{\vec{N}}\@ifnextchar{^}{\!\:}{}}
        \newcommand{\vO}{{\vec{O}}\@ifnextchar{^}{\!\:}{}}
        \newcommand{\vP}{{\vec{P}}\@ifnextchar{^}{\!\:}{}}
        \newcommand{\vQ}{{\vec{Q}}\@ifnextchar{^}{\!\:}{}}
        \newcommand{\vR}{{\vec{R}}\@ifnextchar{^}{\!\:}{}}
        \newcommand{\vS}{{\vec{S}}\@ifnextchar{^}{\!\:}{}}
        \newcommand{\vT}{{\vec{T}}\@ifnextchar{^}{\!\:}{}}
        \newcommand{\vU}{{\vec{U}}\@ifnextchar{^}{\!\:}{}}
        \newcommand{\vV}{{\vec{V}}\@ifnextchar{^}{\!\:}{}}
        \newcommand{\vW}{{\vec{W}}\@ifnextchar{^}{\!\:}{}}
        \newcommand{\vY}{{\vec{Y}}\@ifnextchar{^}{\!\:}{}}
        \newcommand{\vX}{{\vec{X}}\@ifnextchar{^}{}{}}          
        \newcommand{\vZ}{{\vec{Z}}\@ifnextchar{^}{\!\:}{}}
        \makeatother

        \newcommand{\vnz}{{\vec{0}}}

%% file: bibmacros.tex
\usepackage{nth}
\usepackage{intcalc}
\usepackage{etoolbox}
\usepackage{xstring}

\newcommand{\ECCCurl}[2]{http://eccc.hpi-web.de/report/\ifnumcomp{#1}{>}{93}{19}{20}#1/#2/}

\newcommand{\shortECCC}[2]{\texttt{\href{http://eccc.hpi-web.de/report/\ifnumcomp{#1}{>}{93}{19}{20}#1/#2/}{eccc:TR#1-#2}}}

\newcommand{\parseECCC}[1]{
\StrSubstitute{#1}{TR}{}[\tmpstring]%
\IfSubStr{\tmpstring}{/}{ 
\StrBefore{\tmpstring}{/}[\ecccyear]%
\StrBehind{\tmpstring}{/}[\ecccreport]%
}{
\StrBefore{\tmpstring}{-}[\ecccyear]%
\StrBehind{\tmpstring}{-}[\ecccreport]%
}%
\shortECCC{\ecccyear}{\ecccreport}}